\documentclass[11pt]{article}
\usepackage[utf8]{inputenc}
\usepackage[english]{babel}
\usepackage{graphicx} 
\usepackage{amsmath}
\usepackage{amssymb}
\usepackage{graphicx}
\usepackage{float}
\usepackage{caption}
\usepackage{amsthm}
\usepackage{xcolor}
\usepackage{tabu}
\usepackage{fullpage}
\usepackage{hyperref}
\hypersetup{
    colorlinks,
    linkcolor={red!50!black},
    citecolor={blue!50!black},
    urlcolor={blue!80!black}
}
\usepackage{cleveref}
\usepackage{xspace}
\usepackage{thm-restate}

\usepackage[sortcites,maxbibnames=99]{biblatex}
\renewbibmacro*{doi+eprint+url}{
    \printfield{doi}
    \newunit\newblock
    \iftoggle{bbx:eprint}{
        \usebibmacro{eprint}
    }{}
    \newunit\newblock
    \iffieldundef{doi}{
        \usebibmacro{url+urldate}}
        {}
    }

\theoremstyle{plain}
\newtheorem{theorem}{Theorem}
\newtheorem{lemma}[theorem]{Lemma}

\newtheorem{corollary}[theorem]{Corollary}

\theoremstyle{definition}

\newcommand{\TSAT}{\textsc{3SAT}\xspace}
\newcommand{\PTSAT}{\textsc{Planar-3SAT}\xspace}
\newcommand{\LPTSAT}{\textsc{Layered-Planar-3SAT}\xspace}
\newcommand{\MLPTSAT}{\textsc{Monotone-Layered-Planar-3SAT}\xspace}
\newcommand{\MPTSAT}{\textsc{Monotone-Planar-3SAT}\xspace}
\newcommand{\CTSAT}{\textsc{Clover-3SAT}\xspace}
\newcommand{\MCTSAT}{\textsc{Monotone-Clover-3SAT}\xspace}
\newcommand{\USPack}{\textsc{$2\times 2$-Square-Packing}\xspace}
\newcommand{\USCover}{\textsc{Small-Cover}\xspace}
\newcommand{\USPartition}{\textsc{Small-Partitioning}\xspace}

\title{Hardness of Packing, Covering and Partitioning Simple Polygons with Unit Squares}
\author{Mikkel Abrahamsen \and Jack Stade}
\date{April 2024}

\begin{document}

\pagenumbering{roman}
\maketitle

\begin{abstract}
We show that packing axis-aligned unit squares into a simple polygon $P$ is NP-hard, even when $P$ is an orthogonal and orthogonally convex polygon with half-integer coordinates.
It has been known since the early 80s that packing unit squares into a polygon with holes is NP-hard~[Fowler, Paterson, Tanimoto, Inf. Process. Lett., 1981], but the version without holes was conjectured to be polynomial-time solvable more than two decades ago~[Baur and Fekete, Algorithmica, 2001].

Our reduction relies on a new way of reducing from \textsc{Planar-3SAT}. Interestingly, our geometric realization of a planar formula is non-planar.
Vertices become rows and edges become columns, with crossings being allowed.
The planarity ensures that all endpoints of rows and columns are incident to the outer face of the resulting drawing.
We can then construct a polygon following the outer face that realizes all the logic of the formula geometrically, without the need of any holes.

This new reduction technique proves to be general enough to also show hardness of two natural covering and partitioning problems, even when the input polygon is simple.
We say that a polygon $Q$ is \emph{small} if $Q$ is contained in a unit square.
We prove that it is NP-hard to find a minimum number of small polygons whose union is $P$ (covering) and to find a minimum number of pairwise interior-disjoint small polygons whose union is $P$ (partitioning), when $P$ is an orthogonal simple polygon with half-integer coordinates.
This is the first partitioning problem known to be NP-hard for polygons without holes, with the usual objective of minimizing the number of pieces.
\end{abstract}

\tableofcontents

\newpage

\clearpage

\clearpage
\pagenumbering{arabic}

\setcounter{page}{1}

\section{Introduction}

Packing is a large field in computational geometry, operations research and pure mathematics concerned with arranging certain geometric shapes without overlap in a space efficient way.
The importance of the area is emphasized by numerous industrial settings where packing problems appear, such as in shipping, manufacturing, VLSI design and clothing production.

One of the simplest packing problems is to decide if $k$ axis-aligned unit squares can be placed in a given polygon $P$ without overlap.
In this paper, we shall be mainly concerned with the equivalent problem
\USPack:
Given a polygon $P$ and an integer $k$, decide if $P$ has room for $k$ axis-aligned squares of size $2\times 2$.
Focusing on $2\times 2$ squares makes it more convenient to state our results and explain our constructions.

It has been known for more than four decades that \USPack is NP-hard if $P$ can have holes.
This was shown in 1981 independently by Fowler, Paterson and Tanimoto~\cite{DBLP:journals/ipl/FowlerPT81} and by Berman, Leighton and Snyder~\cite{berman1981optimal}, later expanded in the paper~\cite{DBLP:journals/jal/BermanJLSS90}.
The reduction in~\cite{berman1981optimal,DBLP:journals/jal/BermanJLSS90} made use of the at the time recent discovery that \textsc{Planar-3SAT} is NP-hard~\cite{DBLP:journals/siamcomp/Lichtenstein82}, and the authors thus didn't have to develop a ``crossing gadget''.
The reduction in~\cite{DBLP:journals/ipl/FowlerPT81} reduced directly from \textsc{3SAT}, using crossing gadgets.
The reductions are otherwise quite similar in how they get from a \TSAT\ instance $\Phi$ to a polygon:
The edges of $\Phi$ are turned into corridors that can be packed in two optimal ways, corresponding to the values of a binary variable.
Each clause is realized as a carefully designed ``room'', where one more square can fit if the packing in one of the connected corridors corresponds to a value that makes the clause satisfied. 

These works constitute the first published NP-hardness proofs of problems where the input is a polygon that we are aware of.\footnote{The only preceding work with a result of this type seems to be a manuscript by Masek cited by Garey and Johnson~\cite[p.~232]{garey1979computers}, proving NP-hardness of the problem of describing an orthogonal polygon as a union of a minimum number of rectangles.
The manuscript was never published and is now apparently lost.}
The technique of ``building a polygon'' on top of (some version of) \TSAT\ has since been used to show hardness of a great variety of problems, for instance the Art Gallery Problem and other covering problems~\cite{DBLP:journals/tit/ORourkeS83}, as well as problems concerned with triangulations~\cite{DBLP:journals/jacm/MulzerR08,DBLP:journals/corr/LubiwM17,DBLP:conf/approx/FeketeKKMS11}, partitions~\cite{DBLP:journals/ijcga/FeketeM01,DBLP:conf/icalp/Lingas82}, tool paths for milling~\cite{DBLP:journals/algorithmica/ArkinHS00,DBLP:journals/comgeo/ArkinFM00}, Voronoi games~\cite{DBLP:journals/comgeo/FeketeM05}, facility location~\cite{DBLP:journals/ior/FeketeMB05}, separation of point sets~\cite{DBLP:conf/compgeom/DemaineEHILMOW04}, and motion planning~\cite{DBLP:conf/cccg/KirkpatrickKP11}.
However, the method has the downside that it necessarily leads to a polygon with holes, since each bounded face of the plane embedding of $\Phi$ will lead to a hole in the resulting polygon.

While the complexity of \USPack for polygons with holes was settled early, the complexity remained unknown in the case of \emph{simple} polygons, i.e., polygons without holes.
This is Problem 56 in The Open Problems Project~\cite{topp}.
Baur and Fekete~\cite{DBLP:journals/algorithmica/BaurF01} conjectured in 2001 that for any fixed integer $s>1$, there is a polynomial-time algorithm to pack a maximum number of $s\times s$ squares in a simple \emph{grid polygon}, i.e., an orthogonal polygon with vertices at integer coordinates.
There have apparently also been some other attempts to resolve the problem, leading to algorithms for special cases and other results~\cite{el2009decomposing,DBLP:conf/cccg/RenssenS11,DBLP:conf/cccg/El-KhechenDIO09}.
Our main result is that an even more restricted version of the problem is NP-hard.
A polygon $P$ is \emph{orthogonally convex} if, for any vertical or horizontal line $\ell$, the intersection $P\cap \ell$ is connected.
Note that an orthogonally convex polygon is necessarily simple.

\begin{restatable}{theorem}{mainthm}
\label{thm:main}
The problem \USPack\ is NP-hard, even for orthogonally convex grid polygons.
\end{restatable}

Allen and Iacono~\cite{DBLP:journals/corr/abs-1209-5307} mentioned that this special case of \USPack was the simplest packing problem with unknown complexity, and that the problem was ``most likely in P.''

Like the known reductions for polygons with holes, we also reduce from \PTSAT.
Interestingly, our geometric realization of a planar formula is non-planar, where binary values are represented by configurations of horizontal rows and vertical columns of squares, and these often intersect each other in the interior of the polygon.
The crucial observation is that movement in one direction does not influence movement in the other direction, so binary values can be ``transported'' independently in both directions through a crossing.

The technique proves to be general enough to also show hardness of some other problems.
We say that a polygon $Q$ is \emph{small} if $Q$ is contained in an axis-aligned $2\times 2$ square.
We show that it is NP-hard to find optimal covers and partitions of a simple polygon using small polygons.
In the problem \USCover, we are given as input a polygon $P$ and an integer $k$ and want to decide if there exists $k$ small polygons whose union is $P$.
The problem \USPartition\ is similar, but where we require the $k$ small polygons to be pairwise interior-disjoint.
We show that both of these problems are NP-hard, even when $P$ is simple.

\begin{theorem}\label{thm:covering}
The problem \USCover\ is NP-hard, even for simple grid polygons.
\end{theorem}

\begin{theorem}\label{thm:partition}
The problem \USPartition\ is NP-hard, even for simple grid polygons.
\end{theorem}

There have been many prior examples of covering problems that are intractable for simple polygons, and some are even $\exists\mathbb R$-hard like the Art Gallery Problem~\cite{DBLP:journals/jacm/AbrahamsenAM22} and covering with convex polygons~\cite{DBLP:conf/focs/Abrahamsen21}.
It is more remarkable that \Cref{thm:partition} gives the first example of a \emph{partitioning} problem that is hard already for simple polygons (with the usual objective of minimizing the number of pieces).
For other partitioning problems, like partitioning into convex~\cite{chazelle1985optimal} or star-shaped pieces~\cite{DBLP:journals/corr/abs-2311-10631} or trapezoids~\cite{DBLP:journals/jacm/AsanoAI86}, there are known polynomial-time algorithms.
Partitioning problems usually become hard in the presence of holes, which is also the case for the here mentioned problems~\cite{o1987art,DBLP:journals/jacm/AsanoAI86} (a notable exception is partitioning into rectangles, which can be done optimally in polynomial time even for polygons with holes~\cite{DBLP:journals/siamcomp/ImaiA86}).
Like packing, polygon decomposition forms a large subfield in computational geometry, with several books and survey papers that give an overview of the state-of-the-art at the time of publication~\cite{keil1985minimum,chazelle1985approximation,o1987art,DBLP:journals/pieee/Shermer92,chazelle1994decomposition,keil1999polygon,o2004polygons}.
Abrahamsen and Rasmussen~\cite{DBLP:journals/corr/abs-2211-01359} recently described an $13$-approximation algorithm for finding a partition of a simple polygon into a minimum number of small pieces.
The problem is motivated by various settings in manufacturing and shipping.

\subsection{Other related work on packing}

Several other packing problems have been shown to be NP-hard.
Here we mention the problem of packing axis-aligned squares of varying sizes into a square~\cite{DBLP:journals/jpdc/LeungTWYC90},
packing segments into a simple polygon~\cite{PackingSegments},
packing disks into a square~\cite{DBLP:journals/corr/abs-1008-1224} and packing $1\times 3$-rectangles (that can be rotated) into an orthogonal polygon with holes~\cite{DBLP:journals/comgeo/BeauquierNRR95}.

Allen and Iacono~\cite{DBLP:journals/corr/abs-1209-5307} showed that it is NP-hard to pack identical simple (small) polygons $Q$ into a simple (larger) polygon $P$.
Here, both $Q$ and $P$ are specified as part of the input.
We show that this problem is hard even when $Q$ is the $2\times 2$ square and $P$ is orthogonally convex.

Some packing problems are even known to be $\exists\mathbb R$-complete, and thus likely not in NP.
Abrahamsen, Miltzow and Seiferth~\cite{DBLP:journals/corr/abs-2004-07558} showed that when the pieces can be rotated, the problem of packing convex polygons into a square is $\exists\mathbb R$-complete.
When the pieces can only be translated, the problem is $\exists\mathbb R$-complete if arcs from hyperbolae can appear on the boundaries of the pieces or the container.

On the positive side, Hochbaum and Maass~\cite{DBLP:journals/jacm/HochbaumM85} gave a PTAS for \USPack in grid polygons with holes.
Faster schemes have since then been described by Agarwal, van Kreveld and Suri~\cite{DBLP:journals/comgeo/AgarwalKS98} and Chan~\cite{DBLP:journals/ipl/Chan04}.
El-Khechen~\cite{el2009decomposing} and van Renssen and Speckmann~\cite{DBLP:conf/cccg/RenssenS11} described families of simple grid polygons where \USPack can be solved optimally in polynomial time.
El-Khechen, Dulieu, Iacono and van Omme~\cite{DBLP:conf/cccg/El-KhechenDIO09} showed that \USPack into grid polygons with holes is in NP.
This is not immediately clear since, if $P$ is given in the standard representation as the coordinates of the vertices in binary, then the number of squares that fit in $P$ can be exponential in $n$ (the number of vertices), so specifying the placements of the squares is not a valid certificate.

Aamand, Abrahamsen, Ahle and Rasmussen~\cite{DBLP:journals/talg/AamandARA23} proved that packing dominoes, i.e., rectangles of size $1\times 2$ that can be rotated, into a given grid polygon $P$ with holes is polynomial-time solvable.
Already when we go to $2\times 2$-squares or $1\times 3$-rectangles, the problem becomes NP-hard, and as we show in this paper for $2\times 2$ squares, this is even the case for \emph{simple} grid polygons.

\subsection{Technical overview}

We present two constructions for packing.
The first construction shows that \USPack is NP-hard for simple grid polygons.
The second construction shows that this remains true if the polygon is restricted to be orthogonally convex.
The second result is strictly stronger than the first, but the proof is much more complicated, so we believe it is justified to give the constructions separately. 

Our constructions are made possible by two key ideas.
The first idea is a way to constrain what a packing can look like locally, even in parts of the polygon that are far from the boundary.
We show that we can define a certain set of $1\times 1$ squares called \emph{reference centers} in the polygon in such a way that any $2\times 2$ square in the polygon with integer coordinates must contain a reference center.
A packing is \emph{perfect} if there are as many squares as there are reference centers.
It is then straightforward to verify how the individual parts of our construction can be packed, since for each reference center, there are just four relevant $2\times 2$ squares to consider, namely one containing the reference center in each of the quadrants.

The second idea is a new way of drawing a planar graph where the vertices become rows and the edges become columns. 
We reduce from \PTSAT, and our drawing provides a schematic for our construction.
Rows and columns are allowed to cross in these schematics, but the planarity of the graph ensures that the endpoints of all rows and columns are incident to the outer face of the drawing.
This makes it possible to construct a simple polygon, following the boundary of the outer face, that ``implements'' the functionality of the rows and columns.

Variables are represented by neighbouring rows of squares that can be packed in essentially only two different ways, corresponding to the truth values of a variable.
Dependency between variables is created with so-called PUSH gadgets, where a column of \emph{push squares} will be pushed up or down depending on the value of a variable.
A crucial observation is that a column of push squares can cross any number of other variable rows without interacting with them.
Only the variable where the column ends will be affected by the push squares.

Our first construction, described in \Cref{sec:simplepacking}, shows that \USPack is NP-hard for simple grid polygons.
The construction is not so involved and quite straightforward to verify.

In \Cref{sec:covering}, we study the problems \USCover and \USPartition.
A construction very similar to the one for packing proves that \USCover is NP-hard, establishing \Cref{thm:covering}.
A main difference is that in packing, information propagate by squares \emph{pushing} each other, because they are not allowed to overlap, while in covering, the squares are \emph{pulling} each other, because they must cover everything.
We also show that in a satisfiable instance, an optimal covering has a nice property that allows it to be converted into a partition of the same cardinality, hence proving \Cref{thm:partition}.

In \Cref{sec:orthconvexpack}, we then turn our attention to packing $2\times 2$ squares in an orthogonally convex grid polygon.
This is a strictly stronger result than the one from \Cref{sec:simplepacking}, that requires a different, much more complicated construction.
We develop the problem \emph{\CTSAT}, which is a modification of \PTSAT\ under which our schematics have an orthogonally convex shape. This builds on work by Pilz \cite{PilzLayeredPlanar3SAT}, who showed that a related problem called \textsc{Variable-Clause-Linked-Planar-3SAT} is NP-hard. 

The variable and clause components are built out of a large number of criss-crossing rows and columns. Since it is harder to isolate the different pieces, the dependencies between different parts of the construction are much more complicated. We need a long technical verification to ensure that the packings behave as we intend.

\section{Packing in simple polygons}\label{sec:simplepacking}

\subsection{Schematics of the construction}\label{sec:schematic}

We reduce from \MPTSAT, as introduced by de Berg and Khosravi~\cite{DBLP:journals/ijcga/BergK12}.
An instance of \TSAT\ is \emph{monotone} if in each clause, all literals are positive or all literals are negative.
An instance $\Phi=(F,G)$ of \MPTSAT\ consists of a monotone instance $F$ of \TSAT\ with variables $x_1,\ldots,x_k$ and clauses $c_1,\ldots,c_\ell$ and a plane graph $G$ where
\begin{itemize}
\item the vertices of $G$ are $\{x_1,\ldots,x_k,c_1,\ldots,c_\ell\}$,
\item the edges of $G$ are $\{x_ic_j\mid x_i\in c_j\lor \lnot x_i\in c_j\}\cup C$, where $C=\{x_1x_2,\ldots,x_{k-1}x_k,x_kx_1\}$,
\item the cycle $C$ separates all positive clauses from all negative clauses.
\end{itemize}
The following lemma was proved in~\cite{DBLP:journals/ijcga/BergK12}.

\begin{lemma}
\MPTSAT\ is NP-complete.
\end{lemma}

Let $\Phi$ be an instance of \MPTSAT\ with variables $x_1,\ldots,x_k$.
We now describe schematically the overall construction of a grid polygon $P$, so that $P$ has a packing with a certain number of $2\times 2$ squares if and only if $\Phi$ is satisfiable.
The first step is shown in \Cref{fig:ConversionSchematic2} (middle).
We make a horizontal segment for each variable $x_1,\ldots,x_k$ in this order bottom-up, where the $x$-coordinates of the right endpoints are non-decreasing, as are those of the left endpoints.
Furthermore, the right endpoint of $x_1$ is to the right of the left endpoint of $x_k$, so that all the segments have a common horizontal overlap.
These are the \emph{main} variable rows.

The positive clauses are represented as rows above $x_k$ and the negative as rows below $x_1$.
In the embedding of $G$, a clause can be \emph{nested} inside another, defining a partial order on the clauses.
For instance, in \Cref{fig:ConversionSchematic2}, $c_2$ and $c_3$ are nested inside $c_1$.
If a positive clause $c_i$ is nested inside $c_j$, then we draw the row of $c_i$ above that of $c_j$ in the schematics.
If instead the clauses are negative, the row of $c_i$ is below that of $c_j$.
Each edge $x_ic_j$ is realized as a vertical segment connecting the rows of $x_i$ and $c_j$.
All vertical segments to the positive clauses are placed to the right of all those to the negative clauses.
The left-to-right ordering of the edges to the positive clauses along the variables $x_1,\ldots,x_k$ in $G$ is preserved by the corresponding vertical segments, as well as the ordering of the edges to the negative clauses.
The row of a positive (resp.~negative) clause $c_j$ starts at the top (resp.~bottom) endpoint of the vertical segment corresponding to the leftmost edge incident to $c_j$ in $G$ and ends at the top (resp.~bottom) endpoint of the rightmost edge.

In the next step, as shown in \Cref{fig:ConversionSchematic2} (right), we replace each clause row by three \emph{auxiliary} variable rows, one for each of the variables connected to the clause.
In the right side of the auxiliary rows, we connect them using vertical segments to an OR gadget, represented by a red horizontal segment.
As sketched in the figure, we construct a polygon whose boundary approximately follows the outer face of the resulting arrangement, but the finer details will be given later.

\begin{figure}
\centering
\includegraphics[page=1,width=\textwidth]{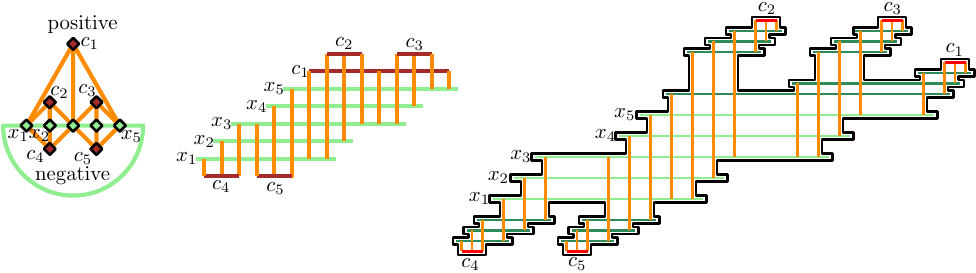}
\caption{The figure shows schematically how to convert an instance of \MPTSAT\ into a polygon.
Left: The instance of \MPTSAT that we start with.
Middle: We realize variables and clauses as rows.
Right:
We replace the clause rows with dark green auxiliary variables connected to red OR gadgets.
The boundary of our constructed polygon is sketched in black.
}
\label{fig:ConversionSchematic2}
\end{figure}

The use of auxiliary variables is necessary to make sure the OR gadgets can be placed appropriately, so that they are not crossed by other columns and without the need of creating holes in the final polygon.
In order to not introduce any holes in our polygon, it is important that every endpoint of a segment be incident to the outer face of the arrangement induced by the segments.
Indeed, the endpoints correspond to some ``functionality'' that we need to realize using an appropriate polygon boundary.

\begin{lemma}
There is a schematics as described, where the segments representing the OR gadgets are not crossed by any edges, and all segment endpoints are incident to the outer face of the drawing.
\end{lemma}

\begin{proof}
We proceed by induction on the number of clauses.
The claim is trivial with no clauses.
Consider a formula $\Phi$ with $n$ clauses and suppose inductively that the claim holds for $n-1$ clauses.
Let $c$ be a clause of maximum depth of nestedness in $\Phi$ and let $\Phi'$ be $\Phi$ with $c$ removed.
By the induction hypothesis, we can consider a schematics of $\Phi'$ with the stated properties.
Without loss of generality, we consider the case where $c$ is a positive clause, so in the schematics, $c$ should be realized with three auxiliary variable rows connected to an OR gadget, and the auxiliary variables should be above the other OR gadgets.
We can stretch the schematics of $\Phi'$ horizontally to make enough room for the columns connecting the main variable rows and the new auxiliary variable rows.
These columns then appear consecutively along the main variable rows.
We can then also draw the auxiliary variable rows and the OR gadget.
We thus avoid crossing through other OR gadgets and we keep all segment endpoints incident to the outer face of the drawing.
\end{proof}

Algebraically, we make use of the following equivalence when introducing auxiliary variables:
\begin{align}
& (x_i\vee x_j \vee x_k)
\iff \nonumber \\
& \Big(\exists y_i,y_j,y_k:
(y_i\implies x_i)\wedge (y_j\implies x_j)\wedge(y_k\implies x_k)\wedge (y_i\vee y_j\vee y_k)\Big)
\label{eq:xy}
\end{align}
Here, $x_i,x_j,x_k$ are the original variables and $y_i,y_j,y_k$ are the auxiliary variables introduced for that particular clause.
For a negative clause, we use an analogous equivalence:
\begin{align}
& (\neg x_i\vee \neg x_j \vee \neg x_k)
\iff \nonumber \\
& \Big(\exists y_i,y_j,y_k:
(\neg y_i\implies \neg x_i)\wedge (\neg y_j\implies \neg x_j)\wedge(\neg y_k\implies \neg x_k)\wedge (\neg y_i\vee \neg y_j\vee \neg y_k)\Big)
\label{eq:xyneg}
\end{align}

\subsection{Reference centers}\label{sec:refcenters}

By the following lemma, we can restrict our attention to packings where the squares have integer coordinates.
The lemma appears to be folklore, but is also proved in \cite{DBLP:journals/algorithmica/BaurF01}.

\begin{lemma}\label{Finiteness}
If a grid polygon can be packed with $k$ axis-aligned $2\times 2$ squares, then such a packing can be chosen where the coordinates of the vertices of all the squares are integers.
\end{lemma}

\begin{proof}
Consider a packing with $k$ squares that minimizes the sum  over all squares of the sum of both coordinates of the square center.
\end{proof}

We are going to construct a grid polygon $P$ based on the formula $\Phi$.
The unit squares of the form $[2k-1,2k]\times [2\ell-1,2\ell]$ for $k,\ell\in\mathbb Z$ that are contained in $P$ are called the \emph{reference centers}.

\begin{lemma}\label{lem:referenceCenters}
In a packing of a grid polygon using $2\times 2$ squares with integer coordinates, each square covers exactly one reference center.
\end{lemma}

\begin{proof}
The statement follows by inspection of the four combinations of the parity of the coordinates of a $2\times 2$ square with integer coordinates.
\end{proof}

We get from \Cref{lem:referenceCenters} that there can be at most as many squares as reference centers in any packing.
We say that a packing is \emph{perfect} if it has squares of integer coordinates and consists of as many squares as there are reference centers.
This means that for each reference center, there are only $4$ possible ways for a square in a perfect packing to cover the reference center, i.e., the reference center must be in one of the four quadrants of the square.
This makes it straightforward to verify the individual parts of our construction, at least in principle.

A square in a perfect packing \emph{pushes up} (resp.~\emph{down}, \emph{left}, \emph{right}) if the reference center is contained in one of the lower (resp.~upper, right, left) quadrants.
For clarity, we will draw the reference centers in our diagrams.
We use a system of color coding that describes some of the constraints on how that tile can move, as shown in \Cref{fig:ColorKey}.
Squares with pink reference centers will always push down and left, while squares with red reference centers will always push up and left.
Squares with orange reference centers can move vertically but will always push left. 
Squares with green reference centers can move horizontally but will always push down. 
Squares with dark blue reference centers will always push up.
Squares with light blue reference centers will have some freedom in both horizontal and vertical directions.

\begin{figure}
\centering
\includegraphics[page=4]{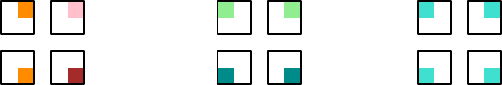}
\caption{The color key.}
\label{fig:ColorKey}
\end{figure}

\subsection{Components and gadgets}\label{sec:types}

We informally distinguish between \emph{gadgets} and \emph{components} in our constructions, where a gadget is a few edges of the boundary creating some simple functionality, and a component is thought of as a whole region of the polygon that can involve an arbitrary number of gadgets.

First, we describe a \emph{variable component} in \Cref{sec:simplevar}, which has two possible positions representing the values of a binary variable.
The variable components are represented in these schematics as green horizontal rows.
The two positions of a variable component are \emph{plus} and \emph{minus}, corresponding to the values true and false of a binary variable.

Next, we create \emph{PUSH gadgets} in \Cref{sec:rows,sec:implies}.
A PUSH gadget is a section of the polygon boundary that interacts with a variable component.
Each PUSH gadget pushes on a \emph{push column}, which is shown as an orange column in our schematics.
We describe a PUSH-UP-IF-MINUS and a PUSH-DOWN-IF-PLUS gadget.
Consider a variable component $x$ that is below another $y$.
We can make a PUSH-UP-IF-MINUS gadget on $x$ and a PUSH-DOWN-IF-PLUS on $y$ and a connection between these using a push column.
Since a push column cannot be pushed both up and down at the same time, we have ensured that $x$ is plus or $y$ is minus, so we have made the implication $y\implies x$.

Note that variable components only push on push columns via the presence of PUSH gadgets.
Without a PUSH gadget, a push column and a variable component will cross each other without interacting.

The last gadget is an OR gadget, as described in \Cref{sec:or}.
We represent this as a red row in our schematics. An OR gadget interacts with the top of three push columns and always pushes down on at least one of them.
Unlike the variable components, a push column \emph{cannot} pass through an OR gadget; any push column that touches an OR gadget should terminate there.
This is why we have attached the OR gadget using auxiliary rows, as in \Cref{fig:ConversionSchematic2} (right).

\subsection{The variable component}\label{sec:simplevar}

Consider a variable row $x$ in the schematics.
We make a corresponding variable component that contains a pair of rows of squares.
At the ends of the rows, the polygon boundary constrain the rows to always be unaligned, as shown in \Cref{fig:VarGadget}.
In between the ends, there is a PUSH gadget for each edge to the variable row $x$ in the schematics.
The PUSH gadgets are described in \Cref{sec:implies} below.

\begin{lemma}\label{lem:variable}
In any perfect packing, one row of the variable component pushes right and the other pushes left.
\end{lemma}

\begin{proof}
We first consider the two leftmost squares in the variable component and note that the polygon boundary forces at least one of them to push to the right. Similarly, (at least) one of the two rightmost squares must push to the left. So all the squares in one of the rows must push to the left and all the squares in the other row must push to the right.
\end{proof}

As we have seen, a perfect packing of the variable component has two possible positions, and they correspond to the values of a binary variable.
The transition from one position to the other corresponds to all the squares rotating one step either clockwise or counterclockwise around the cycle formed by the squares.
We say that \emph{plus} is the position where the squares are rotated in the positive (i.e., counterclockwise) direction, and \emph{minus} is the position where they rotated in the negative direction.

\begin{figure}
\centering
\includegraphics[page=1,scale=0.666666]{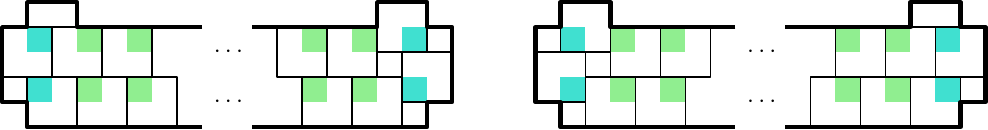}
\caption{The ends of the variable component.
The plus position is shown left and the minus is shown right.
Some squares in the middle may be pushed up.}
\label{fig:VarGadget}
\end{figure}

\subsection{Rows and alignment}\label{sec:rows}

We describe a way to make a dependency between two variable components $x$ and $y$, where $x$ is below $y$, as shown schematically in \Cref{fig:ImpliesSchematic}.
This is done by an upside down pyramid of squares that are raised by $1$ unit.
The squares are raised by a PUSH gadget of the lower variable $x$.
The pyramid should be able to cross variable rows of other variables in between $x$ and $y$ without interacting with them.

One extra layer of squares on the left edge of the pyramid may move up or stay down depending on the positions of $x$ and $y$.
The squares in this layer are called a \emph{push column} or \emph{push} squares.
This layer acts like a wire connecting the bottom and top edges of the polygon that bound the variable components $x$ and $y$, respectively.
The gadgets described in \Cref{sec:implies} work by pushing on the push squares.
This structure is pyramid shaped because a stack of raised squares has to increase in width every time it passes between a pair of rows with opposite horizontal alignments.

\begin{figure}
\centering
\includegraphics[page=2]{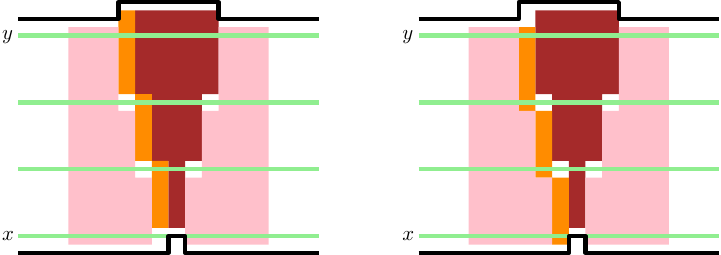}
\caption{Schematic of how to make a dependency between two variable components.
When crossing each variable component in between, the width of the pyramid grows by two squares.
The push column is shown in orange. Depending on the position of the bottom variable component $x$, this column may be pushed up (left).
Depending on the position of the top variable component $y$, it may be pushed down (right).
Since the column can't be pushed both up and down, this creates a constraint between $x$ and $y$.}
\label{fig:ImpliesSchematic}
\end{figure}

\begin{figure}
\centering
\includegraphics[page=2,scale=0.6666666]{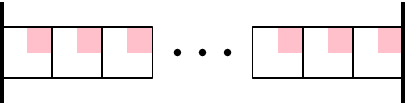}
\caption{A static row.
Some squares may be pushed up, but all push left.}
\label{fig:StaticRow}
\end{figure}

Recall that each variable component consists of one pair of neighbouring unaligned rows.
We want to control the width of the pyramid at the top, which means that the number of times it crosses a pair of unaligned rows should not depend of the positions of the individual variable components crossed by the pyramid.
To do this, we place each variable component between two \emph{static} rows such that, in either position of the variable component, one of the static rows is aligned with the adjacent row in the variable component and the other is unaligned.
\Cref{fig:StaticRow} shows a static row.
In each static row, the rightmost square is pushed to the left by an edge of the polygon, so the whole row of squares are in a left position.

\begin{figure}
\centering
\includegraphics[page=1]{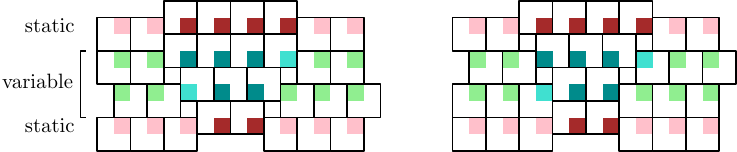}
\caption{A pyramid crossing a variable component.
The size of the pyramid always grows by 2 squares, regardless of the position of the variable component.}
\label{fig:RowAlignment}
\end{figure}

\Cref{fig:RowAlignment} shows in detail how the width of a pyramid grows by two squares whenever it crosses a variable component.
This is expressed by the following lemma.

\begin{lemma}\label{lem:PyramidGrowth}
Consider a perfect packing.
If a block of consecutive squares in a static row below a variable component pushes up, then the corresponding squares and their two horizontal neighbors in the static row above the variable component must also push up.
If a block of consecutive squares in an upper static row pushes down, then the corresponding squares and their two horizontal neighbors in the static row below must also push down.
\end{lemma}

\begin{proof}
We prove the first claim; the other one is analogous.
A static row is always left-aligned with respect to the reference centers.
One of the rows of the variable component is right-aligned, which causes one more square to the left of that row to rise, as compared to the row below.
In the next row (which is either the upper variable row or the static row above the variable), one more square to the right is then also pushed up.
\end{proof}

\subsection{The PUSH gadgets}\label{sec:implies}

Consider a variable $x$ that is part of a positive clause in $\Phi$.
We then have an auxiliary variable $y$ above $x$, and we need to make the implication $y\implies x$.
We create an upside-down pyramid from $x$ to $y$ that has a column of push squares on the left. 
This creates the constraint that both gadgets cannot simultaneously push on the push column, i.e.,
\[x \text{ does not push up on its push square} \vee y \text{ does not push down on its push square.}\]

It hence suffices to make gadgets ensuring
\begin{align*}
x \text{ pushes up on its push square} & \iff x\text{ is minus}, \quad \text{and} \\
y \text{ pushes down on its push square} & \iff y\text{ is plus}
\end{align*}

To this end, we make the two gadget \emph{PUSH-UP-IF-MINUS} and \emph{PUSH-DOWN-IF-PLUS}, respectively.
\Cref{fig:ImpliesBottomAligned} shows the PUSH-UP-IF-MINUS gadget, which simply consists of small indent in the polygon wall that forces some squares to rise.
\Cref{fig:ImpliesLeftAligned} shows the PUSH-DOWN-IF-PLUS gadget, which consists of an \emph{out}dent in the polygon wall allowing some squares to rise.
The width of this outdent is adjusted according to the width of the push pyramid, which, as stated by \Cref{lem:PyramidGrowth}, depends on the number of variable components it crosses.
This also ensures that the pyramids don't collide in the interior of the polygon.

\begin{figure}
\centering
\includegraphics[page=2]{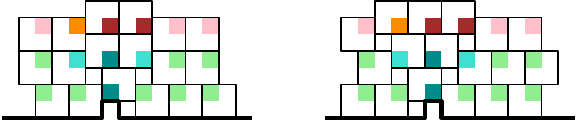}
\caption{The PUSH-UP-IF-MINUS gadget.
The gadget forces some number of squares in the static row (top, red) to rise, and requires an additional push square to the left to rise if the variable component (bottom two rows, green) is minus.
The plus position is shown on the left and the minus position is shown on the right.}
\label{fig:ImpliesBottomAligned}
\end{figure}

\begin{figure}
\centering
\includegraphics[page=1]{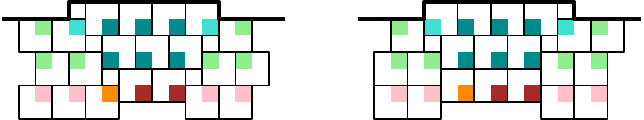}
\caption{The PUSH-DOWN-IF-PLUS gadget.
The gadget allows some number of squares in the static row (bottom, red) to rise, and \emph{allows} the square to left of these to rise only if the variable component (top two rows, green) is plus.
Note that the top gadgets may need to be much wider since the constraint pyramid could increase in width many times.}
\label{fig:ImpliesLeftAligned}
\end{figure}

\begin{figure}
\centering
\includegraphics[page=3,width=\textwidth]{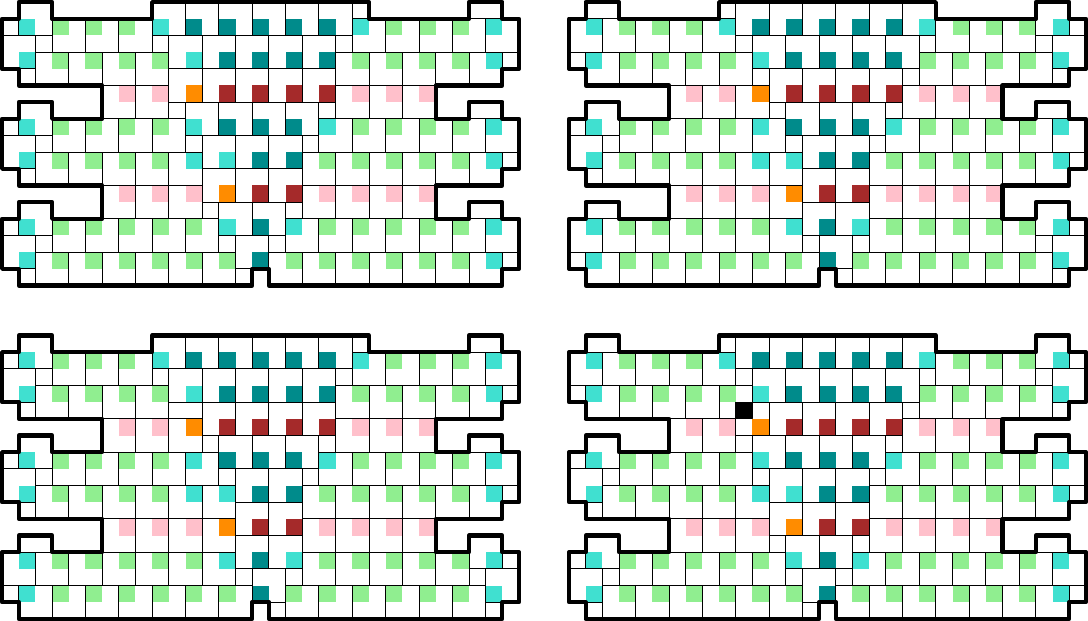}
\caption{Example of a dependency using the PUSH-UP-IF-MINUS gadget on the bottom variable and the PUSH-DOWN-IF-PLUS gadget on the top variable.
Top left: The bottom and top variables are both minus.
Top right: The bottom and top variables are both plus.
Bottom left: The bottom variable is plus and the top is minus.
Since no gadget pushes on the push square, this introduces some slack so there are more than one perfect packing.
Bottom right: The bottom variable is minus and the top is plus, which causes an overlap at the black square.}
\label{fig:BigImplies1}
\end{figure}

Suppose now that $x$ is part of a negative clause.
We then have an auxiliary variable $y$ below $x$, and need to realize the implication $\neg y\implies\neg x$, or equivalently $x\implies y$.
We can thus realize this implication using the gadget PUSH-DOWN-IF-PLUS on $x$ and PUSH-UP-IF-MINUS on $y$. 

\Cref{fig:BigImplies1} shows an example of how to make a dependency using these two gadgets.
Out of the four possible combinations of positions of the top and bottom variable components, three are allowed by the constraint.
Two of these have one gadget pushing on the push column, while the third has neither gadget pushing its push square.

In conclusion, we have described gadgets with the properties stated by the following lemma.

\begin{lemma}\label{lem:implies}
In any perfect packing, the following holds.
If a variable is minus, the push square of any PUSH-UP-IF-MINUS gadget pushes up.
If a variable is plus, the push square of any PUSH-DOWN-IF-PLUS gadget pushes down.
As a consequence, combining a PUSH-UP-IF-MINUS gadget and a PUSH-DOWN-IF-PLUS gadget, we can realize the implications $y\implies x$ and $\neg y\implies\neg x$, i.e., ensure that the values of variables encoded by the packing satisfy the implications.
\end{lemma}

\subsection{The OR gadgets}\label{sec:or}

We create two types of OR gadgets, namely positive OR gadgets for the positive clauses and negative for the negative clauses.
\Cref{fig:OrGadget} shows the positive OR gadget.
The OR gadget is connected to three auxiliary variables that are below the gadget, using upside down pyramids of raised squares with push columns, as shown in \Cref{fig:OrGadgetConnection}.
These pyramids are created by PUSH-UP-IF-MINUS gadgets on the auxiliary variables.
If all three push columns push up, then it is impossible to pack the OR gadget.
On the other hand, \Cref{fig:OrGadgetOpen} shows that it is possible to pack the OR gadget if one or more of the columns is not pushed up.
In this way we can think of the OR gadget as being a gadget that pushes down on one of three input columns.
Hence, one of the auxiliary variables must be plus.

\Cref{fig:NegOrGadget} shows the negative OR gadget.
Here we connect the auxiliary variables using PUSH-DOWN-IF-PLUS gadgets, as shown in \Cref{fig:NegOrGadgetConnection}.
We use a slightly different variant of the PUSH-DOWN-IF-PLUS gadget than what is shown in \Cref{fig:ImpliesLeftAligned}:
We shift the outdent one unit square to the right, which causes the column of push squares to be on the right side of the pyramid instead of the left.
This allows for a slightly simpler OR gadget than if we were to use the usual push column on the left side of the pyramids.
By inspection, we obtain the following lemma.

\begin{figure}
\centering
\includegraphics[page=1]{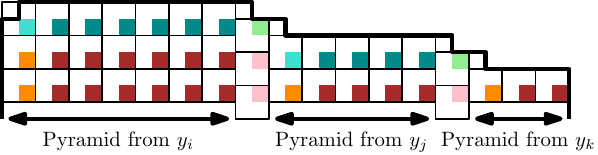}
\caption{The positive OR gadget.
In this figure, all the orange push columns are pushed up, meaning that we can't pack all the squares into the OR gadget.
The labels $y_i,y_j,y_k$ refer to \Cref{eq:xy}.
}
\label{fig:OrGadget}
\end{figure}

\begin{figure}
\centering
\includegraphics[page=2]{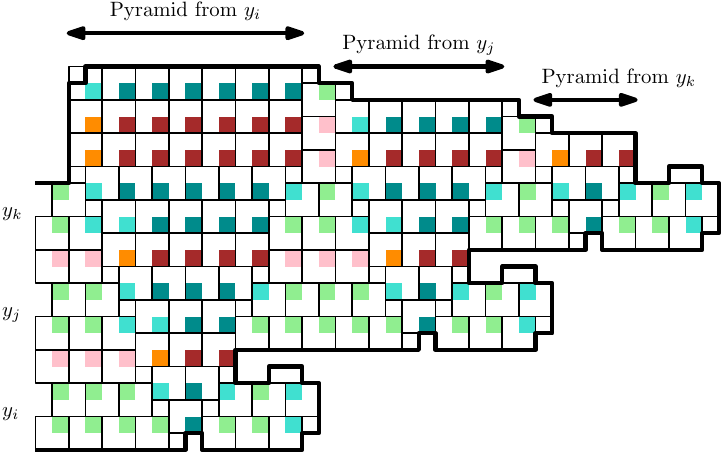}
\caption{The figure shows how the auxiliary variable components are connected to the OR gadget using pyramids.
Here is shown the situation where all variables are minus, so there is no perfect packing and the top left square sticks out.}
\label{fig:OrGadgetConnection}
\end{figure}

\begin{figure}
\centering
\includegraphics[page=3,width=\textwidth]{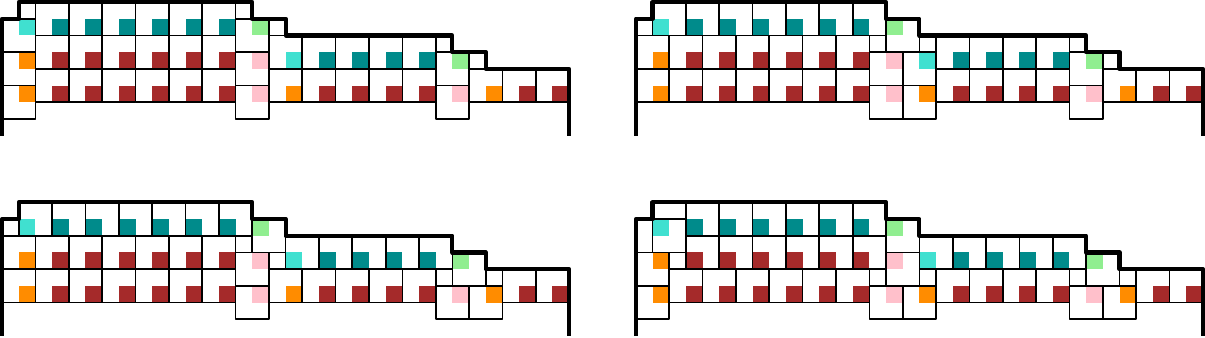}
\caption{Satisfying assignments of the OR gadget.
As in \Cref{fig:OrGadget}, we say that the push columns represent values $y_i,y_j,y_k$ from left to right, respectively.
The top two and the bottom left diagrams show packings where the push squares of $y_i,y_j,y_k$ are down, respectively.
To the bottom right is shown the situation that all three are down, which creates some slack so that there are more perfect packings.
}
\label{fig:OrGadgetOpen}
\end{figure}

\begin{figure}
\centering
\includegraphics[page=4]{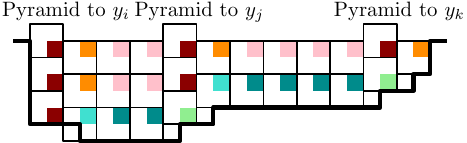}
\caption{The negative OR gadget.
In this figure, all the orange push columns are pushed down, meaning that we can't pack all the squares into the OR gadget.
The labels $y_i,y_j,y_k$ refer to \Cref{eq:xyneg}.
}
\label{fig:NegOrGadget}
\end{figure}

\begin{figure}
\centering
\includegraphics[page=5]{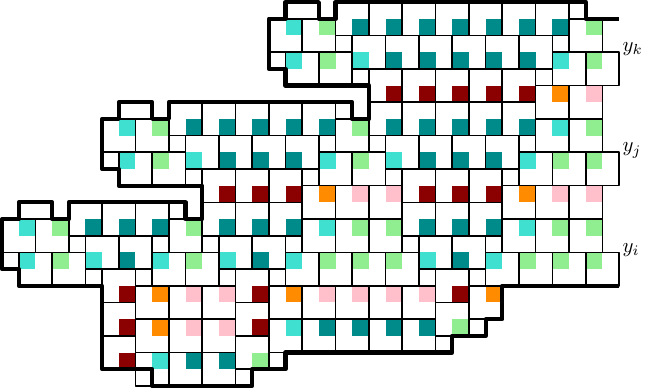}
\caption{The figure shows how the variable components are connected to the negative OR gadget.
Here is shown the situation where all variables are plus, so there is no perfect packing, and the bottom left square sticks out.}
\label{fig:NegOrGadgetConnection}
\end{figure}

\begin{lemma}\label{lem:or}
In a positive OR gadget, at least one of the connected auxiliary variables is plus.
In a negative OR gadget, at least one of the connected auxiliary variables is minus.
\end{lemma}

\subsection{Verifying the construction}\label{sec:verificationSimple}

It is clear from our diagrams that a perfect packing exists when the instance is satisfiable, so it remains to check that every perfect packing corresponds to a satisfying assignment of $\Phi$, as established by the following lemma.

\begin{lemma}\label{lem:verifySimple}
If there exists a perfect packing of $P$, then $\Phi$ is satisfied.
\end{lemma}

\begin{proof}
By \Cref{lem:variable}, each variable component encodes a value of one of the variables $x_i$ of $\Phi$ or an auxiliary $y_i$ used for an OR gadget.
Consider a clause $C(x_i,x_j,x_k)$ of $\Phi$.
By \Cref{lem:or}, the packing encodes values of the auxiliary variables $y_i,y_j,y_k$ that satisfy the corresponding clause $C(y_i,y_j,y_k)$.
It now follows from \Cref{lem:implies} that $C(x_i,x_j,x_k)$ is also satisfied.
Hence, the full formula $\Phi$ is likewise satisfied.
\end{proof}

It is clear that our polygon can be constructed in polynomial time from the instance $\Phi$.
We have thus proven the following theorem.

\begin{theorem}\label{thm:simplepack}
The problem \USPack\ is NP-hard, even for simple grid polygons.
\end{theorem}

\section{Covering and partitioning}\label{sec:covering}

The proof of \Cref{thm:covering} is almost analogous to that of \Cref{thm:simplepack}, but we need to modify the gadgets and replace ``push'' by ``pull'' in many places.

Let us first show a connection between covering and partitioning.
Recall that we define a polygon $Q$ to be \emph{small} if $Q$ is contained in an axis-aligned $2\times 2$ square.

\begin{lemma}\label{lem:coverpartition}
Suppose a polygon $P$ is contained in the union of $k$ axis-aligned $2\times 2$ squares in such a way that the intersection of $P$ with each square is a star-shaped polygon with the kernel containing the center of the square. Then $P$ can be partitioned into $k$ small polygons.
\end{lemma}

\begin{proof}
Consider the $L_\infty$ distance Voronoi diagram for the set of centers of squares.
For each square $S$ with center $c$ and Voronoi region $V$, we use $V\cap P$ as a piece.
We need to show that $V\cap P$ is small, and it suffices to show that (i) $V\cap P\subseteq S$ and (ii) $V\cap P$ is connected.
Part (i) follows since the squares cover $P$, so every point in $P$ is within a distance of $1$ from the center of some square.
For part (ii), note that since $V$ is a Voronoi cell, $V$ is a star-shaped polygon with a kernel containing $c$, and $P\cap S$ is also a star-shaped polygon with a kernel containing $c$ by assumption.
Hence, $V\cap P = V\cap (P\cap S)$ is also a star-shaped polygon and therefore connected.
\end{proof}

Note that axis-aligned unit squares could be replaced by circles, or indeed any symmetric convex shape (as long as rotations are not allowed), and the same result can be obtained by changing the metric used.
The optimal coverings considered in the proof of \Cref{thm:covering} satisfy the conditions of the lemma, implying \Cref{thm:partition}.

Consider a polygon $P$ and a set $\mathcal S$ of axis-aligned $2\times 2$ squares so that $P\subseteq \bigcup \mathcal S$.
We say that $\mathcal S$ is a \emph{square cover} for $P$.
In our construction, we adopt the use of reference centers, as described in~\Cref{sec:refcenters}.
Since the area of reference centers covered by one $2\times 2$ square is exactly $1$, we know that $\lvert \mathcal S\rvert \geq k$, where $k$ is the number of reference centers.
We say that the cover $\mathcal S$ is \emph{perfect} if $\lvert \mathcal S\rvert = k$.

We construct a polygon $P$ based on a \MPTSAT\ instance $\Phi$ so that $P$ has a perfect square cover if and only if $\Phi$ is satisfiable.
For our polygon $P$ it holds that if a perfect square cover $\mathcal S$ exists, then for each square $S\in\mathcal S$, the polygon $P\cap S$ is either $S$ or three quadrants of $S$, in particular $P\cap S$ is connected.
Hence, $\Phi$ is satisfiable if and only if there are $k$ small polygons whose union is $P$, so \USCover\ is NP-hard.
As the polygons $P\cap S$ have that form, we also know by \Cref{lem:coverpartition} that $\Phi$ is satisfiable if and only if there is a partition of $P$ into $k$ small polygons.
Hence, \USPartition\ is also NP-hard.
In the following, we therefore analyze how a perfect square cover for $P$ must look.

Analogous to \Cref{Finiteness}, we have the following.

\begin{lemma}
If an orthogonal polygon with integer coordinates has a square cover of size $k$, then such a square cover can be chosen where the coordinates of the vertices of all the squares are integers.
\end{lemma}

This again allows us to consider only four possible squares for each reference center, with the reference center in each of the quadrants.

\subsection{Components and gadgets}

\begin{figure}
\centering
\includegraphics[page=1]
{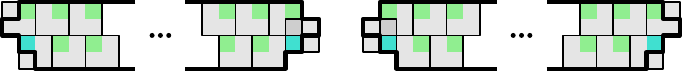}
\caption{Variable component for the covering problem.}
\label{fig:covervariable}
\end{figure}

Due to the similarity with our construction for packing, we shall give a less detailed description of the construction.
\Cref{fig:covervariable} shows the variable component, and we have the following analogue of \Cref{lem:variable}.

\begin{lemma}\label{lem:variableCover}
In any perfect cover, one row of the variable component pulls left and the other pulls right.
\end{lemma}

\begin{proof}
One row must pull left to cover the left end of the variable component, and the other must pull right to cover the right end.
\end{proof}

\begin{figure}
\centering
\includegraphics[page=2]
{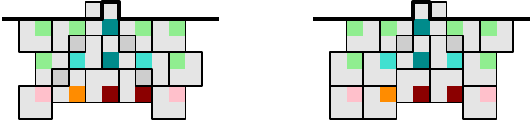}
\caption{The PULL-UP-IF-PLUS gadget.
The gadget pulls up a pyramid of squares, and the orange pull square is also pulled up in the plus position.}
\label{fig:pullupifup}
\end{figure}

\begin{figure}
\centering
\includegraphics[page=3]
{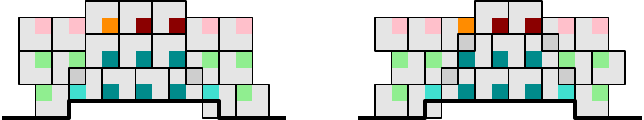}
\caption{The PULL-DOWN-IF-MINUS gadget.
The gadget allows some number of squares to pull up and requires the orange pull square to stay down in the minus position.}
\label{fig:pulldownifdown}
\end{figure}

\Cref{fig:pulldownifdown,fig:pullupifup} shows the PULL-UP-IF-PLUS and PULL-DOWN-IF-MINUS gadgets.
Instead of a column of push squares, we now have a column of \emph{pull squares}.
Note that a PULL-UP-IF-PLUS gadget on a variable $y$ forces a pyramid of squares to pull up.
Similarly as in packing, we can make a matching PULL-DOWN-IF-MINUS gadget on a variable $x$ below $y$, resulting in a dependency between the variables.
We have the following analogy of \Cref{lem:implies}.

\begin{lemma}\label{lem:impliesCovering}
In any perfect covering, the following holds.
If a variable is plus, the pull square of any PULL-UP-IF-PLUS gadget pulls up.
If a variable is minus, the pull square of any PULL-DOWN-IF-MINUS gadget is pulls down.
As a consequence, combining a PULL-UP-IF-PLUS gadget and a PULL-DOWN-IF-MINUS gadget, we can realize the implications $y\implies x$ and $\neg y\implies\neg x$, i.e., ensure that the values of variables encoded by the packing satisfy the implications.
\end{lemma}

\begin{figure}
\centering
\includegraphics[page=5,width=\textwidth]
{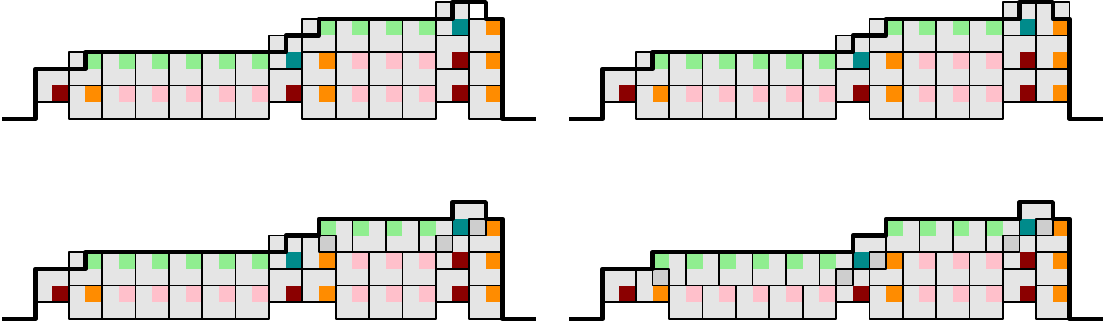}
\caption{The positive OR gadget.
Top left: The situation where all the orange pull squares pull down, and a unit square at the top is not covered.
The other three:
If one pull square goes up, everything can be covered.}
\label{fig:coverORPos}
\end{figure}

\begin{figure}
\centering
\includegraphics[page=4]
{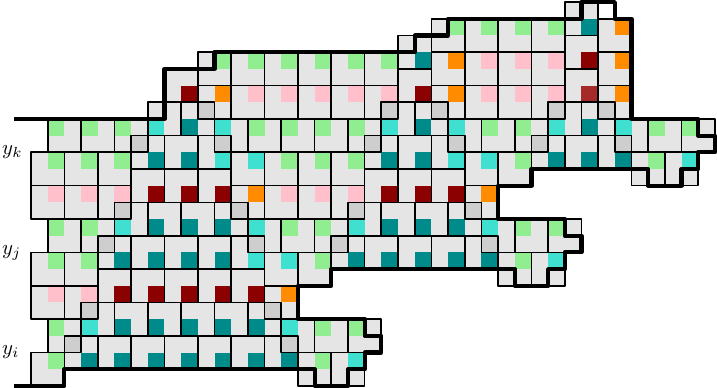}
\caption{Connecting the auxiliary variables to the positive OR gadget using PULL-DOWN-IF-MINUS gadgets.
We have made a slight adjustment to the PULL-DOWN-IF-MINUS gadgets as compared to \Cref{fig:pulldownifdown} in order to have the pull squares in the right side of the pyramids.}
\label{fig:coverORPosCon}
\end{figure}

\begin{figure}
\centering
\includegraphics[page=7]
{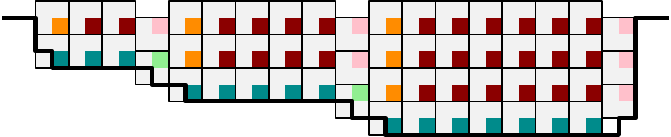}
\caption{The negative OR gadget.
We show the situation where all the orange pull squares pull up, and a unit square at the bottom is not covered.
If one pull square goes down, everything can be covered.}
\label{fig:coverORNeg}
\end{figure}

\begin{figure}
\centering
\includegraphics[page=6]
{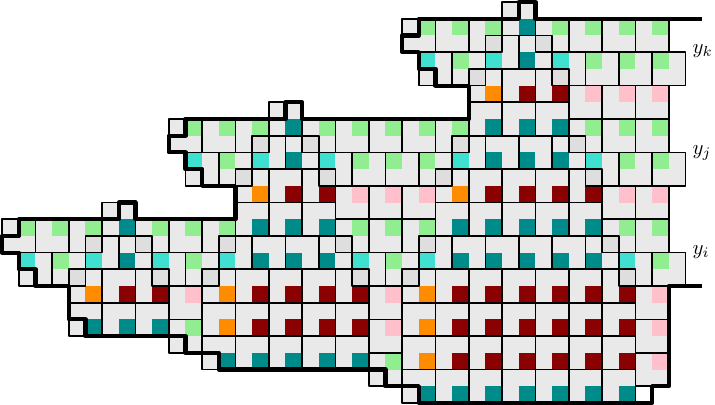}
\caption{Connecting the auxiliary variables to the negative OR gadget using PULL-UP-IF-PLUS gadgets.}
\label{fig:coverORNegCon}
\end{figure}

\Cref{fig:coverORPos,fig:coverORPosCon} show the positive OR gadget, and \Cref{fig:coverORNeg,fig:coverORNegCon} show the negative.
By inspection, we get the following lemma, identical to \Cref{lem:or}.

\begin{lemma}\label{lem:orCover}
In a positive OR gadget, at least one of the connected auxiliary variables is plus.
In a negative OR gadget, at least one of the connected auxiliary variables is minus.
\end{lemma}

We conclude with the following analogue of \Cref{lem:verifySimple} (which has an analogous proof).

\begin{lemma}\label{lem:verifySimpleCover}
If there exists a perfect covering of $P$, then $\Phi$ is satisfied.
\end{lemma}

We have then shown \Cref{thm:covering}.
Note that if a perfect cover exists, then so does a perfect cover where the intersection of $P$ with any of the squares $S$ is either $S$ or three of the quadrants of $S$.
Hence, we get \Cref{thm:partition} from \Cref{lem:coverpartition}.

\section{Packing in orthogonally convex polygons}\label{sec:orthconvexpack}

Finally, we show that \USPack for orthogonally convex grid polygons is also NP-hard.
This is a strictly stronger result than \Cref{thm:simplepack}, but the proof is so much more difficult that it seems better to give it separately.
Recall that a polygon $P$ is \emph{orthogonally convex} if for any horizontal or vertical line $\ell$, the intersection $P\cap\ell$ is connected.

\subsection{Overview of the construction}

The idea is again to convert the incidence graph of an instance of \TSAT\ into a schematic, which then forms the structure of the polygon. The schematics that we use in \Cref{sec:schematic} are usually not going to have an orthogonally convex shape. We refine \PTSAT in such a way that the schematics generated \emph{are} orthogonally convex. We call this new problem \emph{\CTSAT}. This is described in \Cref{sec:Clover3SAT}. 

In \Cref{sec:ConvexVariableGadgets}, we show that we can make variable components in the orthogonally convex setting. The idea is to form parts of the ``boundary'' of the gadgets with squares in the packing rather than with the actual polygon boundary. Unsurprisingly, this make the verification of the variable components much more complicated. One key idea is that we can add redundancy to the constraints that form the variable component, so that one variable component can be verified without having already verified all the others.

Instead of having isolated OR gadgets, we use much more complicated \emph{clause components}. The clause components are described in \Cref{sec:ClauseGadgets}. The clause components create a set of constraints that are not obviously equivalent to the \TSAT instance. Verifying the clause gadgets requires a detailed algebraic analysis, which is given in \Cref{sec:ClauseGadgetVerification}.

There will be two sections of border between the clause and variable components. These are formed by static rows that we call the \emph{membrane rows}. The clause and variable components are each verified independently, with any interactions controlled by some simple properties of squares in the membrane rows.

\subsection{\CTSAT}\label{sec:Clover3SAT}

The first step in our construction in \Cref{sec:simplepacking} for simple polygons was to convert a planar graph to a schematic that has a row for each variable and clause and a column for each edge, as shown in \Cref{fig:ConversionSchematic2}. In general, it isn't possible to convert an instance of \PTSAT\ into a schematic that is suitable for generating an orthogonally convex polygon. Instead, we will have to consider a further restriction of \PTSAT. In this section, we define the problem \CTSAT, show that it is NP-hard, and show that an instance can be converted into a schematic that has an orthogonally convex shape.

An instance of \emph{\CTSAT} consists of the following:

\noindent
\textbf{Input:}
An instance $\Phi$ of 3SAT containing variables $x_1,\dots, x_n$, two sets of clauses $c_1, \dots, c_p$ and $d_1, \dots, d_q$, and a planar embedding of the graph $G$ that contains a vertex for each $x_i$, $c_i$ or $d_i$ and:

\begin{itemize}
    \item An edge $(x_i, c_j)$ whenever $x_i$ or $\neg x_i$ appears in $c_j$, and similarly for the clauses $d_j$.
    \item Edges $(x_i, x_{i+1})$ for $i=1,\dots, n-1$, $(c_i, c_{i+1})$ for $i=1,\dots, p-1$, and $(d_i, d_{i+1})$ for $i=1,\dots, q-1$. 
    \item Edges $(x_1, c_1)$, $(c_1, x_n)$, $(x_n, d_1)$, and $(d_1, x_1)$.
\end{itemize}

\noindent
\textbf{Question:}
Is $\Phi$ satisfiable?

We say that an instance of \CTSAT is \emph{monotone} if each clause $c_i$ is positive (i.e.,~it is of form $x_j \vee x_k \vee x_\ell$) and each clause $d_i$ is negative (i.e.,~it is of form $\neg x_j\vee \neg x_k \vee \neg v_\ell$).
The problem \emph{\MCTSAT} is \CTSAT restricted to monotone instances. 

\CTSAT is a restriction of (variable-linked) \PTSAT\ with an additional constraint that it should also be possible to link some of the clauses as well. Figure \ref{fig:Clover3SAT} shows an example of the ``clover'' graph that appears in this definition.

\begin{figure}
\centering  
\includegraphics[page=1]{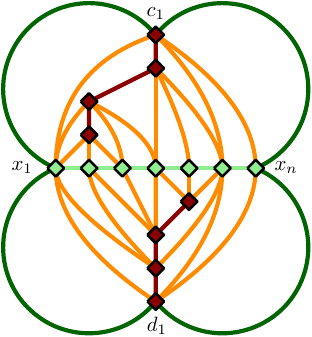}
\caption{A graph coming from an instance of \CTSAT. The $4$ outer edges $(x_1, c_1)$, $(c_1, x_n)$, $(x_n, d_1)$, and $(d_1, x_1)$ exist to prevent constraint edges from wrapping around to the other side of the variables. When drawn like this, all the edges for $c$ constraints hit the variable vertices from above and all the edges for $d$ constraints hit variable vertices from below.
}
\label{fig:Clover3SAT}
\end{figure}

In 2018, Pilz \cite{PilzLayeredPlanar3SAT} showed that \PTSAT\ remains hard if there is a cycle that passes through all the clause vertices and then all of the variable vertices. This is called \textsc{Variable-Clause-Linked-Planar-3SAT}.
The reduction is from \PTSAT, but unlike either variable-linked or clause-linked \PTSAT, the reduction requires modifying the original graph. By reducing instead from \MPTSAT, Pilz also shows that this problem is hard when all the edges inside the cycle represent positive literals and all the edges outside the cycle represent negative literals.

\textsc{Variable-Clause-Linked-Planar-3SAT} is closely related to \CTSAT, and the proof of hardness follows Pilz \cite{PilzLayeredPlanar3SAT} closely. In particular, both proofs use something like what we call \LPTSAT\ as an intermediate step.
An instance of \LPTSAT\ consists of:

\noindent
\textbf{Input:} An instance $\Phi$ of \TSAT\ and a planar, integer-coordinate straight line embedding of the incidence graph $G$ where the variable vertices have even $y$ coordinates, the clause vertices have odd $y$ coordinates, and each edge is between a pair of vertices whose $y$-coordinates differ by $1$.

\noindent
\textbf{Question:} Is $\Phi$ satisfiable?

We call a problem instance \emph{monotone} if clause vertices with $y$ coordinates of the form $4k+1$ are positive and clause vertices with $y$ coordinates of the form $4k-1$ are negative. The problem \LPTSAT\ restricted to monotone instances is called \emph{\MLPTSAT}. 

\begin{lemma}
\MLPTSAT\ is NP-hard.
\end{lemma}

\begin{proof}
Our proof closely follows Pilz~\cite{PilzLayeredPlanar3SAT}. We reduce from an instance $\Phi$ of \MPTSAT.

We draw the incidence graph with integer-coordinate vertices, straight edges, and no crossings.
Since we have a cycle through the variable-vertices, we can draw this graph in such a way that the variable-vertices all have $y$-coordinate $0$, the clause vertices have odd $y$-coordinates, and the positive clauses have positive $y$-coordinates while the negative clauses have negative $y$-coordinates.
The size of the coordinates needed is no more than polynomial in the size of $\Phi$.

Next, we split each edge whenever it crosses a layer. The process of splitting edges will (temporarily) cause some clauses to contain both negated and non-negated literals. In order to end up with a monotone instance, the edge splitting process will preserve the following properties:

\begin{itemize}
    \item Variable-vertices have even $y$-coordinates
    \item Clause-vertices have odd $y$-coordinates
    \item If a variable-vertex has $y$ coordinate $i\equiv 0\pmod 4$ and is connected to a clause vertex with $y$-coordinate $k$, then the sign of that variable in that clause is $\text{sign}(k-i)$
    \item If instead the variable vertex has coordinate $i\equiv 2 \pmod 4$, then the sign of that variable in that clause is $\text{sign}(i-k)$
\end{itemize}

This is satisfied by the initial configuration since the variable-vertices initially all have $y$-coordinate $0$. Whenever an edge spans a $y$-distance of more than $1$, we can split this edge in a way that preserves this property. \Cref{fig:EdgeSplitting} shows the $4$ cases that can occur depending on the $y$-coordinate of the variable vertex and whether the clause vertex is above or below it. For example, suppose a variable-vertex $v$ has $y$-coordinate $i\equiv 0 \pmod 4$ and is connected to a clause-vertex $c$ with $y$-coordinate $k>i+1$. We add a variable $u$ with $y$-coordinate $i+2$ and a clause $d=(v\vee u)$ with $y$ coordinate $i+1$. Replace $v$ in $c$ with $\neg u$. The new clause $d$ is equivalent to $\neg u\implies v$, so the new \TSAT\ instance is equivalent, and we have reduced the the total number of times that an edges crosses a layer by $2$. The other cases are shown in \Cref{fig:EdgeSplitting}.

\begin{figure}
\centering
\includegraphics[page=3]{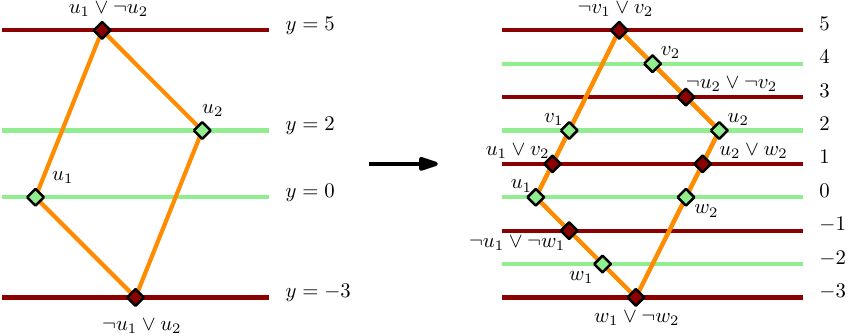}
\caption{Splitting edges to form a layered planar graph. For our purposes, a \TSAT\ instance has \emph{at most} 3 vertices per clause.}
\label{fig:EdgeSplitting}
\end{figure}

We repeat this process until all edges are between adjacent layers. If a clause has $y$-coordinates of form $4k+1$ then variables from both adjacent layers must appear non-negated, and if a clause has $y$-coordinates of form $4k-1$ then variables from both adjacent layers appear negated. So the result is indeed a monotone instance.
\end{proof}

These layered planar graphs appear in the proof by Pilz~\cite{PilzLayeredPlanar3SAT}. They are then wrapped up in a spiral to obtain a graph with the variable-clause-linked property. To obtain a clover graph, we wrap a layered planar graph in a spiral in a slightly different way.

\begin{lemma}
\MCTSAT\ is NP-hard. 
\end{lemma}

\begin{proof}
We can add edges to a layered planar graph from an instance of \MLPTSAT\ to turn it into a clover graph, as shown in \Cref{fig:CloverWrapping1} (left). This can be illustrated more clearly by wrapping the layered planar graph graph around in a spiral, as shown in \Cref{fig:CloverWrapping1} (right). One of the clause paths goes through all the positive clause vertices and the other goes through the negative clause vertices, so this produces a monotone instance of \CTSAT.
\end{proof}

\begin{figure}
\centering
\includegraphics[page=2]{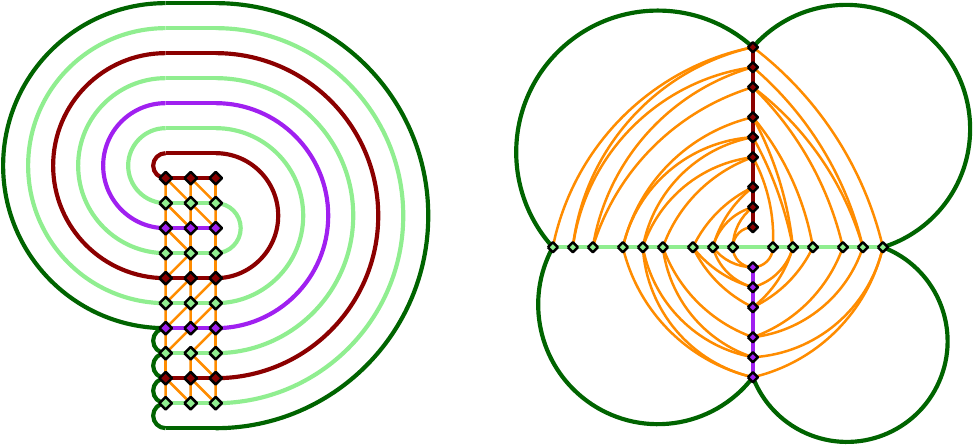}
\caption{Left: If we order the vertices appropriately, then we can add the extra edges to convert a layered-planar graph into a clover graph. The positive clauses are shown in red and the negative clauses are shown in purple.
Right: The same graph can be put into a form like that in Figure \ref{fig:Clover3SAT}. The layered planar graph is wrapped around the center in a spiral.
}
\label{fig:CloverWrapping1}
\end{figure}

The planar embedding of the clover graph does not appear directly in our construction. Instead, we use the embedding to extract some combinatorial structure from the graph. The conversion from a clover graph to a schematic for the orthogonally convex construction is shown in \Cref{fig:ConvexSchematic}. \Cref{lem:ConvexConversion} describes the properties that we will need.

\begin{figure}
\centering
\includegraphics[page=1]{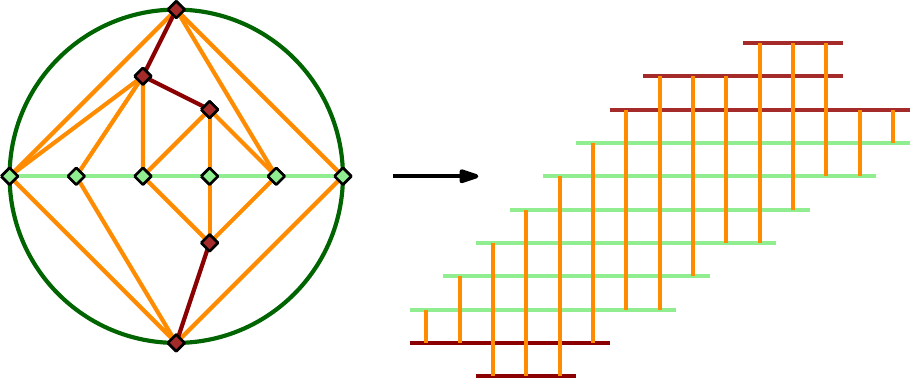}
\caption{Converting a clover graph to a schematic. Note that the order of the clauses is reversed---the innermost clauses in the graph become the outermost clauses in the schematic.}
\label{fig:ConvexSchematic}
\end{figure}

\begin{lemma}\label{lem:ConvexConversion}
An instance of \CTSAT\ can be used to produce a schematic that is orthogonally convex in the following sense:
each $y$-coordinate represents a variable or clause row, with variable rows in the middle and sets of clause rows above and below the variable rows.
Each of the upper clause rows covers a strictly smaller range of $x$-coordinates than the row below it, and each of the lower clause rows covers a strictly smaller set of $x$-coordinates than the row above it.
If a variable row spans the $x$ coordinates $[a, b]$, then the one above it spans $x$-coordinates $[c, d]$ with $a < c < b < d$.

If the instance is monotone, then the lower set of clause rows represent negative clauses and the upper set of clause rows represent positive clauses. Furthermore, we have the following property: 

Suppose there is a clause $y_i\vee y_j\vee y_k$ and the literal column for $y_i$ passes through the row for another clause $c$, then the literal columns for $y_j$ and $y_k$ must also pass through the clause row for $c$ (here the $y_i$ are literals that could be of the form $x_j$ or $\neg x_j$).
\end{lemma}

\begin{proof}
Start with an instance of \CTSAT\ with vertices representing variables $x_1,\dots, x_n$ and clauses $c_1, \dots, c_p$ and $d_1, \dots, d_q$. The $i$th variable row from the bottom represents the variable $x_i$. The structure of the variable rows is straightforward, but we should check that it is possible to arrange the clause rows appropriately. 

In our schematic, the row for $c_{i+1}$ is directly above the row for clause $c_{i}$ (and similarly $d_{i+1}$ has a row below the row for $d_{i}$). We would like to show that it is possible to define, for a clause $c_i$, an interval $I_i\subseteq \{1,\dots, n\}$ such that:

\begin{itemize}
    \item $I_{i+1}\subset I_i$
    \item If $c_i$ has  variables $x_j, x_k$ and $x_\ell$, then $j,k,\ell\in I_i$ and each of $j,k,\ell$ is either not in $I_{i+1}$ or is an endpoint of $I_{i+1}$
\end{itemize}

We show this by induction on $i$. First, we set $I_1=\{1,\dots, n\}$. We will also inductively define regions $R_i$ in the plane. Set $R_1$ to be the region bounded by the path through the variable vertices in $G$ and the edges $(x_1, c_1)$ and $(c_1, x_n)$.

Let $I_i=[a, b]$. Say the variables in $c_i$ are $x_j, x_k$ and $x_\ell$ with $j\le k\le \ell$. The indices $i, j$ and $k$ are in $I_i$ by inductive assumption. In the graph $G$, $c_i$ is represented by a vertex that has three ``legs''. These legs and the link from $c_{i-1}$ divide $R_{i}$ into four subregions (if $i=1$, then $c_i$ is already on the boundary of $R_i$, so the $3$ legs are enough to produce $4$ regions). The clause $c_{i+1}$ must be in one of these regions, along with all the further clauses (see \Cref{fig:CloverRegions}). Each of these regions is adjacent to a variable vertices with indices in one of the intervals:

\[[a, j], [j, k], [k, \ell], [\ell, b]\]

\begin{figure}
\centering
\includegraphics[page=4]{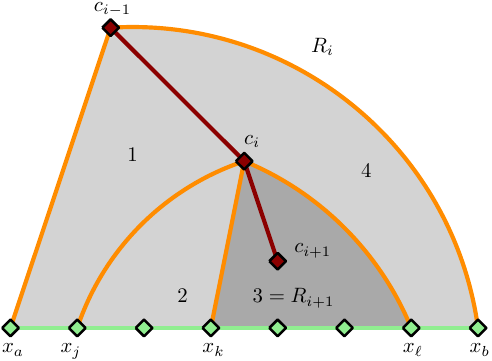}
\caption{The region $R_i$ is separated into $4$ smaller regions, one of which is $R_{i+1}$.}
\label{fig:CloverRegions}
\end{figure}

So define $R_{i+1}$ to be the region that contains $c_{i+1}$ and define $I_{i+1}$ to be the corresponding interval. The indices $i$, $j$, and $k$ are all either outside of $I_{i+1}$ or an endpoint of it. Since the edges corresponding to the literals in $c_{i+1}$ can't cross boundary of the region $R_{i+1}$, we see that $I_{i+1}$ contains the indices of all the variables that appear in $c_{i+1}$.

These intervals can then be used to draw the clauses $c_i$ in our schematic. Since the intervals are nested, this schematic is orthogonally convex. We can do the same thing for the $d_i$ (only upside-down).

In a monotone instance, the $c_i$ are positive and the $d_i$ are negative, so all the upper clauses are positive and all the lower clauses are negative.

For the last part, note that all the literal columns for clause $c_i$ pass through the row for clause $c_j$ for each $j<i$ and don't intersect the row for clause $c_{k}$ when $k$ is $>i$ (similarly for the $d_i$). Also, a literal column for a clause $c_i$ can never intersect the row for a clause $d_k$ (and visa-versa).
\end{proof}

In analogy with the common definition of \MPTSAT, we have so far put the positive clauses above the variables and the negative clauses below the variables. This is of course arbitrary. There are several choices like this involved in our construction, and it seems to us that it is not possible to simultaneously take the most natural choices for all of them. For this reason, we will actually use rotated schematics where the negative clauses are above the variables and the negative clauses are below.

\subsection{Stacks, reference centers, and crossings}

In the reduction of \Cref{sec:simplepacking}, the reference centers formed a regular grid.
In the orthogonally convex case, a few rows will have reference centers of form $[2k, 2k+1]\times [2\ell-1, 2\ell]$ instead of $[2k-1, 2k]\times [2\ell-1, 2\ell]$.
The following lemma generalizes \Cref{lem:referenceCenters}.

\begin{lemma}\label{lem:ConvexReferenceCenters}
Suppose that $P$ is grid polygon, $I\subset \mathbb{Z}$ and let the reference centers be squares in $P$ of form $[2k-1, 2k]\times [2\ell-1, 2\ell]$ for $\ell\in I$ and $[2k, 2k+1]\times [2\ell-1, 2\ell]$ for $\ell\notin I$.
Then any $2\times 2$ square in $P$ that has integer coordinates contains exactly one reference center.
\end{lemma}

\begin{proof}
A $2\times 2$ square with integer coordinates has a span of $y$-coordinates that overlaps with one row containing reference centers.
As the square has width $2$, it contains a single reference center in that row.
\end{proof}

The color scheme used in \Cref{sec:simplepacking} is not so useful for the orthogonally convex gadgets. The figures in this section will use a looser color-coding of the reference centers based on the function that each square serves in the construction.

Next, we introduce the concept of a \emph{stack}. This is a generalization of the pyramids from \Cref{sec:simplepacking}. 

A \emph{vertical stack} is a maximal connected set of squares that share their vertical alignment. A \emph{horizontal stack} is a maximal connected set of squares sharing horizontal alignment. We say that a vertical stack \emph{pushes up} if all the squares in the stack push up, and \emph{pushes down} otherwise. Similarly, a horizontal stack \emph{pushes right} if all the squares push right, and \emph{pushes left} otherwise.

Most of our reference centers are of the form $[2k-1, 2k]\times [2\ell-1, 2\ell]$, so in general a vertical stack always pushes either up or down and a horizontal stack pushes either right or left. However, each variable component will require exactly one row where the reference centers are of the form $[2k, 2k+1]\times [2\ell-1, 2\ell]$. A horizontal stack containing these squares can't be given a left/right direction in general. We will avoid discussing horizontal stacks that contain these squares.

A \emph{static row} is a row of squares where an edge of the polygon pushes the rightmost square to the left or pushes the leftmost square to the right.
All the static rows will have reference centers of the form $[2k-1, 2k]\times [2\ell-1, 2\ell]$.
A \emph{left static row} has squares that push left and a \emph{right static row} has squares that push right.
If a pair of neighboring static rows have opposite horizontal alignments, we call this a horizontal \emph{static shift}. 

In \Cref{sec:simplepacking}, we saw that a pyramid (i.e., a vertical stack) always grew when passing between two unaligned rows.
The following lemma makes a similar claim.

\begin{lemma}\label{lem:StaticShift}
Any vertical stack passing through a static shift grows in width by at least $1$ square.
\end{lemma}

\begin{proof}
Without loss of generality, we consider a vertical stack that pushes up.
If the stack has width $k$ in the row just below the static shift, then clearly, the stack pushes up $k+1$ squares in the row just above the static shift.
\end{proof}

We also need to understand the situation where a vertical stack crosses a horizontal stack.
\Cref{fig:StackCrossing} shows $4$ different ways in which this can happen.
As expressed by the following lemma, such a crossing always makes the stacks grow in width by at least $2$ squares in total.

\begin{figure}
\centering
\includegraphics[page=1,angle=90,width=\textwidth]{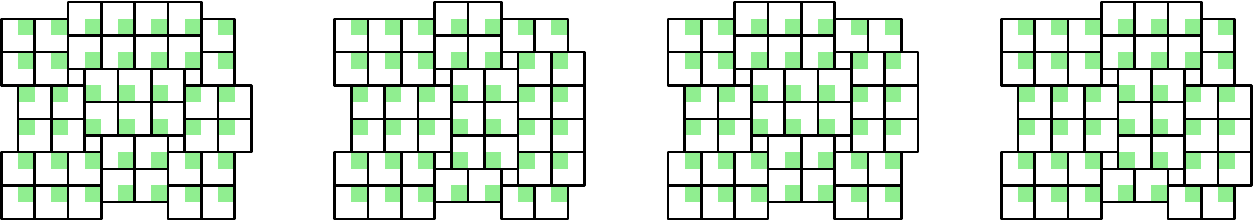}
\caption{
A vertical stack crossing a horizontal stack.
First: The horizontal stack maintains the same width, while the vertical stack increases in width by $2$ squares.
Second: The vertical stack maintains the same width while the horizontal stack grows by $2$ squares.
Third:
Each stack grows by $1$ square.
Fourth:
A different way for each stack to grow by $1$ square.
}
\label{fig:StackCrossing}
\end{figure}

\begin{lemma}\label{lem:StackGrowth}
Consider a set of reference centers with upper right corners $(2x,2y)$ for $x=0,\ldots,x_0+1$ and $y=0,\ldots,y_0+1$.
Consider a packing of $2\times 2$ squares where each of these reference centers is covered by a square, and let $\langle x,y\rangle$ denote the square covering the reference center with upper right corner $(2x,2y)$.
Suppose that
\begin{itemize}
\item $\langle 1,0\rangle,\ldots, \langle x_0,0\rangle$ push up, and

\item $\langle 0,1\rangle,\ldots, \langle 0,y_0\rangle$ push right.
\end{itemize}
Then it holds that
\begin{enumerate}
\item $\langle 1,y_0+1\rangle,\ldots, \langle x_0,y_0+1\rangle$ push up, \label{enum:1}

\item $\langle x_0+1,1\rangle,\ldots, \langle x_0+1,y_0\rangle$ push right, \label{enum:2}

\item $\langle x_0+1,y_0+1\rangle$ pushes up or right, and \label{enum:3}

\item $\langle 0,y_0+1\rangle$ pushes up or $\langle x_0+1,0\rangle$ pushes right. \label{enum:4}
\end{enumerate}

In conclusion, when a vertical stack crosses a horizontal stack in a region where the reference centers all have the form $[2k-1, 2k]\times [2\ell-1, 2\ell]$, then the two stacks together grow in width by at least $2$ squares.
\end{lemma}

\begin{figure}
\centering
\includegraphics[page=2]{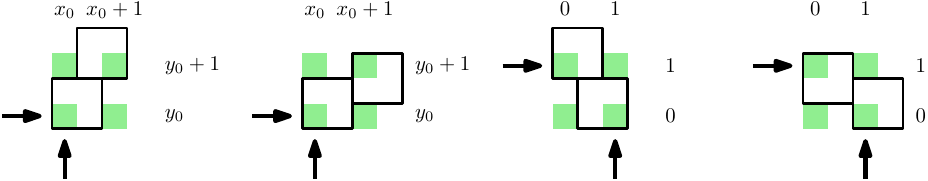}
\caption{Two left figures: The square $\langle x_0+1,y_0+1\rangle$ pushes up or right.
Two right figures: The square $\langle 0,1\rangle$ pushes up or $\langle 1,0\rangle$ pushes right.}
\label{fig:CrossingLemma}
\end{figure}

\begin{proof}
Statements \ref{enum:1} and \ref{enum:2} follow since the reference centers form a regular grid.
\Cref{fig:CrossingLemma} (left) shows that \ref{enum:3} follows.
\Cref{fig:CrossingLemma} (right) shows that $\langle 0, 1\rangle$ pushes up or $\langle 1,0\rangle$ pushes right.
Then statement \ref{enum:4} also follows.
\end{proof}

The concept of a stack was inspired by the work by El-Khechen, Dulieu, Iacono and van Omme \cite{DBLP:conf/cccg/El-KhechenDIO09} showing that these packing problems are in NP even when the coordinates of the polygon can be exponentially large. The methods in \cite{DBLP:conf/cccg/El-KhechenDIO09} seem to suggest that there could be several other types of crossings, but all of these would leave a reference center uncovered. These can be disregarded as long as the reference tiling is a regular grid. We will have to be more careful for stacks that pass through rows with differently-aligned reference centers.

\subsection{Verification order and membrane rows}

\begin{figure}
\centering
\includegraphics[page=1]{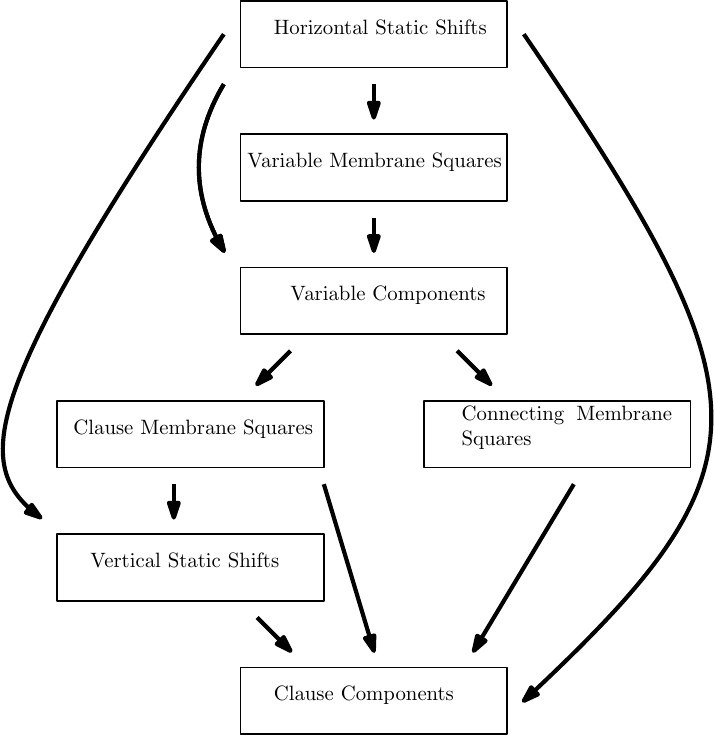}
\caption{Logical dependence of the different parts of the construction.}
\label{fig:Flowchart}
\end{figure}

We are given a formula $\Phi$ of \MCTSAT and describe how to construct an equivalent instance of \USPack with an orthogonally convex polygon.
This boils down to describing variable and clause components.
However, it won't be possible to fully verify the clause components until after the variable components are verified. Verifying the variable components, on the other hand, requires determining a few details about packings of squares in the clause components. The purpose of this section is to explain how the variable components interact with the clause components, allowing us to isolate the verification of the variable component from the verification of the clause component.

The variable components are separated from the clause components by two special static rows that we call the \emph{membrane rows}.
There are two such rows, one above and one below the variable components.
These separate the variable components from the upper and lower clause components, respectively. 
All the dependence of the variable components on the clause components is through these rows---the variable components don't (directly) depend on squares in the clause components.
Similarly, all the dependence of the clause components on the variable components is through these rows.
We say a square in one of the membrane rows \emph{pushes in} if it pushes towards the variable components---that is, it is in the upper membrane row and pushes down or is in the lower membrane row and pushes up.
A square in one of the membrane rows \emph{pushes out} if it does the opposite---that is, it pushes \emph{away} from the variable components.
We divide the squares in the membrane rows into $4$ types depending on the role that they play in the construction.
Each square in either of the two membrane rows will be one of these types.
The types are as follows:

\begin{itemize}
    \item \emph{Variable membrane squares} are squares that need to always push in order to verify the variable components.
    \item \emph{Connecting membrane squares} are squares that may or may not be allowed to push in depending on the position of a variable component. 
    \item \emph{Clause membrane squares} are squares that need to push out in order to verify the clause components.
    \item \emph{Spacing membrane squares} are all the other squares on the membrane rows. We will describe satisfying assignments that have these pushing out, but the verification doesn't depend on this; there could be some slack in these squares.
\end{itemize}

The variable membrane squares and connecting membrane squares are organized into \emph{blocks}, each containing some number of adjacent squares. The clause membrane squares are exactly the squares directly adjacent to a block of variable membrane or connecting membrane squares.

We will verify the variable components first, then verify the clause components (see \Cref{fig:Flowchart}). In particular, we will only know that the clause membrane squares push out \emph{after} verifying the variable components. However, we can show that the variable membrane squares push in already.
Each block of variable membrane squares is part of a vertical stack created in a clause component above or below the variable components.
An edge of the polygon pushes a square in, which creates the stack.
By \Cref{lem:StaticShift}, the vertical stack grows each time it passes through a static shift. By the time the stack reaches the membrane row, it has grown to the full width of the block of variable membrane squares, and so all those squares push in. \Cref{fig:VariableGadgetSquares} (left) shows a simplified example of this.
In order to be able to make this construction, we just need to make sure that each block of variable membrane squares has a number of squares equal to $1$ plus the number of static shifts above it (or below it for a block of variable membrane squares in the lower membrane row).

\begin{figure}
\centering
\includegraphics[page=1]{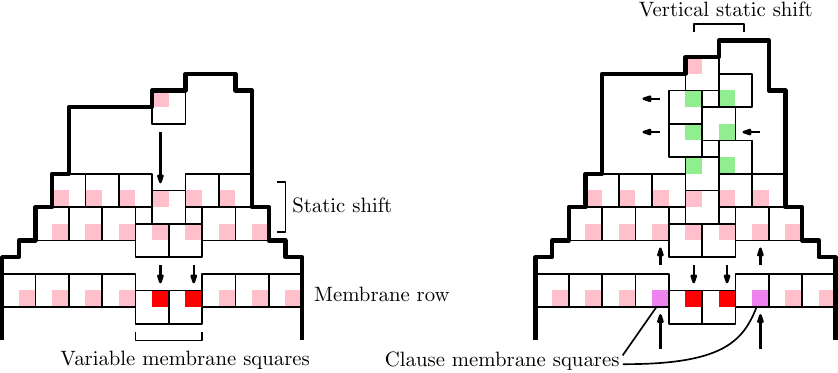}
\caption{Left: The variable membrane squares always push down.
Right: Each block of variable membrane squares has clause membrane squares on either side, creating a vertical static shift on each side.}
\label{fig:VariableGadgetSquares}
\end{figure}

\begin{lemma}\label{lem:VariableGadgetSquares}
The variable membrane squares push in.
\end{lemma}

\begin{proof}
Each block of variable membrane squares in the upper membrane row has width equal to $1$ plus the number of static shifts above it. Each block of variable membrane squares in the lower membrane row has width equal to $1$ plus the number of static shifts below it. 
At the top or bottom of the polygon, a section of polygon boundary forces one square to push in for each block of variable membrane squares. So by \Cref{lem:StaticShift}, the variable membrane squares push in.
\end{proof}

Each block of variable membrane squares will have a pair of clause membrane squares on either side.
This creates \emph{vertical static shifts} between the stack containing the variable membrane squares and the squares next to it.
\Cref{fig:VariableGadgetSquares} (right) shows a simplified example of this.

\begin{lemma}\label{lem:VerticalStaticShift}
If the clause membrane squares push out, then any horizontal stack that doesn't contain a static row and passes through a vertical static shift grows in size by $1$ square.
\end{lemma}

\begin{proof}
By \Cref{lem:StaticShift}, the stack containing the variable membrane squares grows at each static shift (moving inwards). By the same lemma, the two adjacent stacks created by the clause membrane squares grow at each static shift when \emph{moving out}. So the squares on opposite sides of a vertical static shift must have opposite vertical alignments, and the result follows by the same proof as \Cref{lem:StaticShift}.
\end{proof}

We are now ready to describe the variable components. 

\subsection{Variable components}\label{sec:ConvexVariableGadgets}

The purpose of this section is to describe the variable components and show how they can be packed in a satisfying assignment. In \Cref{sec:VariableGadgetVerification}, we will verify that any packing of the variable components must correspond to an assignment of variables in the \MCTSAT\ instance.

The variable components used in \Cref{sec:simplepacking} involved two rows that always have alternate alignment, sandwiched between two static rows that are aligned with each other.
Here, we instead use variable components where the two rows always have the same alignment; as shown in \Cref{fig:ConvexVar1}. In order for this to work, the bottom row of reference centers needs to have a different horizontal alignment than all the other reference centers.

\begin{figure}
\centering
\includegraphics[page=1]{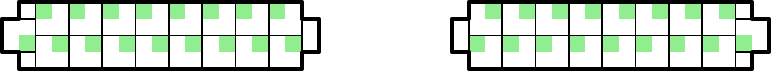}
\caption{The variable rows for orthogonally convex polygons.
The left figure shows what we call the \emph{minus} position and the right figure shows what we call the \emph{plus} position.}
\label{fig:ConvexVar1}
\end{figure}

Similarly as in \Cref{sec:simplepacking}, we designate the two positions of the variable components by plus and minus, with plus corresponding to squares pushing in a counterclockwise direction around the gadget and minus corresponding to squares pushing in a clockwise direction. The advantage of this type of variable component is that parts of its boundary can be formed by other static squares, allowing the boundary of the polygon to remain orthogonally convex. This is shown in \Cref{fig:ConvexVar2}. In order for these to be the only packings of the gadget, we need to know that the two marked squares push in. We call these the \emph{pinch squares}.

\begin{figure}
\centering
\includegraphics[page=2]{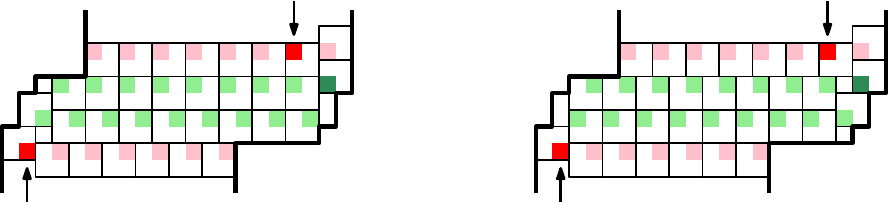}
\caption{The squares with red reference centers are the \emph{pinch} squares. These push in to form part of the boundary of the variable rows.}
\label{fig:ConvexVar2}
\end{figure}

The static rows above and below the variable component have different horizontal alignments, with the upper static rows pushing right and the lower static rows pushing left. As long as the pinch squares push in, the variable component will always be aligned with either the top or bottom row of static squares, so a stack passing through it always grows once, see \Cref{fig:ConvexVar3}.

\begin{figure}
\centering
\includegraphics[page=3]{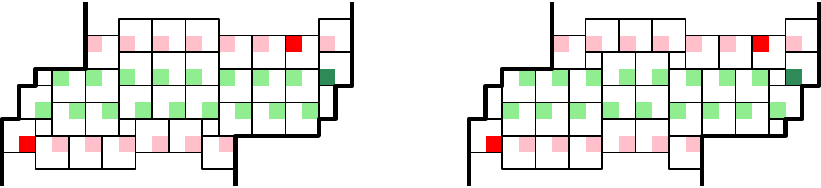}
\caption{A stack grows once when it passes through the variable component.}
\label{fig:ConvexVar3}
\end{figure}

The width of a block of pinch squares depends on the number of static shifts and variable components above it, so we should make sure that these blocks can be exactly as large as they need to be. In order to make sure the clause membrane squares push out, we should also make sure that the squares on either side of each block of pinch squares push out. \Cref{fig:ConvexVar4} shows how these can be accomplished.

\begin{figure}
\centering
\includegraphics[page=4]{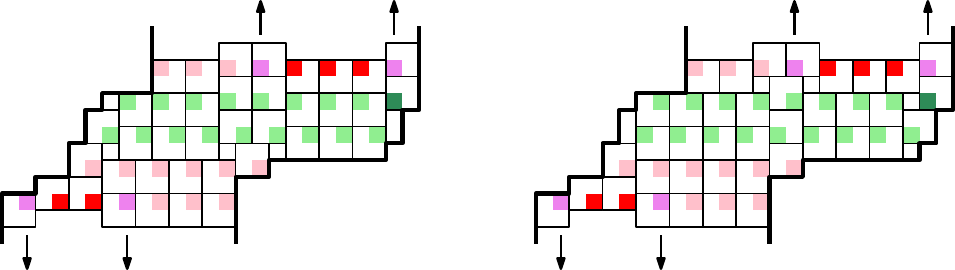}
\caption{The blocks of pinch squares can be arbitrarily wide, and we can ensure that the squares adjacent to each block of pinch squares push out.}
\label{fig:ConvexVar4}
\end{figure}

\subsubsection{Helper rows}\label{sec:HelperRows}

A few rows immediately above and below the variable rows in a variable component are \emph{helper rows}, and these will be necessary for our PUSH gadgets to work.
The PUSH gadgets will be described later in \Cref{sec:ConvexPushers}.
The \emph{upper helper rows} are immediately above the variable rows, and the \emph{lower variable rows} are immediately below the variable rows. A helper row is a row that is sometimes forced to split, with squares on the left of the split pushing left and squares on the right of the split pushing right. The squares in the helper rows are shown with light blue and dark green reference centers in our diagrams. The helper rows will satisfy the follow:

\begin{itemize}
    \item If the variable component is plus, then the squares in the upper helper rows push away from the helper row gadgets. That is, the squares in the upper helper rows that are to the right of the helper row gadgets (dark green reference centers) push right, and the squares in upper helper rows that are to the left of the helper row gadgets (light blue reference centers) push left.
    \item If the variable component is minus, then the squares in the lower helper rows push away from the helper row gadgets. That is, the squares in the lower helper rows that are to the right of the helper row gadgets (light blue reference centers) push right, and the squares in upper lower rows that are to the left of the helper row gadgets (dark green reference centers) push left.
\end{itemize}

When the helper rows don't split, they can have slack, but we will carefully design the variable components so that this slack can't propagate outside the helper rows. \Cref{fig:ConvexHelpers} shows how the helper rows are created. Each helper row gadget creates both an upper and a lower helper row, so the number of upper helper rows is the same as the number of lower helper rows.

In order for the helper rows to be created, the squares marked in \Cref{fig:ConvexHelpers} need to push down. We call these the \emph{helper row creation squares}.
The blocks of helper row creation squares can be any size, so we can always create the helper row gadget no matter how many static shifts and variable components are above. 

Since the polygon should be orthogonally convex, each helper row gadget requires adding an extra static row below it.

\begin{figure}
\centering
\includegraphics[page=5,width=\textwidth]{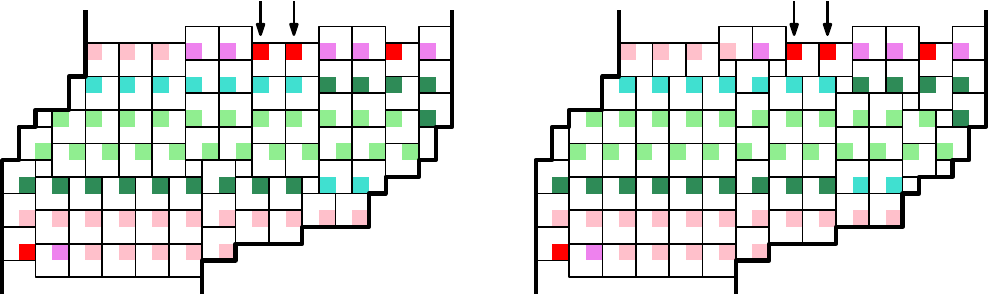}
\caption{The marked squares are the \emph{helper row creation squares}. The rows directly above and below the variable rows are the \emph{helper rows}. We may need a large number of helper rows, but there will always be the same number of upper helper rows as lower helper rows.
}
\label{fig:ConvexHelpers}
\end{figure}

\subsubsection{PUSH gadgets}\label{sec:ConvexPushers}

The variable components should control the blocks of connecting membrane squares in the membrane row.
Each block of connecting membrane squares forms part of a stack that terminates at a \emph{PUSH gadget} in one of the variable components. Some of the blocks of connecting membrane squares should depend on the variable components, but others should always be allowed to push down. So we should create $4$ types of PUSH gadget---these are the PUSH-UP-IF-PLUS, PUSH-UP-NEVER, PUSH-DOWN-IF-MINUS, and PUSH-DOWN-NEVER gadgets. 

The PUSH-UP-IF-PLUS gadget is shown in \Cref{fig:ConvexPushUp}. If the variable component is plus, then it isn't possible for \emph{all} the marked squares to push down. This will prevent all the squares in some block of connecting membrane squares in the upper membrane row from pushing down.

\begin{figure}
\centering
\includegraphics[page=16]{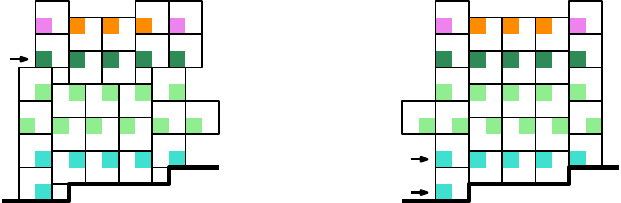}
\caption{The PUSH-UP-IF-PLUS gadget.}
\label{fig:ConvexPushUp}
\end{figure}

Since the polygon is orthogonally convex, we need to add an extra row below the gadget for spacing. The part of the PUSH gadget that interacts with the polygon boundary should carry some information about the position of the variable component, so it isn't possible to use a static row for this purpose. This is why we need the helper rows.
There will be helper rows on either side of the variable component, with the lower helper rows ``helping'' to space out the PUSH-UP gadgets and the upper helper rows helping with the PUSH-DOWN gadgets. The squares in the helper rows are shown with light blue and dark green reference centers in our figures.

The PUSH-DOWN-IF-MINUS gadget is just a mirror image of the PUSH-UP-IF-PLUS, and is shown in \Cref{fig:ConvexPushDown}. If the variable component is minus, then the connecting membrane squares can't all push up.

\begin{figure}
\centering
\includegraphics[page=17]{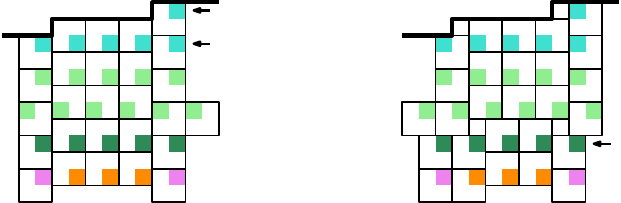}
\caption{The PUSH-DOWN-IF-MINUS gadget.}
\label{fig:ConvexPushDown}
\end{figure}

Some of the blocks of connecting membrane squares not supposed to depend on variable components (see \Cref{sec:ClauseGadgets}).
The squares in these blocks should all be able to push in for either position of the variable gadget. 
The clause membrane squares adjacent to this block still need to push out.
This is accomplished with a PUSH-NEVER gadget.
The PUSH-UP-NEVER gadget is shown in \Cref{fig:ConvexPushNever}.
Like the true PUSH gadgets, these gadgets each consume a helper row. Note that a PUSH-NEVER gadget may have true PUSH gadgets on both its right and left. The PUSH-UP-NEVER gadget is almost identical to the PUSH-UP-IF-PLUS gadget, but there is an extra $1\times 1$ square of space that allows the gadget to not push when the variable component is plus.

\begin{figure}
\centering
\includegraphics[page=2]{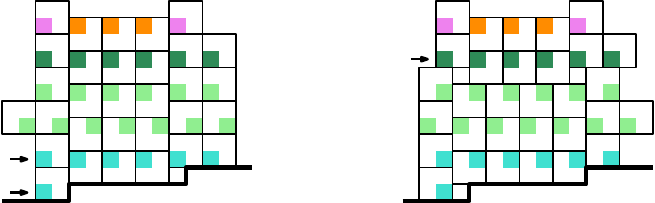}
\caption{The PUSH-UP-NEVER gadget. The PUSH-DOWN-NEVER gadget is a mirror image.}
\label{fig:ConvexPushNever}
\end{figure}

\subsubsection{Redundancy columns}

Each variable component relies on some squares from below pushing up and some squares from above pushing down.
The main part of the verification of the variable components is done working from top-to-bottom, so a variable component may depend on the variable components above it, but shouldn't depend on the ones below it.

The first step in verifying a single variable component is to ensure that the two variable rows share their alignment. This is the only step that requires some squares from below to push up. Observe that there are two ``bad'' cases where the rows do not share an alignment. These are shown in \Cref{fig:ConvexVar5}.

\begin{figure}
\centering
\includegraphics[page=15]{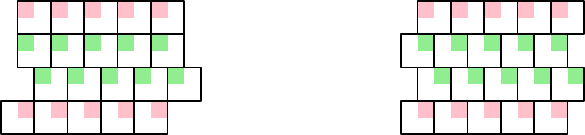}
\caption{The variable rows should always be aligned with each other. These are the two cases that need to be ruled out by some gadgets.}
\label{fig:ConvexVar5}
\end{figure}

These are not symmetric.
The case where both push right is excluded by knowledge about the squares packed in parts of the polygon above this variable component.
The case where both push left must be ruled out by other means.
Here, we use that a vertical stack crossing the unaligned variable rows would grow more in width than if the rows were aligned.
We utilize this insight by introducing vertical stacks called \emph{redundancy columns}.
This is shown in \Cref{fig:RedundancyColumns1,fig:RedundancyColumns2}.

\begin{figure}
\centering
\includegraphics[page=8]{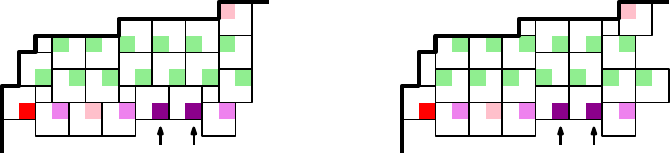}
\caption{The dark purple squares are \emph{redundancy squares}, which create a \emph{redundancy column}. The redundancy column prevents both variable rows from pushing left.}
\label{fig:RedundancyColumns1}
\end{figure}

\begin{figure}
\centering
\includegraphics[page=9]{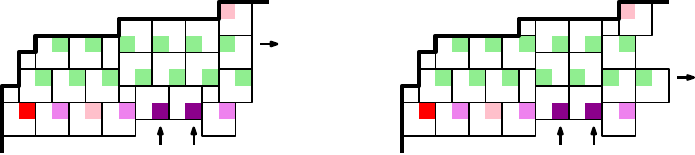}
\caption{Even if the left pinch squares fail to push up, at least one of the variable rows still has to push right because of the redundancy column.}
\label{fig:RedundancyColumns2}
\end{figure}

The redundancy columns need to pass through rows that haven't been verified yet, so we cannot be sure about its width as it reaches this variable component. We \emph{can} have a large number of redundancy columns.
As long as one of these columns has full width by the time it hits the variable component, it will enforce the necessary constraint. In \Cref{sec:VariableGadgetVerification}, we will show that if the total number of redundancy columns is larger than the number of helper rows and variables rows below the variable component, then at least one of the redundancy columns needs to reach full width. 

\Cref{fig:RedundancyColumns3} shows how to pack the polygon when there are multiple redundancy columns and some helper rows.
Each redundancy column requires adding an extra static row above the upper helper rows for spacing reasons.

\begin{figure}
\centering
\includegraphics[page=10]{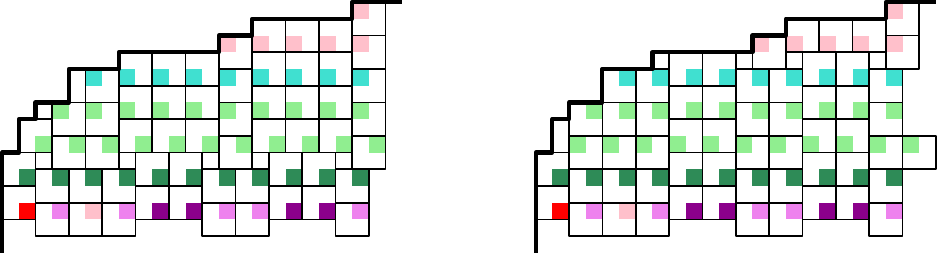}
\caption{Multiple redundancy columns.}
\label{fig:RedundancyColumns3}
\end{figure}

The redundancy columns only help to verify the part of the variable component that is to the right of the redundancy columns. The PUSH-DOWN gadgets will be to the \emph{left} of the redundancy columns, so we do need to eventually verify that the pinch square on the left pushes up. The left pinch squares are verified \emph{after} the right parts of \emph{all} the variable components have been verified.

In summary, the strategy for verifying the variable components is to first verify the parts of each gadget that are \emph{right} of the redundancy columns working from top-to-bottom, and then to verify the remaining parts of each variable component working from bottom-to-top.

\subsubsection{The complete variable component}

\Cref{fig:VariableGadgetSchematic} shows a schematic of a full variable component. Each variable component has two variable rows and some number of helper rows above and below the variable rows. These are flanked by static rows.
Because the polygon needs to be orthogonally convex, we need several static rows on each side for spacing reasons.
The upper static rows are created in order to make redundancy columns, and the lower are created in order to make helper rows.
The lower static rows push left and the upper static rows push right.
There is a static shift between the lower static rows for this variable component (which push left) and the upper static rows for the next variable component (which push right).

\begin{figure}
\centering
\includegraphics[page=7,width=\textwidth]{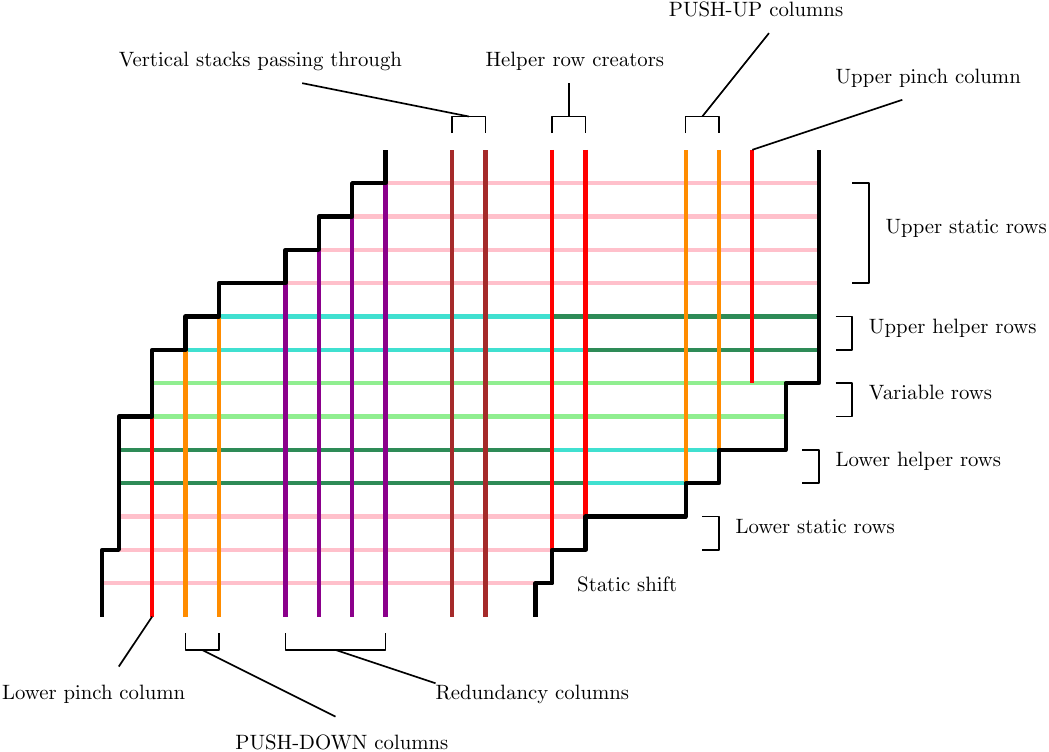}
\caption{A schematic of a variable component.
}
\label{fig:VariableGadgetSchematic}
\end{figure}

There are three types of vertical stacks that are used to verify the variable components. These are pinch columns, helper row creators, and redundancy columns. The upper pinch column is always near the right edge of the gadget and the lower pinch column is at the left edge. We place the redundancy columns as far right as possible and place the helper row creators as far left as possible. So all the PUSH-UP columns are between the helper row creators and the upper pinch column, while the PUSH-DOWN columns are between the lower pinch column and the redundancy columns. 

\Cref{fig:FullVarGadget} shows a complete variable component with helper rows, a PUSH gadget on each side, and one redundancy column.
It is not hard to check from the construction that there exists a packing corresponding to any assignment of variables in the formula $\Phi$, as described by the following lemma.

\begin{figure}
\centering
\includegraphics[page=11,width=\textwidth]{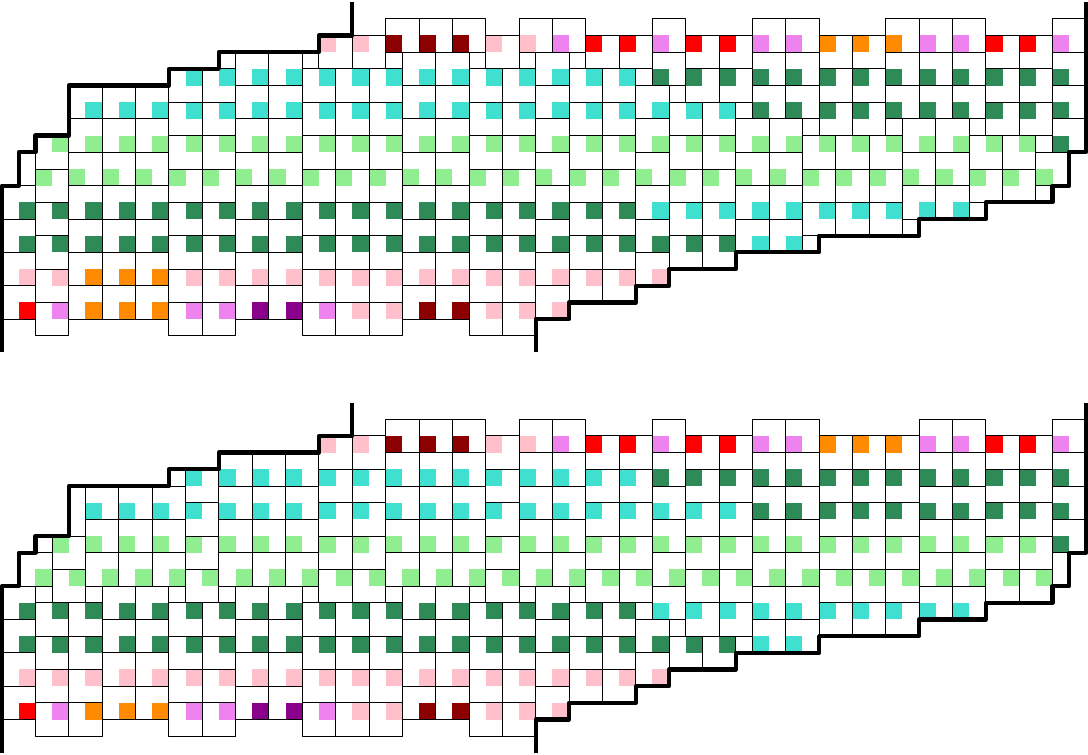}
\caption{A variable component in the plus (top) and minus (bottom) position. There are two pairs of helper rows, one PUSH-DOWN-IF-MINUS gadget, one PUSH-UP-IF-PLUS gadget, one redundancy column, and one vertical stack passing through.}
\label{fig:FullVarGadget}
\end{figure}

\begin{lemma}\label{lem:VariablePackingsExist}
For any assignment of variables, there is a packing of the variable components where the variable membrane squares push in, and:

\begin{itemize}
\item Each square in a block of connecting membrane squares in the upper (resp.~lower) membrane row that is connected to a variable in the minus (resp.~plus) position pushes in.
\item Each square in a block of connecting membrane squares connected to a PUSH-NEVER gadget pushes in.
\end{itemize}
\end{lemma}

\subsection{Verification of the variable components}\label{sec:VariableGadgetVerification}

We now give the verification of the various properties of a variable component. Throughout this section, the square $\langle x, y\rangle$ refers to the the square containing the reference center with upper right corner $(2x, 2y)$.
We choose coordinates so that the squares in the upper variable row are written $\langle x, 0\rangle$ for an integer $x$.
Squares in the lower variable row then have the form $\langle x+\frac 12, -1\rangle$.
We also suppose that there are $n$ helper rows, so the static rows on either side of the variable row are written $\langle x, n+1\rangle$ and $\langle x, -n-2\rangle$.

In order to verify a variable component, we need to be able to show that at least one of the redundancy columns reaches full width. We use the following observation:

\begin{lemma}\label{lem:Redundancy}
Consider two rows of squares with $k$ rows between them.
Let $L$ be the set of squares in the lower row that push up, and let $x_s$ be the $x$-coordinate of the left corners of each $s\in L$.
Let $U$ be the squares of the top row where the $x$-coordinate of the left corners is $x_s-1$, $x_s$ or $x_s+1$ for some $s\in L$.
Then at most $k$ squares in $U$ push down.
\end{lemma}

\begin{proof}
For each row of squares, there are three cases:

\begin{itemize}
    \item All the squares push left. 
    \item All the squares push right. 
    \item Some of the squares on the left push left and the rest push right. 
    There is a \emph{split} separating the two sets of squares. 
\end{itemize}

Each split has width $1$, and the squares (and hence also the splits) have integer coordinates, so each split is directly above at most one of the squares in the bottom row.

Consider the square $s\in L$, covering the range of $x$-coordinates $[x_s, x_s+2]$.
If there is no split with the range $[x_s, x_s+1]$, then all squares covering that range push up, and similarly if there is no split in $[x_s+1, x_s+2]$.
There can be a split in the range $[x_s, x_s+2]$ for at most $k$ squares $s\in L$, one for each of the intermediate layers.
Hence, all but at most $k$ squares in $U$ must push up.
\end{proof}

\begin{corollary}\label{lem:ShiftRedundancy}
Suppose there is a left static row of squares of the form $\langle x, 0\rangle$ and $k+1$ rows above it there is a right static row of squares of the form $\langle x, k+1\rangle$.
Suppose that there are blocks of squares $\langle a_i, 0\rangle$ through $\langle b_i, 0\rangle$ for $i\in \{1, \dots, j\}$ that all push up.
Then for all but $k$ indices, all of the squares $\langle a_i-1, k+1\rangle$ through $\langle b_i, k+1\rangle$ push up.
\end{corollary}

This is illustrated in \Cref{fig:ShiftRedundancy}. Note that \Cref{lem:StaticShift} is a special case of \Cref{lem:ShiftRedundancy} when $k=0$.

\begin{lemma}\label{lem:VariableRedundancy}
Suppose that a variable component (which has possibly not yet been verified) has $m$ helper rows in total. Suppose there are $k$ stacks below the variable component pushing up. Then at least $k-m-2$ of those stacks grow in size by $2$ squares by the time they reach the next variable component above this one.
\end{lemma}

\begin{proof}
There is a left static row below the lowest helper row and a right static row above the highest helper row, with $m+2$ rows in between them. There is a static shift between this variable component and the next one above. So the result follows by \Cref{lem:ShiftRedundancy} and \Cref{lem:StaticShift}.
\end{proof}

We are now ready to begin the verification of a single variable component.
We first show that the variable rows share their alignment to the right of the rightmost redundancy column, as expressed by the following lemma.

\begin{figure}
\centering
\includegraphics[page=14]{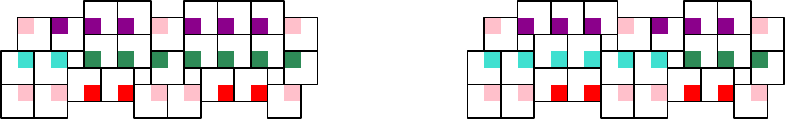}
\caption{There is $1$ slack row between two static rows, so at least one of the vertical stacks needs to grow by $1$ square.}
\label{fig:ShiftRedundancy}
\end{figure}

\begin{lemma}\label{lem:ConvexVar1}
Suppose that the rightmost square in the upper block of pinch squares for a variable component pushes down.
Then the squares in the variable rows that are to the right of the rightmost redundancy column share their horizontal alignment. Precisely, each square in the upper row spans the same $x$-coordinates as a square in the lower row and each square in the lower row spanes the same $x$-coordinates as a square in the upper row (except for the rightmost squares at the end of the variable gadget).
\end{lemma}

\begin{proof}
Together with the pinch squares, the edges of the polygon at the right end of the variable row force the squares in one of the rows to push left, as shown in \Cref{fig:VarRightSide}.

\begin{figure}
\centering
\includegraphics[page=12]{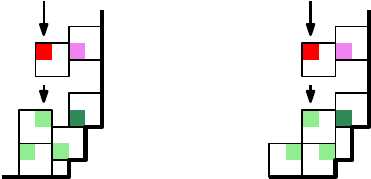}
\caption{If rightmost of the upper pinch squares pushes down, then one of the rows in the variable component pushes left.}
\label{fig:VarRightSide}
\end{figure}

Recall that the number of redundancy columns is larger than the number of variable rows and helper rows in all the gadget below it. By \Cref{lem:VariableRedundancy}, the number of redundancy columns that have full width decreases by $m+2$ when crossing a variable component with $m$ helper rows. So at least one of the redundancy columns must grow by its full size (growing twice at each of the lower variable components, once at the static shift and once when crossing the variable rows).

Let $\langle x, -2-n\rangle$ be the rightmost redundancy square in a redundancy column that reaches full width. So $\langle x, -2-n\rangle$ pushes up, and therefore $\langle x, -2\rangle$ pushes up. The boundary of the polygon forces squares in column $x+1$ push down, so square $\langle x+1, 0\rangle$ pushes down. So one of $\langle x+1, 0\rangle$ or $\langle x+\frac12, -1\rangle$ pushes left, as shown in \Cref{fig:VarRedundancy}.

\begin{figure}
\centering
\includegraphics[page=13]{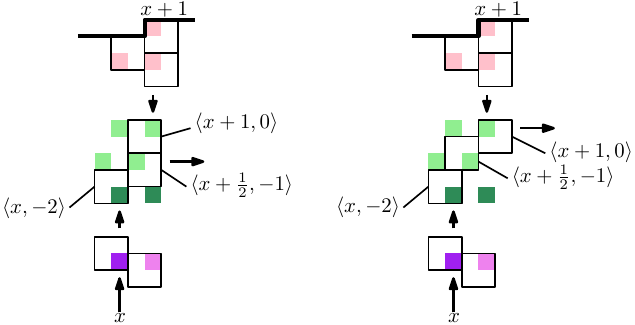}
\caption{If the rightmost square in any block of redundancy squares pushes up, then at least one of the rows in the variable component pushes right (for squares to the right of that square).}
\label{fig:VarRedundancy}
\end{figure}

The conclusion is that, for squares to the right of the rightmost redundancy column, one of the variable rows pushes left and the other pushes right. Since these rows have differently-aligned reference centers, the two rows are aligned with each other.
\end{proof}

When the reference centers form a regular grid and a square pushes up, all the squares in the column above that square must also push up. This isn't \emph{a priori} true for a column passing through a variable component since the reference centers don't form a regular grid. \Cref{fig:ConvexVar3} shows that a similar conclusion \emph{does} still hold as long as the variable rows are aligned.
We formulate this observation as the following lemma:

\begin{lemma}\label{lem:ConvexVar2}
Suppose that a square $\langle x, 0\rangle$ in the upper row of a variable component is aligned with a square (either $\langle x+\frac12, -1\rangle$ or $\langle x-\frac12, -1\rangle$) in the lower variable row. 
If the square $\langle x, 0\rangle$ pushes down, then the square $\langle x, -2\rangle$ also pushes down. If the $\langle x, -2\rangle$ pushes up, then the square $\langle x, 0\rangle$ also pushes up.
\end{lemma}

\begin{proof}
Straightforward.
\end{proof}

When the upper pinch squares push down, \Cref{lem:ConvexVar1} lets us define the plus/minus position of the two variable rows as the position matched by the squares to the right of the rightmost redundancy column.
The following lemma expresses what we need to know about the alignments of the helper rows.

\begin{lemma}\label{lem:ConvexVar3}
Suppose that the upper pinch squares push down and all the helper row creation squares for the variable component push down. The if the variable component is in the plus (resp.~minus) position, then all the squares in the upper (resp.~lower) helper rows push left if they are left of the helper row squares and push right if they are right of the helper row squares.
\end{lemma}

\begin{proof}
We show the claim about the upper helper rows when the squares in the variable rows are plus.
The claim about the lower helper rows for the minus position follows by a similar argument.

Suppose that the variable rows are plus.
We first show the claim that the squares in the helper rows to the left of the leftmost set of helper row creation squares push left.
We proceed by induction, showing that the $k$th upper helper row that are to the left of the $k$th helper row gadget from the right must push left.
The base case is $k=0$---we think of the ``zeroth'' helper row as being the top row of the variable component.

The inductive step is shown in \Cref{fig:HelperVerification1} (left).
Let $\langle x_k, n+1\rangle$ be the rightmost square in the $k$th block of helper row creation squares, which pushes down by assumption.
So the square $\langle x_k, k\rangle$ pushes down also. By inductive assumption, the square $\langle x_k+1, k-1\rangle$ pushes left. This square also pushes up because of the boundary of the polygon (\Cref{lem:ConvexVar2} is used here since this column crosses over the row with differently-aligned reference centers). This means that the square $\langle x_k, k\rangle$ must push to the left, so all the squares in the $k$th helper row to the left of $\langle x_k, k\rangle$ push to the left.
So for $x$-coordinates to the left of the leftmost of the helper row creation squares, the squares in all the upper helper rows push left.

\begin{figure}
\centering
\includegraphics[page=6]{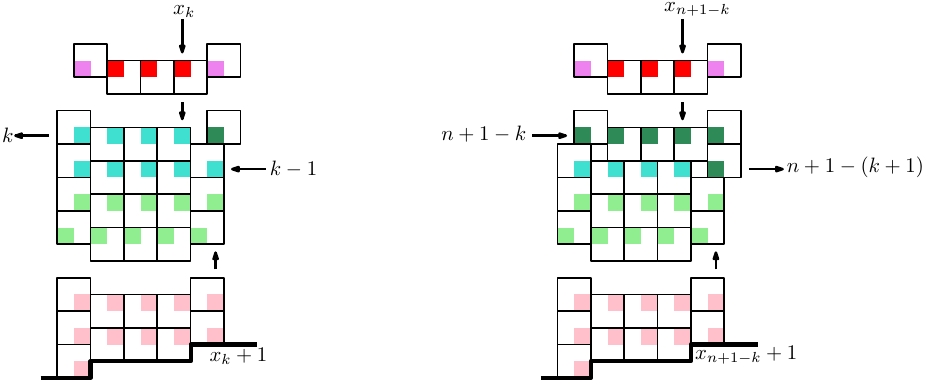}
\caption{Left: The square $\langle x_k, k\rangle$ pushes left.
Right: The square $\langle x_{n+1-k}+1, n+1-(k+1)\rangle$ pushes right.}
\label{fig:HelperVerification1}
\end{figure}

By induction, we now show that the squares in the $(n+1-(k+1))$th helper row that are to the right of the $(n+1-k)$th helper row gadget from the right push to the right. We think of the upper static row as being the $(n+1)$st helper row, so the case for $k=0$ is clear.

The inductive step is shown in \Cref{fig:HelperVerification1} (right). We label the reference centers as before. Now square $\langle x_{n+1-k}, n+1-k\rangle$ pushes down and right, and square $\langle x_{n+1-k}+1, n+1-(k+1)\rangle$ pushes up. So square $\langle x_{n+1-k}+1, n+1-(k+1)\rangle$ pushes right, and all the squares in the $n+1-(k+1)$ helper row that are to the right of $\langle x_{n+1-k}+1, n+1-(k+1)\rangle$ push to the right.

We have now determined the horizontal positions of squares in the upper helper rows when the variable component is plus.
\end{proof}

We are now ready to prove a claim made early in this section and shown in \Cref{fig:ConvexVar3} that a vertical passing through a variable component grows in width by $1$.

\begin{lemma}\label{lem:ConvexVar4}
If all the helper row creation squares and upper pinch squares push down, then any vertical stack passing through the variable component grows in width by $1$.
\end{lemma}

\begin{proof}
Suppose that square $\langle x, -2-n\rangle$ in the below static row pushes up. Any stack that is the left of the redundancy columns will hit the top wall of the polygon instead of passing through the variable component, so we can assume that this square is to the right of the redundancy columns and so the conclusions of \Cref{lem:ConvexVar1,lem:ConvexVar2} can be used. The square $\langle x, n+1\rangle$ pushes up (by \Cref{lem:ConvexVar2}), and we want to show that the square $\langle x-1, n+1\rangle$ also pushes up. There are two cases to consider. 

If the variable component is plus, then the square $\langle x, n\rangle$ pushes left by \Cref{lem:ConvexVar3}. $\langle x, n\rangle$ pushes up and the square $\langle x-1, n+1\rangle$ pushes right, so $\langle x-1, n+1\rangle$ must also push up.

If the variable component is minus, then $\langle x, -2\rangle$ pushes left by \Cref{lem:ConvexVar3}. $\langle x, -2\rangle$ also pushes up, and so the squares $\langle x+\frac12, -1\rangle$ and $\langle x-\frac12, -1\rangle$ must also push up (these squares overlap the $x$-coordinates of $\langle x, n\rangle$ by \Cref{lem:ConvexVar1}). Again by \Cref{lem:ConvexVar1}, this means that $\langle x, 0\rangle$ and $\langle x-1, 0\rangle$ push up, so $\langle x-1, n+1\rangle$ pushes up.

So if $\langle x, -2-n\rangle$ pushes up, then $\langle x, n+1\rangle$ and $\langle x-1, n+1\rangle$ push up. Similarly, it can be shown that if $\langle x, n+1\rangle$ pushes down, then $\langle x, -2-n\rangle$ and $\langle x+1, -2-n\rangle$ push down. This gives the required result.
\end{proof}

\begin{corollary}\label{lem:ConvexStackGrowth}
Each vertical stack grows by $2$ squares for each variable component that it fully passes through.
\end{corollary}

\begin{proof}
There is a static shift between every pair of variable components, so this is clear by \Cref{lem:StaticShift,lem:ConvexVar4}.
\end{proof}

We can now verify that the push squares of the PUSH gadgets work as claimed.

\begin{lemma}\label{lem:ConvexPushGadgets}
If a variable component is plus, then the rightmost square of each of its PUSH-UP gadgets pushes up. If a variable component is minus, then the leftmost square of each of its PUSH-DOWN gadgets pushes down.
\end{lemma}

\begin{proof}
The PUSH-UP gadgets are right of the helper row gadgets and the PUSH-DOWN gadgets are left of the helper row gadgets. So by \Cref{lem:ConvexVar3}, if the variable component is in the up position then the upper helper row squares in a PUSH-UP gadget push right, which means that the rightmost wire square pushes up (see \Cref{fig:ConvexPushUp}). Similarly, if the variable component is in the down position, then the lower helper row squares in a PUSH-DOWN gadget push left, so the leftmost wire square pushes down (see \Cref{fig:ConvexPushDown}).
\end{proof}

Put together, we can now verify all the properties of the variable components.

\begin{lemma}\label{lem:ConvexVariableGadgets}
Suppose that all of the variable membrane squares push in.
The our construction has the following properties:

\begin{itemize}
    \item All the clause membrane squares push out.
    \item If all the squares in a block of connecting membrane squares connected to a PUSH-UP-IF-PLUS gadget push down, then that variable component must be minus.
    \item If all the squares in a block of connecting membrane squares connected to a PUSH-DOWN-IF-MINUS gadget push up, then that variable component must be in the push position.
\end{itemize}
\end{lemma}

\begin{proof}
The conclusions of \Cref{lem:ConvexVar1,lem:ConvexVar3} require only that squares in the static row above a variable component push down. If the conclusions of these lemmas hold for the first $k$ variable components from the top, then by \Cref{lem:ConvexStackGrowth} the upper pinch squares and helper row creation squares for $(k+1)$st variable component push down. By induction, we conclude that the results of \Cref{lem:ConvexVar1,lem:ConvexVar3} hold for all variable components. 

Again by \Cref{lem:ConvexStackGrowth}, we can now conclude that all lower pinch squares push up. So each pair of variable rows in a variable component are aligned going all the way to the left edge (this is important because the PUSH-DOWN gadgets are left of redundancy gadgets). 

Each set of variable membrane squares connects to a block of pinch, helper row creation, or redundancy squares. The two squares on either side of such a block push out by construction of the gadgets (see our diagrams). So by \Cref{lem:ConvexStackGrowth}, the squares adjacent to the block of variable membrane squares also push out. This is shown in \Cref{fig:MultipleConvexGrowth1}.

\begin{figure}
\centering
\includegraphics[page=8]{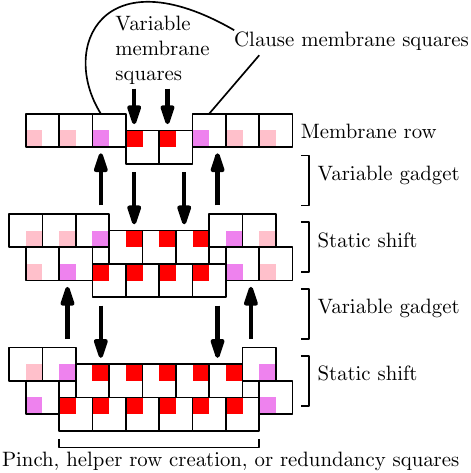}
\caption{A block of e.g.~pinch squares are created by a block of variable squares at the boundary. The squares adjacent to the block of pinch squares push out by construction of the gadget, so by \Cref{lem:ConvexStackGrowth} the squares adjacent to the set of variable membrane squares also push out.}
\label{fig:MultipleConvexGrowth1}
\end{figure}

By a similar principal and \Cref{lem:ConvexStackGrowth}, the connecting membrane squares behave as required. That is, the clause membrane squares adjacent to the block of connecting membrane squares push out, and it is only possible for all the squares in the block to push in if the variable component is in the appropriate position (see \Cref{lem:ConvexPushGadgets}).
This is shown in \Cref{fig:MultipleConvexGrowth2}.

\begin{figure}
\centering
\includegraphics[page=5]{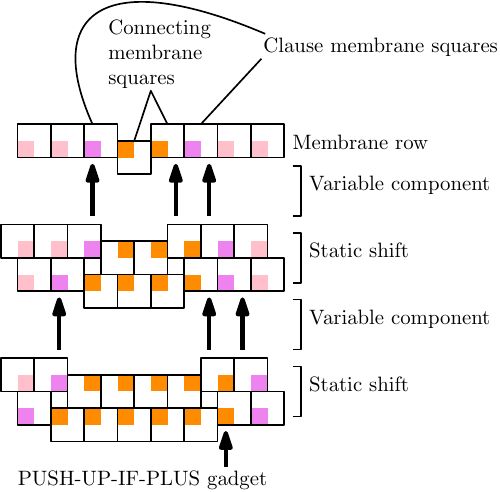}
\caption{When the push squares are up, the rightmost connecting membrane square is also up, is if all connecting membrane squares are down, the variable component containing the PUSH-UP-IF-PLUS gadget must be minus.}
\label{fig:MultipleConvexGrowth2}
\end{figure}
\end{proof}

\subsection{Clause components}\label{sec:ClauseGadgets}

We will need two sets of clause components, one above the variable components and one below. These are both created in the same way, so throughout this section we just consider the top set of components.
To create the bottom set of components, just exchange ``up'' and ``down''.
The clause components above the variable components represent negative clauses in the formula $\Phi$ that we reduce from, while the clause components below the variables represent positive clauses.

\subsubsection{Tester stacks and SWITCH gadgets}

Each set of clause components is formed by a large number of criss-crossing horizontal and vertical stacks. Each stack is limited in how wide it can grow.
\Cref{fig:StackStart} shows how the start of a stack is created.
On the other side of the polygon, the boundary prevents the stack from growing by more than a certain amount.
For a vertical stack, the membrane row prevents the stack from growing more than a certain amount, as shown in \Cref{fig:VerticalStackEnd}.

\begin{figure}
\centering
\includegraphics[page=4]{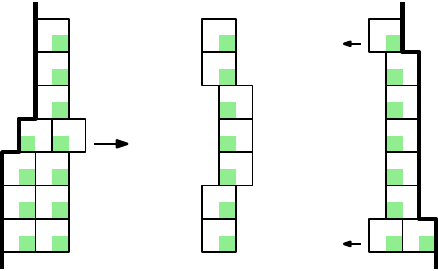}
\caption{The squares in the middle must form part of a horizontal stack that pushes to the right, due to the edge of the polygon at the left side that pushes a square to the right.
The adjacent squares could also be part of this stack.
A vertical stack or a stack pushing to the left could be created in a similar way.
For horizontal stacks, a section of the polygon boundary on the opposite side prevents the horizontal stack from growing more than a set amount, as shown here.
}
\label{fig:StackStart}
\end{figure}

\begin{figure}
\centering
\includegraphics[page=6]{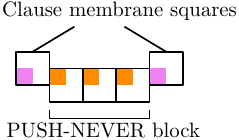}
\caption{A block of connecting membrane squares (which are attached to a PUSH-NEVER gadget in a variable component below) allow a vertical stack created in a clause component above to grow only by up to a certain amount and no more.}
\label{fig:VerticalStackEnd}
\end{figure}

An important ingredient is the \emph{SWITCH gadget}, which must create a vertical stack or a horizontal stack, but generally not both.
This is shown in \Cref{fig:RightAngle}.
Consider a negative clause $\lnot x_i\vee \lnot x_j\vee \lnot x_k$ of the formula $\Phi$ that we reduce from.
We make a corresponding clause component in which we have three SWITCH gadgets corresponding to $x_i,x_j,x_k$, respectively.
These define variables $y_i,y_j,y_k$, and for each of these $y_\ell$, we define $y_\ell=1$ if the SWITCH gadget creates a vertical stack, and otherwise we define $y_\ell=0$.
The clause component will enforce the constraint
$y_i+y_j+y_k=1$, i.e., exactly one SWITCH gadget will make a vertical stack.
When $y_\ell=1$ and the respective gadget creates a vertical stack, this stack hits a block of connecting membrane squares that is connected to a PUSH-UP-IF-PLUS gadget in the corresponding variable component $x_\ell$.
An example of this is shown in \Cref{fig:PushConnection}.
The PUSH-UP-IF-PLUS gadget only allows all these squares to push down if the variable component is minus.
This creates a constraint $y_\ell\implies \neg x_\ell$.
Since one of the SWITCH gadgets make a vertical stack, we conclude that the clause $\lnot x_i\vee \lnot x_j\vee \lnot x_k$ is satisfied.
This is only a one way implication---if $y_\ell$ is $0$ then there is no constraint on $x_\ell$, so even though just one SWITCH gadget makes a vertical stack, there can be more than one variable satisfying the clause.

\begin{figure}
\centering
\includegraphics[page=2]{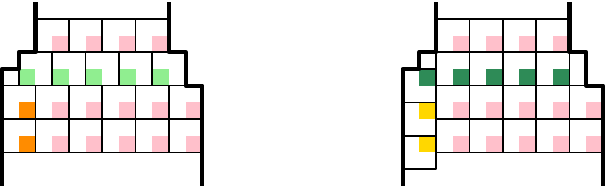}
\caption{A \emph{SWITCH gadget} that either creates a vertical stack or a horizontal stack.}
\label{fig:RightAngle}
\end{figure}

\begin{figure}
\centering
\includegraphics[page=3]{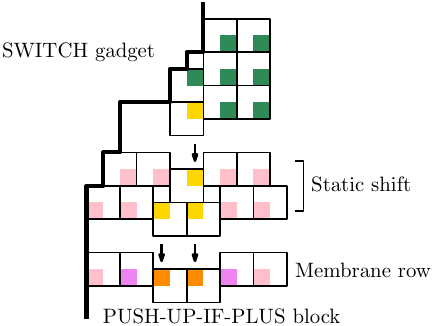}
\caption{If the SWITCH gadget creates a vertical stack, then all the squares in a block of connecting membrane squares push down. These connect to a PUSH-UP-IF-PLUS gadget, so this can only happen if the variable component is minus.}
\label{fig:PushConnection}
\end{figure}

The corresponding construction in the clause components below the variable components result in the implication $y_i\implies x_k$, since there we connect the SWITCH gadget to a PUSH-DOWN-IF-MINUS gadget.
Hence, we can also realize a positive clause $x_i\vee x_j\vee x_k$.

The rows and columns created by the SWITCH gadgets are called \emph{literal rows} and \emph{literal columns}.
We say that a literal row is \emph{aligned} if there is no stack present and \emph{unaligned} otherwise.
Note that the literal row created by a SWITCH gadget on the left side of the polygon should be unaligned when pushing right, but the literal row created by a SWITCH gadget on the right side of the polygon should be unaligned when pushing left.
Whenever we need a SWITCH gadget on the right side of the polygon, we will need a static shift above and below it to change the alignment of the static squares. This can be seen in \Cref{fig:FullOrGadget1}.

There are also some stacks that will always exist. These are called \emph{tester rows} and \emph{tester columns}.
Each tester row is constrained to grow at most a fixed number of times, a constraint that is easier to satisfy the fewer vertical stacks there are. The tester columns, on the other hand, create constraints that get easier to satisfy the fewer \emph{horizontal} stacks there are.

\subsubsection{Defining the clause components}

We are now ready to fully define the clause components. \Cref{fig:ConvexORSchematic} shows the setup at a schematic level. Due to the structure of the schematics produced from an instance of \MCTSAT, all the literal columns and tester columns from the $i$th clause from the top pass through the literal and testers rows from the $k$th clause if $k>i$. 

\begin{figure}
\centering
\includegraphics[page=3]{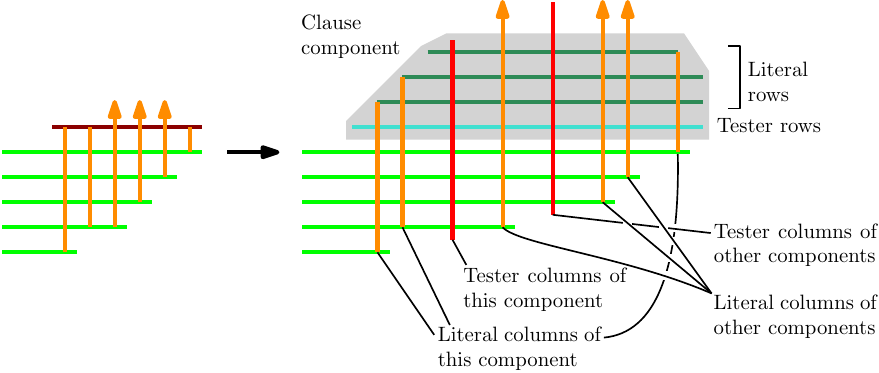}
\caption{Expanding the schematic described in \Cref{lem:ConvexConversion} to make space for the tester rows and tester columns. Each red or light blue line in this schematic represents multiple tester rows or columns.
}
\label{fig:ConvexORSchematic}
\end{figure}

Now we just need to specify the number of tester rows and tester columns for each clause and the amount by which each stack in the construction may grow. Here we mostly ignore horizontal static shifts (needed to change the alignment between left and right SWITCH gadgets) and vertical static shifts (needed to push down the variable membrane squares). By \Cref{lem:StaticShift,lem:VerticalStaticShift}, each stack grows exactly once when it passes through one of these shifts. When creating the final construction, the full amount that each stack can grow will be increased to account for this. The tester rows and tester columns are defined as follows:

\begin{itemize}
\item Starting with $i=0$, the $i$th clause component from the top has $3+i$ tester columns.
These start above the literal rows for this clause component and are created by a section of the top wall of the polygon.
Let $t_i=3(i+1)+\frac12i(i+1)$ be the total number of tester columns in clauses $0,\ldots,i$.

\item Each tester column is allowed to grow by at most two squares on each side, so by at most $4$ squares in total. Since the column starts with width $1$, this means that it should be allowed to reach a size of at most $5$ squares (plus the total number of static shifts that is passes through).

\item The $i$th clause component from the top has $1+2t_i$ tester rows directly below its literal rows.

\item A tester row below the $i$th clause component is allowed to grow by at most $i+1+t_i$ squares on each side, so by $2(i+1+t_i)$ squares in total.
\end{itemize}

An unaligned literal column should always push down on an entire block of connecting membrane squares, even if it only grows at static shifts. Since the clause membrane squares adjacent to the connecting membrane squares push up, this prevents an unaligned literal column from growing any further. 
So it just remains to specify how much the literal rows can grow.
Each clause component has $3$ literal rows.

\begin{itemize}
    \item The top literal row in the $i$th clause component from the top (again starting at $i=0$) can grow by at most $i+t_{i-1}$ squares on each side, so by at most $2(i+t_{i-1})$ squares in total.
    \item The lower two literal rows in the $i$th clause component can grow by an additional square on each side, so by at most $2(1+i+t_{i-1})$ in total.
\end{itemize}

This finishes the description of the clause component and thus of the entire polygon.
It remains to show that if $\Phi$ is satisfiable, then there is a perfect packing of the clause components and to verify that if there is a perfect packing, then $\Phi$ is satisfiable.

\subsubsection{Packing the clause components}

The following lemma describes how to pack the clause components using a satisfying assignment of $\Phi$.

\begin{lemma}\label{lem:ClausePackingsExist}
Suppose there is a satisfying assignment of the \MCTSAT\ formula $\Phi$. Then there is a packing of the upper clause components where:

\begin{itemize}
    \item The clause membrane squares push up.
    \item Each square in a block of connecting membrane squares connected to a PUSH-UP gadget pushes up unless the corresponding variable is $0$ in the satisfying assignment.
\end{itemize}
\end{lemma}

\begin{proof}
Each literal has a corresponding SWITCH gadget. For each clause, choose a true literal in that clause. The SWITCH gadgets corresponding to those literals are set to create a vertical stack and not a horizontal stack. The remaining SWITCH gadgets each create a horizontal stack and not a vertical stack.

So each clause component has two unaligned literal rows and one unaligned literal column. This determines which rows and columns are aligned or unaligned. Now we just need to specify how the stacks grow at each crossing. This is done by assigning each stack a priority $p$. When two stacks cross, the stack with a larger value of $p$ grows by $2$ squares and the stack with the smaller value of $p$ does not grow. The priorities are assigned as follows:

\begin{itemize}
    \item The literal columns have priority $0$.
    \item The literal rows in the $i$th clause component have priority $2i+1$.
    \item The tester columns starting above the $i$th clause component have priority $2i+2$.
    \item The tester rows have priority $\infty$.
\end{itemize}

That is to say, the literal columns only grow at static shifts. Each tester column grows at the first two unaligned literal rows that it crosses, then doesn't grow any more. 

The unaligned literal rows in the $i$th clause component grow when they cross a tester column for one of the previous clause components, but do not grow at the tester columns for that column. The literal rows also grow at each of the unaligned literal columns. There are $t_{i-1}$ and $i$ unaligned literal columns coming from the clause components above. The lower of the two literal rows also cross some of the literal columns from \emph{this} variable component, so may need to grow $1+i+t_{i-1}$ times.

The tester rows grow at all the stacks that they pass through. Below the $i$th clause component, there are $t_i$ tester columns and $i+1$ unaligned literal columns (recall that $i$ starts at $0$).

\Cref{fig:BasicOR,fig:FirstORPackings,fig:SecondClause} show some schematics of what this looks like.

Recall that a block of connecting membrane squares that is connected to a PUSH-UP gadget below is connected to a literal column above. Recall also that the clause above contain only negative literals. The literal columns that push down all correspond to true literals, which correspond to variables with value $0$. So the only connecting membrane squares that push down are connected to variables that are $0$ in the satisfying assignment.
\end{proof}

\begin{figure}
\centering
\includegraphics[page=4]{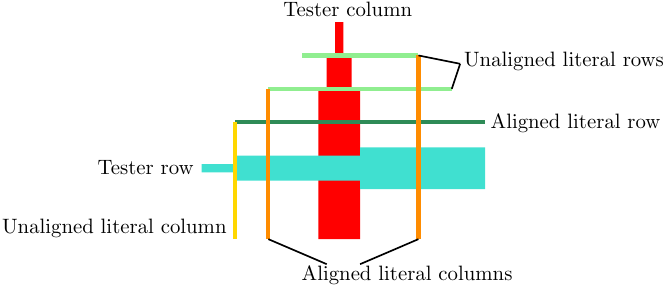}
\caption{A simplified clause component. The tester column (red) grows twice at the two unaligned literal rows (green). There are two unaligned literal rows (light green), which have corresponding aligned literal columns (orange). There is one aligned literal row (dark green), with a corresponding unaligned literal column (gold). The tester row (light blue) grows at the tester column and at the unaligned literal column.}
\label{fig:BasicOR}
\end{figure}

\begin{figure}
\centering
\includegraphics[page=5]{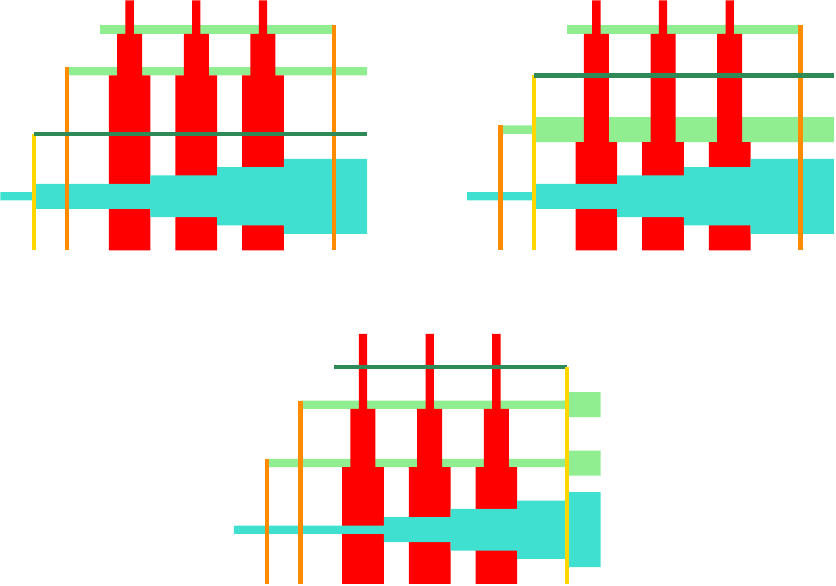}
\caption{The three ways to pack the first clause component, depending on which literal is true. Here only one tester row is pictured, in the actual construction there would be seven. Note that in the top left figure, the middle literal row doesn't grow to its full width.}
\label{fig:FirstORPackings}
\end{figure}

\begin{figure}
\centering
\includegraphics[page=6]{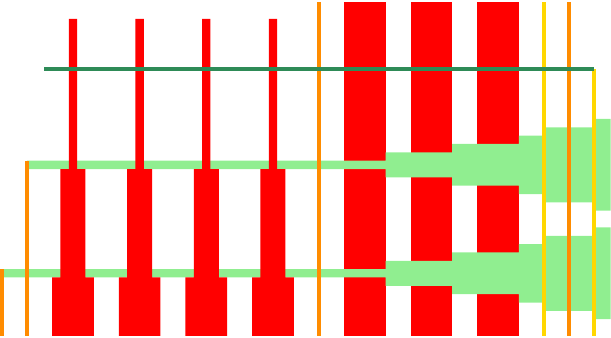}
\caption{The second clause component. The bottom two literal rows need to be able to grow $5$ times, once for the crossing with the third literal column, once for the one true literal from the first clause component, and once for each of the tester columns from the first clause component. The top literal row is only allowed to grow $4$ times. The tester rows are omitted. The second clause component has a $4$th tester column in case the first one sends some slack down.}
\label{fig:SecondClause}
\end{figure}

These are \emph{not} the only ways to pack the clause components. \Cref{fig:PartialGrowth} shows a schematic of a different packing of the first clause component.
That the rightmost tester column has not grown to its maximum width, creating some slack that can propagate down into the clause components below.
This extra slack is the reason that the number of testers needs to increase in the subsequent gadgets. This construction may seem unnecessarily complicated, but it is actually quite difficult to limit this slack propagation without needing exponentially many testers.

\begin{figure}
\centering
\includegraphics[page=7]{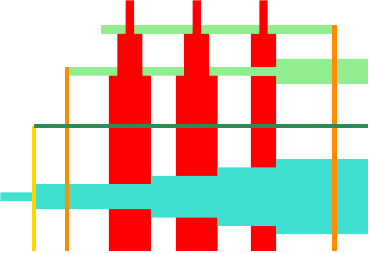}
\caption{If the lowest literal row is the one that is aligned, then some slack can propagate.}
\label{fig:PartialGrowth}
\end{figure}

\Cref{fig:FullOrGadget1,fig:FullOrGadget2,fig:FullOrGadget3,fig:FullOrGadget4} show how the first clause component is realized and how to pack it in a satisfying assignment.

\begin{figure}
\centering
\includegraphics[scale=0.8,page=8]{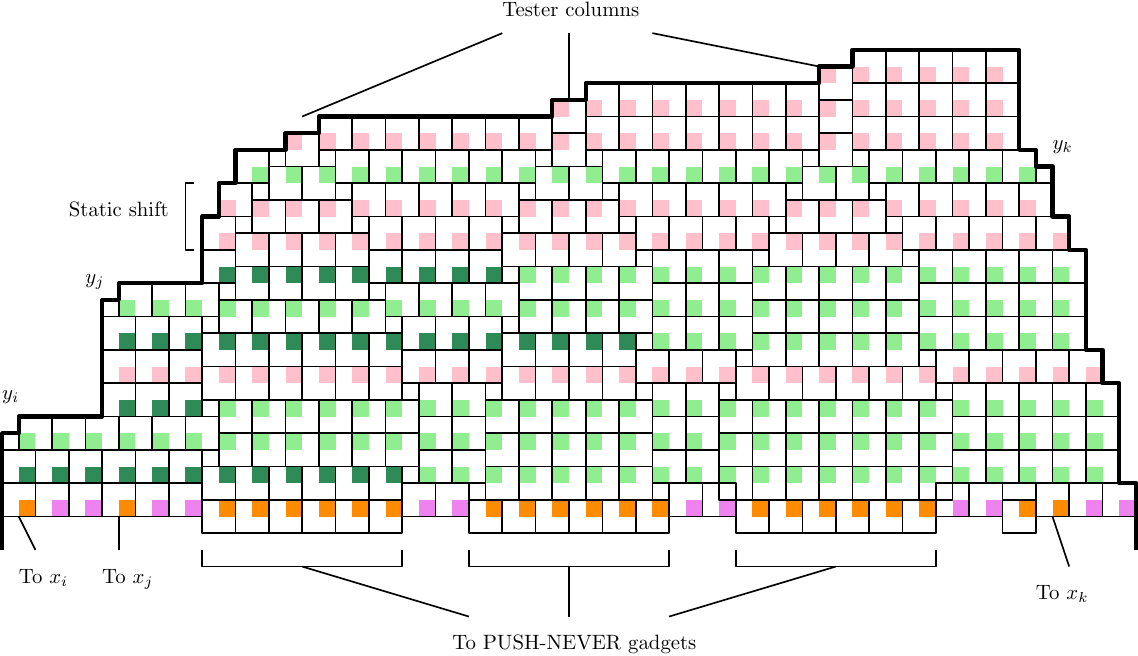}
\caption{The clause component, showing a state where none of the literals are true and so no perfect packing exists (two squares overlap).
}
\label{fig:FullOrGadget1}
\end{figure}

\begin{figure}
\centering
\includegraphics[width=\textwidth,page=9]{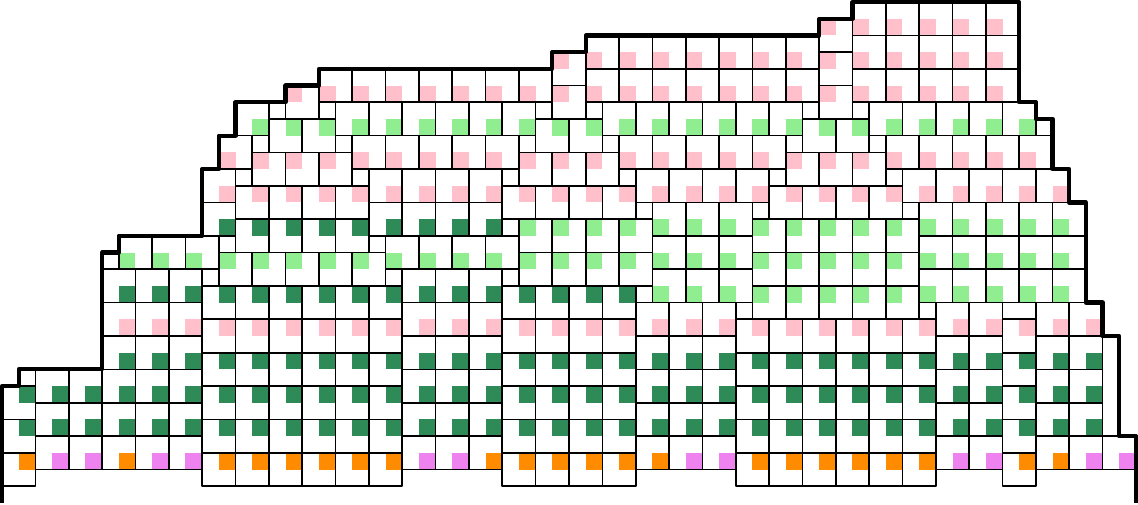}
\caption{A perfect packing \emph{does} exist if the first literal column is allowed to lower. The packing shown here is \emph{not} quite the same packing described in \Cref{lem:ClausePackingsExist}. Instead we show a case where the middle tester column has not grown to its maximum width in this component, creating some slack that can propagate down.}
\label{fig:FullOrGadget2}
\end{figure}

\begin{figure}
\centering
\includegraphics[width=\textwidth,page=10]{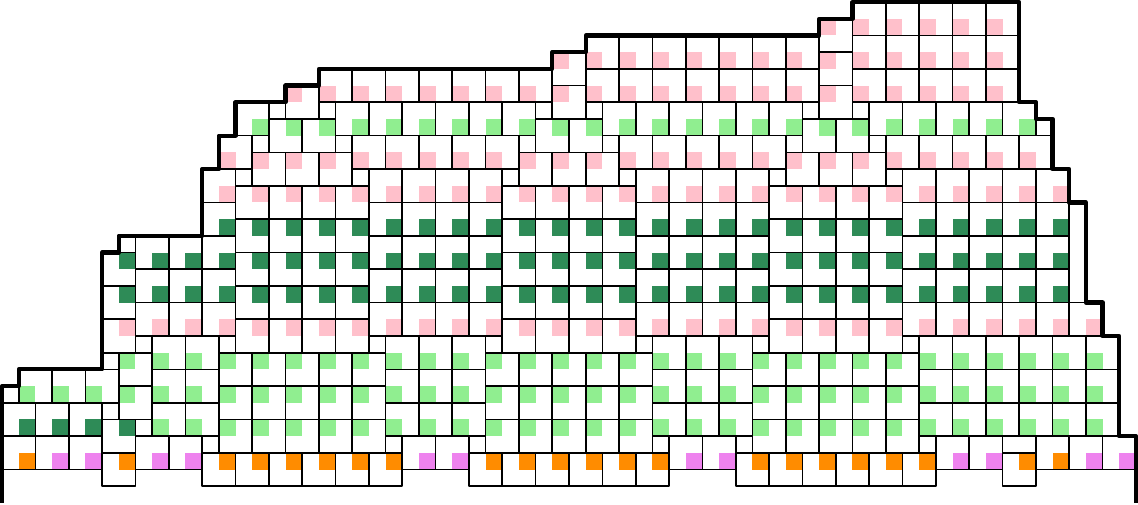}
\caption{A perfect packing when the second literal column is allowed to lower.}
\label{fig:FullOrGadget3}
\end{figure}

\begin{figure}
\centering
\includegraphics[width=\textwidth,page=11]{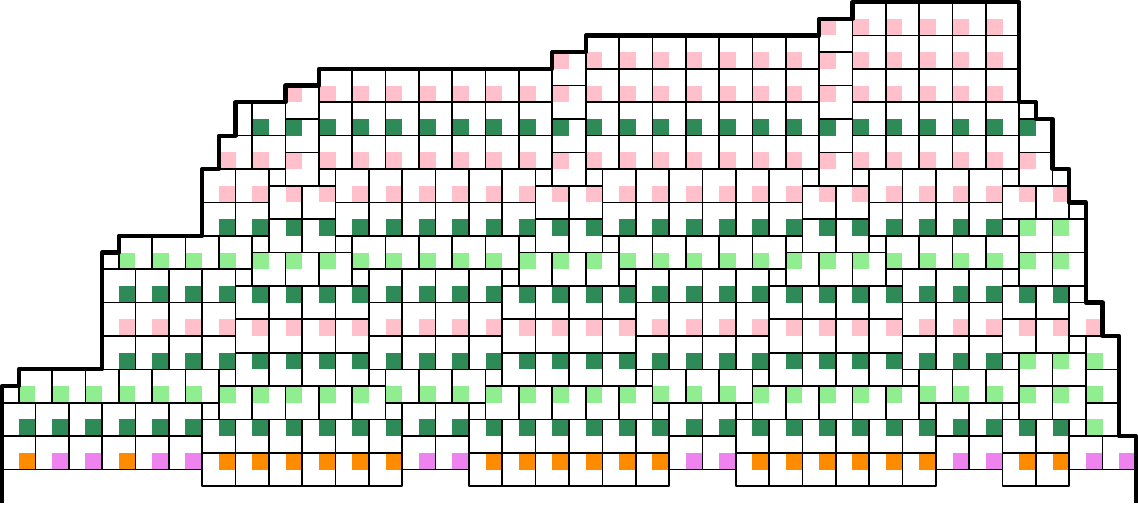}
\caption{A perfect packing when all the squares in the third literal column are allowed to lower.}
\label{fig:FullOrGadget4}
\end{figure}

\subsection{Verification of the clause components}\label{sec:ClauseGadgetVerification}

We now formally verify that any packing of the clause components requires one of the literal columns for each clause to push down. To carefully check this, we write some algebraic constraints that must be satisfied. Say there are $n$ clauses $c_i: y_{3i}+y_{3i+1}+y_{3i+2}=1$, with $c_0$ being at the top and $c_{n-1}$ being at the bottom.

For the $i$th SWITCH gadget (starting with $i=0$), we define a variable $y_i$ that takes values in $\{0, 1\}$. The variable $y_i$ is $1$ when the SWITCH gadget creates a vertical stack and $0$ otherwise.

We give a name to each stack in the construction that can grow. The row for the literal $y_i$ is $\ell_i$. The $j$th tester column for the clause $c_i$ is called $c_{i, j}$. The $j$th tester row for the clause $c_i$ is called $r_{i, j}$. 

Recall that when two stacks cross each other, combined they grow by a total of at least $2$ squares (\Cref{lem:StackGrowth}). For the crossing between a horizontal stack $a$ and a vertical stack $b$, we designate a variable $x(a, b)\in \{0, 1, 2\}$ that records how many times $a$ grows when passing through this crossing. The vertical stack $b$ then grows $2-x(a, b)$ times. The stacks created by unaligned literal columns can't grow, so any row that crosses an unaligned literal column grows by $2$ and we don't need an $x$ variable to keep track of this. 

The full set of constraints can now be written. For $0\le i < n$ and $0\le j < 1+2t_i$ there is a tester row $r_{i, j}$. This is allowed to grow at each tester column and at one of the literal columns for each of the clauses above it, up to a total width of $2\left(i+1+t_i\right)$ (recall that $t_i=3(i+1)+\frac12i(i+1)$ is the number of tester columns in clauses $c_0,\ldots,c_i$).
The constraint created by a tester row can then be expressed as follows:

\begin{equation}\label{eqn:TesterRows}
2\sum_{p=0}^{3i+2}y_p+\sum_{p=0}^{i}\sum_{q=0}^{3+p-1}x(r_{i, j}, c_{p, q})\le 2\left(i+1+t_i\right)
\end{equation}

For $0\le i < n$, the uppermost literal row in the $i$th clause is $\ell_{3i}$.
The literal row $\ell_{3i}$ is allowed to grow once at each of the tester columns for \emph{previous} clauses, and at one of the literal columns for each of the clauses above it, for a total of $2\left(i+t_{i-1}\right)$ times.
This leads to the following inequality.

\begin{equation}\label{eqn:TopLiteralRow}
(1-y_{3i})\left(2\sum_{p=0}^{3i-1}y_p+\sum_{p=0}^{i}\sum_{q=0}^{3+p-1}x(\ell_{3i}, c_{p, q})\right)\le 2\left(i+t_{i-1}\right)
\end{equation}

The next two literal rows in the $i$th clause are each allowed to grow by an additional two squares. So for $0\le i < n$ and $m\in \{1, 2\}$, we have the following:

\begin{equation}\label{eqn:LiteralRows}
(1-y_{3i+m})\left(2\sum_{p=0}^{3i+m-1}y_p+\sum_{p=0}^{i}\sum_{q=0}^{3+p-1}x(\ell_{3i+m}, c_{p, q})\right)\le 2\left(1+i+t_{i-1}\right)
\end{equation}

Finally, for $0\le p < n$ and $0\le q<3+i$, we have a tester column $c_{p, q}$, which is allowed to grow $4$ times.
Recall that $c_{p, q}$ crosses all literal and tester rows in clauses $c_{p+1}$ through $c_{n-1}$.
This leads to the constraint:

\begin{equation}\label{eqn:TesterCols}
\sum_{i=p}^{n-1}\sum_{m=0}^{2}(1-y_{3i+m})\left(2-x\left(\ell_{3i+m}, c_{p, q}\right)\right)+\sum_{i=p}^{n-1}\sum_{j=0}^{2t_i}\left(2-x\left(r_{i, j}, c_{p, q}\right)\right)\le 4
\end{equation}

Here we are assuming that a literal row is always aligned when the corresponding literal column is unaligned. There could be both vertical and horizontal stacks starting at a single SWITCH gadget, but this could only increase the values of the left hand sides of \eqref{eqn:TesterRows}--\eqref{eqn:TesterCols}. 

\begin{lemma}\label{lem:ClauseGadget}
Inequalities \eqref{eqn:TesterRows}--\eqref{eqn:TesterCols} imply that, for each $i$, $y_{3i}+y_{3i+1}+y_{3i+2}=1$.
\end{lemma}

\begin{proof}
By induction on $k$, we show that for $i<k$, we have
\begin{equation}
y_{3i}+y_{3i+1}+y_{3i+2}=1.\label{eq:InductiveHypothesis0}
\end{equation}
We will simultaneously show (in the same induction) that:

\begin{equation}\label{eqn:InductiveHypothesis}
\sum_{p=0}^{k-1}\sum_{q=0}^{3+p-1}\sum_{i=p}^{k-1}\sum_{m=0}^{2}(1-y_{3i+m})\left(2-x\left(\ell_{3i+m}, c_{p, q}\right)\right)\ge 4t_{k-1}-2k
\end{equation}

The left hand side of \eqref{eqn:InductiveHypothesis} counts how much the tester columns grow in total when crossing literal rows in the first $k$ clause components.
Recall that each of the $t_{k-1}$ tester columns is allowed to grow by $4$ squares, so in total they can grow by $4t_{k-1}$ squares.
So what \eqref{eqn:InductiveHypothesis} says is that at most $2k$ squares worth of slack is propagating downward after the first $k$ clause components.

The base case for the induction is $k=0$. Since there are no literals and no tester columns, the result is trivial in the base case. 

For induction, suppose the above holds for some $k$. First, we show that $y_{3k}+y_{3k+1}+y_{3k+2}\le 1$.
By using \eqref{eq:InductiveHypothesis0} in \eqref{eqn:TesterRows}, we have, for $0\le j < 1+2t_k$:

\begin{equation}\label{eqn:TesterRowsWithHypothesis}
2k+2(y_{3k}+y_{3k+1}+y_{3k+2})+\sum_{p=0}^{k}\sum_{q=0}^{3+p-1}x(r_{k, j}, c_{p, q})\le 2\left(k+1+t_k\right)
\end{equation}

We want to show that there is a value of $j$ where the sum of the $x(r_{k, j}, c_{p, q})$ is at least $2t_k-1$. All the terms in \eqref{eqn:TesterCols} are positive, so for each $p, q$ we can extract that:

\[\sum_{j=0}^{2t_k}\left(2-x\left(r_{k, j}, c_{p, q}\right)\right)\le 4\]
Summing over the $t_k$ values of $p$ and $q$:

\[\sum_{j=0}^{2t_k}\sum_{p=0}^{k}\sum_{q=0}^{3+p-1}\left(2-x\left(r_{k, j}, c_{p, q}\right)\right)\le 4t_k\]
The outermost sum sums over $2t_k+1$ positive integral values. Since $2(2t_k+1)<4t_k$, at least one of these values must be less than $2$ (and so less than or equal to $1$). That is to say, there is some $j$ such that:

\[\sum_{p=0}^k\sum_{q=0}^{3+p-1}\left(2-x\left(r_{k, j}, c_{p, q}\right)\right)\le 1\]
This can be rearranged to obtain:

\[\sum_{p=0}^k\sum_{q=0}^{3+p-1}x\left(r_{k, j}, c_{p, q}\right)\ge 2t_k-1\]
Subtracting this from \eqref{eqn:TesterRowsWithHypothesis}, we see that:

\[2(y_{3k}+y_{3k+1}+y_{3k+2})\le 3\]
Since the $y_i$ are integer valued, this implies that $y_{3k}+y_{3k+1}+y_{3k+2}\le 1$.

Next we should show that $y_{3k}+y_{3k+1}+y_{3k+2}\ge 1$. Suppose for contradiction that
\begin{equation}\label{eq:ContradictionAssumption}
y_{3k}=y_{3k+1}=y_{3k+2}=0
\end{equation}
By subtracting the inductive hypothesis \eqref{eqn:InductiveHypothesis} from the sum of \eqref{eqn:TesterCols} over $0\le p < k$ and $0\le q < 3+p$, we see that:

\begin{equation}\label{eqn:InductiveHypothesisReverse}
2k \ge \sum_{p=0}^{k-1}\sum_{q=0}^{3+p-1}\sum_{m=0}^{2}\left(2-x\left(\ell_{3k+m}, c_{p, q}\right)\right)=6t_{k-1}-\sum_{p=0}^{k-1}\sum_{q=0}^{3+p-1}\sum_{m=0}^{2}x\left(\ell_{3k+m}, c_{p, q}\right)
\end{equation}
Adding up \eqref{eqn:TopLiteralRow} and versions of \eqref{eqn:LiteralRows} for both values $m\in\{1,2\}$ (and simplifying using \eqref{eq:ContradictionAssumption} and \eqref{eq:InductiveHypothesis0}):

\[6k+\sum_{m=0}^2\sum_{p=0}^{k}\sum_{q=0}^{3+p-1}x(\ell_{3k+m}, c_{p, q})\le 4+6\left(k+t_{k-1}\right)\]
Splitting the sum into the cases $p=k$ and $p<k$, we get:
\[6k+\sum_{m=0}^{2}\sum_{q=0}^{3+k-1}x(\ell_{3k+m}, c_{k, q})+\sum_{m=0}^2\sum_{p=0}^{k-1}\sum_{q=0}^{3+p-1}x(\ell_{3k+m}, c_{p, q})\le 4+6\left(k+t_{k-1}\right)\]
So together with \eqref{eqn:InductiveHypothesisReverse},
we see that:

\[\sum_{m=0}^{2}\sum_{q=0}^{3+k-1}x(\ell_{3k+m}, c_{k, q})\le 4+2k\]
Since $2(3+k)>4+2k$, there must be some value of $q$ for which:

\[\sum_{m=0}^{2}x(\ell_{3k+m}, c_{k, q})\le 1\]
This implies that:

\[\sum_{m=0}^{2}\left(2-x\left(\ell_{3k+m}, c_{k, q}\right)\right)\ge 5\]
From \eqref{eqn:TesterCols} we can extract that:

\[\sum_{m=0}^{2}\left(2-x\left(\ell_{3k+m}, c_{k, q}\right)\right)\le 4,\] 
a contradiction.
We conclude that $y_{3k}$, $y_{3k+1}$ and $y_{3k+2}$ can't all be zero, and since we already saw that $y_{3k}+y_{3k+1}+y_{3k+2}\le 1$, we get $y_{3k}+y_{3k+1}+y_{3k+2}=1$.
It remains to check that \eqref{eqn:InductiveHypothesis} holds for $k+1$. 

For a general function $f$ in two variables:

\[\sum_{p=0}^{k}\sum_{i=p}^{k}f(p, i)=\sum_{p=0}^{k-1}\sum_{i=p}^{k-1}f(p, i)+\sum_{p=0}^{k}f(p, k)\]
So going from $k$ to $k+1$, the left hand side of \eqref{eqn:InductiveHypothesis} increases by:

\[\sum_{p=0}^{k}\sum_{q=0}^{3+p-1}\sum_{m=0}^{2}(1-y_{3k+m})\left(2-x\left(\ell_{3k+m}, c_{p, q}\right)\right)\]
Going from $k$ to $k+1$, the right hand side of \eqref{eqn:InductiveHypothesis} increases by $4\left(3+k\right)-2$ (since $t_k-t_{k-1}$ is $3+k$). So to show that \eqref{eqn:InductiveHypothesis} holds for $k+1$, it is sufficient to show that:

\[\sum_{p=0}^{k}\sum_{q=0}^{3+p-1}\sum_{m=0}^{2}(1-y_{3k+m})\left(2-x\left(\ell_{3k+m}, c_{p, q}\right)\right)\ge 4\left(3+k\right)-2\]

There are $3$ cases to check based on which of $y_{3k}$, $y_{3k+1}$ or $y_{3k+2}$ is equal to $1$. Consider the case $y_{3k+2}=1$. Summing \eqref{eqn:TopLiteralRow} for $i=k$ and \eqref{eqn:LiteralRows} for $i=k$ and $m=1$, and simplifying using \eqref{eq:InductiveHypothesis0}, we get:

\[4k+\sum_{m=0}^1\sum_{p=0}^{k}\sum_{q=0}^{3+p-1}x(\ell_{3k+m}, c_{p, q})\le 2+4\left(k+t_{k-1}\right)\]
Rearranging (using that the number of pairs of $(p, q)$ is $t_k=t_{k-1}+3+k$), we obtain:

\[\sum_{p=0}^{k}\sum_{q=0}^{3+p-1}\sum_{m=0}^1\left(2-x(\ell_{3k+m}, c_{p, q})\right)\ge 4(3+k)-2\]
So:

\[\sum_{p=0}^{k}\sum_{q=0}^{3+p-1}\sum_{m=0}^2\left(1-y_{3k+m}\right)\left(2-x(\ell_{3k+m}, c_{p, q})\right)\ge 4(3+k)-2\]
as required.
The cases $y_{3k}=1$ and $y_{3k+1}=1$ are similar.
By induction, this completes the proof.
\end{proof}

\begin{lemma}\label{lem:ClauseGadgetsFinal}
Suppose there is a perfect packing of the clause components where the membrane clause squares push out. Then if $\Phi$ has a clause of form $\neg x_i\vee \neg x_j\vee \neg x_k$ (resp.~$x_i\vee x_j\vee x_k$), then there is a block of membrane connecting squares in the upper (resp.~lower) membrane row that all push in and connects to a PUSH-UP-IF-PLUS (resp.~PUSH-DOWN-IF-MINUS) gadget for one of the variables $x_i, x_j$ or $x_k$.  
\end{lemma}

\begin{proof}
Whenever a SWITCH gadget creates a vertical stack, the corresponding block of membrane connecting squares must all push in (by \Cref{lem:StaticShift}). \Cref{lem:ClauseGadget} says that, for each clause, at least one of the SWITCH gadgets must create a vertical stack.
So for clause of form $\neg x_i\vee \neg x_j\vee \neg x_k$ (resp.~$x_i\vee x_j\vee x_k$) in $\Phi$, one of the variables has a PUSH-UP-IF-PLUS (resp.~PUSH-UP-IF-MINUS) gadget that is connected to a block of membrane connecting squares that all push in.
\end{proof}

\subsection{Proof of \Cref{thm:main}}

We are now ready to give the proof of our main theorem.

\mainthm*

\begin{proof}
Let $\Phi$ be an instance of \MCTSAT. Construct the orthogonally convex grid polygon $P$ and the number $k$ of squares to be packed as above. The number $k$ is $\mathcal{O}(m^8)$ where $m$ is the number of variables and clauses in $\Phi$. This size is dominated by squares in and surrounding the tester rows. The polygon $P$ can be explicitly constructed in polynomial time, as outlines in the following.

Assume that $\Phi$ has $m$ clauses and $v$ variables.
We start by drawing the full schematics including all the variable rows, helper rows, vertical static shifts, pinch columns, helper row columns, redundancy columns, literal rows, literal columns, tester rows, tester columns, and PUSH columns. We then compute the size that each stack is allowed to grow (taking into account any static shifts), and construct the polygon by gluing together the various gadgets appropriately. In total:

\begin{itemize}
    \item There are $v$ variable gadgets, which \emph{in total} have $\mathcal{O}(m)$ helper rows. Since variables that don't appear in any clauses can be excluded, we can assume that $v=\mathcal{O}(m)$.
    \item Each variable gadget has an upper and lower pinch column, $\mathcal{O}(m)$ helper row columns, and $\mathcal{O}(m)$ redundancy columns. All of these grow by at most $\mathcal{O}(m)$ before terminating. Eacg if these columns has a vertical static shift, so the total number of vertical static shifts in the clause gadgets is $\mathcal{O}(m^2)$.
    \item There are $3m$ literal columns and $\mathcal{O}(m^2)$ tester columns, which each grow by $\mathcal{O}(m)$. These continue as PUSH columns in the variable components.
    \item There are $3m$ literal rows and $\mathcal{O}(m^3)$ tester rows, which each grow by $\mathcal{O}(m^2)$.
\end{itemize}

The size of the polygon is dominated by squares in and around the literal rows. There are $\mathcal{O}(m^3)$ literal rows, growing to a size of $\mathcal{O}(m^2)$. Since there are $\mathcal{O}(m^2)$ columns growing to a width of $\mathcal{O}(m)$ each, the total number of squares in the polygon is at most $\mathcal{O}(m^8)$.

\Cref{lem:VariablePackingsExist,lem:ClausePackingsExist} say that there is a packing of $P$ with $k$ squares of size $2\times 2$ whenever $\Phi$ has a satisfying assignment. \Cref{lem:ConvexVariableGadgets,lem:ClauseGadgetsFinal} say that any such packing of $P$ must correspond to a satisfying assignment of $\Phi$. So by the NP-hardness of \MCTSAT, the problem \USPack is NP-hard for orthogonally convex polygons.
\end{proof}

\section{Concluding remarks}

To our knowledge, these represent the first results on NP-hardness for packing or covering a simple polygon with identical (fixed) shapes and the first NP-hardness result for partitioning simple polygons into connected pieces. There are many interesting problems that are known to be hard for polygons with holes but with unknown complexity for simple polygons. Until now, techniques for showing hardness of many of these problems have not been available. 

The problem \USPack in a grid polygon is equivalent to \textsc{Maximum-Independent-Set} on a grid graph $G$ with diagonals added; see \Cref{fig:MaxIndependentSet}.
We can define $G$ to be \emph{orthogonally convex} if whenever two vertices $u,v\in V(G)$ are from the same row or column, then all grid points between $u$ and $v$ are also vertices of $G$.
Our result has the following interesting consequence:

\begin{figure}
\centering
\includegraphics[page=3]{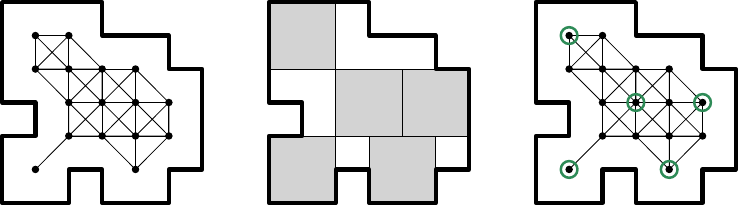}
\caption{The equivalence between \USPack and \textsc{Maximum-Independent-Set}.
Left: A polygon and the equivalent instance of \textsc{Maximum-Independent-Set}.
Middle: A packing with $2\times 2$ squares.
Right: The corresponding independent set.
}
\label{fig:MaxIndependentSet}
\end{figure}

\begin{corollary}
The maximum independent set problem for orthogonally convex grid graphs with diagonals is NP-hard.
\end{corollary}

The most natural class of polygons for which the complexity of \USPack is unresolved is \emph{staircase} grid polygons, i.e., grid polygons where the boundary can be partitioned into two chains, both of which are simultaneously $x$- and $y$-monotone.
Many of our ideas about packing in orthogonally convex polygons work here, but we do not know of a way to make clause components that could be used for this problem.
Another related problem is \USPack when the polygon $P$ is convex, but not a grid polygon.

Allowing the squares to rotate arbitrarily changes the problem significantly, and although it seems obvious that it makes the packing problem no more tractable than in the axis-aligned case, we have not found a way to prove hardness.
A long line of mathematical research has been devoted to this problem when the container $P$ is also a (larger) square.
This was initiated by Erd\H{o}s and Graham~\cite{erdos1975packing} in 1975, and it is still an active research area~\cite{chung2019efficient}.
The complicated nature of this problem is exemplified by the fact that even for a mere $11$ unit squares, it is unknown what is the smallest square in which they can be packed~\cite{Gensane2005}.
It may be possible to apply our ideas to show that it is NP-hard to pack unit squares in a simple polygon with rotation, but this would require at minimum a much more sophisticated version of the reference center idea.

For certain geometric shapes other than squares, our techniques may be useful for showing hardness of the associated packing problems in simple polygons.
The problem that seems most amenable to this approach is packing equilateral triangles (with $180$-degree rotations allowed).
We have not taken the effort to figure out the details, as packing equilateral triangles seems to be of limited interest.
Hardness of packing unit disks may be possible, but has some of the same problems as packing squares with rotations.

It may also be possible to show hardness for packing $n\times k$ rectangles with $90$-degree rotations allowed into a simple (not necessarily grid) polygon using these techniques.
It seems much more difficult to apply our ideas to the problem of packing $n\times k$ rectangles into a simple \emph{grid} polygon. This problem is in P for $1\times 2$ rectangles, but little is known even for $1\times 3$ rectangles.

In a forthcoming paper (with another set of authors), we show that \emph{reconfiguration} from one packing of axis-aligned unit squares to another is PSPACE-hard, even in a simple polygon.
The construction relies on a slight modification of the construction from \Cref{sec:simplepacking}.
Until now, it was only known that reconfiguration in a polygon with holes is PSPACE-hard~\cite{DBLP:journals/ijrr/SoloveyH16}.

\printbibliography

@book{o1987art,
  title={Art Gallery Theorems and Algorithms},
  author={O'Rourke, Joseph},
  year={1987},
  publisher={Oxford University Press}
}

@incollection{chazelle1994decomposition,
  title={Decomposition algorithms in geometry},
  author={Chazelle, Bernard and Palios, Leonidas},
  booktitle={Algebraic Geometry and its applications},
  editor={Chandrajit L. Bajaj},
  pages={419--447},
  year={1994},
  doi={10.1007/978-1-4612-2628-4_27}
}

@article{DBLP:journals/pieee/Shermer92,
  author    = {Thomas C. Shermer},
  title     = {Recent results in art galleries},
  journal   = {Proc. {IEEE}},
  volume    = {80},
  number    = {9},
  pages     = {1384--1399},
  year      = {1992},
  url       = {https://doi.org/10.1109/5.163407},
  doi       = {10.1109/5.163407},
  bibsource = {dblp computer science bibliography, https://dblp.org}
}

@incollection{chazelle1985approximation,
  title={Approximation and decomposition of shapes},
  author={Chazelle, Bernard},
  booktitle={Algorithmic and Geometric Aspects of Robotics},
  editor={Jacob T. Schwartz and Chee-Keng Yap},
  series={Advances in Robotics},
  volume={1},
  pages={145--185},
  year={1987}
}

@incollection{keil1999polygon,
  title={Polygon decomposition},
  author={Keil, J. Mark},
  booktitle={Handbook of computational geometry},
  pages={491--518},
  year={1999},
  chapter={11},
  editor={Sack, J{\"o}rg-R{\"u}diger and Urrutia, Jorge},
  doi={10.1016/B978-044482537-7/50012-7}
}

@incollection{keil1985minimum,
title = {Minimum Decompositions of Polygonal Objects},
editor = {Godfried T. Toussaint},
series = {Machine Intelligence and Pattern Recognition},
volume = {2},
pages = {197-216},
year = {1985},
booktitle = {Computational Geometry},
doi = {10.1016/B978-0-444-87806-9.50012-8},
author = {J. Mark Keil and J{\"o}rg-R. Sack}
}

@incollection{o2004polygons,
  title={Polygons},
  chapter={30},
  author={O’Rourke, Joseph and Suri, Subash and T\'{o}th, Csaba D.},
  booktitle={Handbook of discrete and computational geometry},
  editor={J. E. Goodman and J. O'Rourke},
  pages={787--810},
  year={2018},
  edition={Third edition},
  doi={10.1201/9781315119601},
}

@article{DBLP:journals/ijcga/BergK12,
  author       = {Mark de Berg and
                  Amirali Khosravi},
  title        = {Optimal Binary Space Partitions for Segments in the Plane},
  journal      = {Int. J. Comput. Geom. Appl.},
  volume       = {22},
  number       = {3},
  pages        = {187--206},
  year         = {2012},
  url          = {https://doi.org/10.1142/S0218195912500045},
  doi          = {10.1142/S0218195912500045},
  bibsource    = {dblp computer science bibliography, https://dblp.org}
}

@article{DBLP:journals/comgeo/BeauquierNRR95,
  author       = {Dani{\`{e}}le Beauquier and
                  Maurice Nivat and
                  Eric R{\'{e}}mila and
                  Mike Robson},
  title        = {Tiling Figures of the Plane with Two Bars},
  journal      = {Comput. Geom.},
  volume       = {5},
  pages        = {1--25},
  year         = {1995},
  url          = {https://doi.org/10.1016/0925-7721(94)00015-N},
  doi          = {10.1016/0925-7721(94)00015-N},
  bibsource    = {dblp computer science bibliography, https://dblp.org}
}

@article{DBLP:journals/talg/AamandARA23,
  author       = {Anders Aamand and
                  Mikkel Abrahamsen and
                  Thomas D. Ahle and
                  Peter M. R. Rasmussen},
  title        = {Tiling with Squares and Packing Dominos in Polynomial Time},
  journal      = {{ACM} Trans. Algorithms},
  volume       = {19},
  number       = {3},
  pages        = {30:1--30:28},
  year         = {2023},
  url          = {https://doi.org/10.1145/3597932},
  doi          = {10.1145/3597932},
  bibsource    = {dblp computer science bibliography, https://dblp.org}
}

@article{DBLP:journals/ipl/FowlerPT81,
  author       = {Robert J. Fowler and
                  Mike Paterson and
                  Steven L. Tanimoto},
  title        = {Optimal Packing and Covering in the Plane are NP-Complete},
  journal      = {Inf. Process. Lett.},
  volume       = {12},
  number       = {3},
  pages        = {133--137},
  year         = {1981},
  url          = {https://doi.org/10.1016/0020-0190(81)90111-3},
  doi          = {10.1016/0020-0190(81)90111-3},
  bibsource    = {dblp computer science bibliography, https://dblp.org}
}

@article{DBLP:journals/algorithmica/BaurF01,
  author       = {Christoph Baur and
                  S{\'{a}}ndor P. Fekete},
  title        = {Approximation of Geometric Dispersion Problems},
  journal      = {Algorithmica},
  volume       = {30},
  number       = {3},
  pages        = {451--470},
  year         = {2001},
  url          = {https://doi.org/10.1007/s00453-001-0022-x},
  doi          = {10.1007/S00453-001-0022-X},
  bibsource    = {dblp computer science bibliography, https://dblp.org}
}

@techreport{berman1981optimal,
  author      = {Berman, Francine and Leighton, Frank Thomsons and Snyder, Lawrence},
  title       = {Optimal tile salvage},
  institution = {Purdue University},
  year        = {1981},
  url = {https://docs.lib.purdue.edu/cstech/322},
  number      = {81--396}
}

@article{DBLP:journals/siamcomp/Lichtenstein82,
  author       = {David Lichtenstein},
  title        = {Planar Formulae and Their Uses},
  journal      = {{SIAM} J. Comput.},
  volume       = {11},
  number       = {2},
  pages        = {329--343},
  year         = {1982},
  url          = {https://doi.org/10.1137/0211025},
  doi          = {10.1137/0211025},
  bibsource    = {dblp computer science bibliography, https://dblp.org}
}

@misc{topp,
  title = {The Open Problems Project. Problem 56: Packing Unit Squares in a Simple Polygon},
  url = {https://topp.openproblem.net/p56},
  note = {Accessed: 2023-12-21},
  editor = {Erik D. Demaine and Joseph S. B. Mitchell and Joseph O’Rourke}
}

@phdthesis{el2009decomposing,
  title={Decomposing and packing polygons},
  author={El-Khechen, Dania},
  year={2009},
  school={Concordia University},
  url = {https://spectrum.library.concordia.ca/id/eprint/976664/}
}

@inproceedings{DBLP:conf/cccg/RenssenS11,
  author       = {Andr{\'{e}} van Renssen and
                  Bettina Speckmann},
  title        = {The $2\times 2$ Simple Packing Problem},
  booktitle    = {23rd Annual Canadian Conference on Computational
                  Geometry (CCCG 2011)},
  year         = {2011},
  url          = {http://www.cccg.ca/proceedings/2011/papers/paper14.pdf},
  bibsource    = {dblp computer science bibliography, https://dblp.org}
}

@inproceedings{DBLP:conf/cccg/El-KhechenDIO09,
  author       = {Dania El{-}Khechen and
                  Muriel Dulieu and
                  John Iacono and
                  Nikolaj van Omme},
  title        = {Packing 2{\texttimes}2 unit squares into grid polygons is NP-complete},
  booktitle    = {21st Annual Canadian Conference on Computational
                  Geometry (CCCG 2009)},
  pages        = {33--36},
  year         = {2009},
  url          = {http://cccg.ca/proceedings/2009/cccg09_09.pdf},
  bibsource    = {dblp computer science bibliography, https://dblp.org}
}

@article{DBLP:journals/corr/abs-1209-5307,
  author       = {Sarah R. Allen and
                  John Iacono},
  title        = {Packing identical simple polygons is NP-hard},
  journal      = {CoRR},
  volume       = {abs/1209.5307},
  year         = {2012},
  url          = {http://arxiv.org/abs/1209.5307},
  bibsource    = {dblp computer science bibliography, https://dblp.org},
  doi = {10.48550/arXiv.1209.5307}
}

@article{DBLP:journals/tit/ORourkeS83,
  author       = {Joseph O'Rourke and
                  Kenneth J. Supowit},
  title        = {Some NP-hard polygon decomposition problems},
  journal      = {{IEEE} Trans. Inf. Theory},
  volume       = {29},
  number       = {2},
  pages        = {181--189},
  year         = {1983},
  url          = {https://doi.org/10.1109/TIT.1983.1056648},
  doi          = {10.1109/TIT.1983.1056648},
  bibsource    = {dblp computer science bibliography, https://dblp.org}
}

@article{DBLP:journals/jal/BermanJLSS90,
  author       = {Francine Berman and
                  David S. Johnson and
                  Frank Thomson Leighton and
                  Peter W. Shor and
                  Larry Snyder},
  title        = {Generalized Planar Matching},
  journal      = {J. Algorithms},
  volume       = {11},
  number       = {2},
  pages        = {153--184},
  year         = {1990},
  url          = {https://doi.org/10.1016/0196-6774(90)90001-U},
  doi          = {10.1016/0196-6774(90)90001-U},
  bibsource    = {dblp computer science bibliography, https://dblp.org}
}

@article{DBLP:journals/jacm/MulzerR08,
  author       = {Wolfgang Mulzer and
                  G{\"{u}}nter Rote},
  title        = {Minimum-weight triangulation is NP-hard},
  journal      = {J. {ACM}},
  volume       = {55},
  number       = {2},
  pages        = {11:1--11:29},
  year         = {2008},
  url          = {https://doi.org/10.1145/1346330.1346336},
  doi          = {10.1145/1346330.1346336},
  bibsource    = {dblp computer science bibliography, https://dblp.org}
}

@article{DBLP:journals/ijcga/FeketeM01,
  author       = {S{\'{a}}ndor P. Fekete and
                  Joseph S. B. Mitchell},
  title        = {Terrain Decomposition and Layered Manufacturing},
  journal      = {Int. J. Comput. Geom. Appl.},
  volume       = {11},
  number       = {6},
  pages        = {647--668},
  year         = {2001},
  url          = {https://doi.org/10.1142/S0218195901000687},
  doi          = {10.1142/S0218195901000687},
  bibsource    = {dblp computer science bibliography, https://dblp.org}
}

@article{DBLP:journals/algorithmica/ArkinHS00,
  author       = {Esther M. Arkin and
                  Martin Held and
                  Christopher L. Smith},
  title        = {Optimization Problems Related to Zigzag Pocket Machining},
  journal      = {Algorithmica},
  volume       = {26},
  number       = {2},
  pages        = {197--236},
  year         = {2000},
  url          = {https://doi.org/10.1007/s004539910010},
  doi          = {10.1007/S004539910010},
  bibsource    = {dblp computer science bibliography, https://dblp.org}
}

@article{DBLP:journals/comgeo/ArkinFM00,
  author       = {Esther M. Arkin and
                  S{\'{a}}ndor P. Fekete and
                  Joseph S. B. Mitchell},
  title        = {Approximation algorithms for lawn mowing and milling},
  journal      = {Comput. Geom.},
  volume       = {17},
  number       = {1-2},
  pages        = {25--50},
  year         = {2000},
  url          = {https://doi.org/10.1016/S0925-7721(00)00015-8},
  doi          = {10.1016/S0925-7721(00)00015-8},
  bibsource    = {dblp computer science bibliography, https://dblp.org}
}

@article{DBLP:journals/comgeo/FeketeM05,
  author       = {S{\'{a}}ndor P. Fekete and
                  Henk Meijer},
  title        = {The one-round Voronoi game replayed},
  journal      = {Comput. Geom.},
  volume       = {30},
  number       = {2},
  pages        = {81--94},
  year         = {2005},
  url          = {https://doi.org/10.1016/j.comgeo.2004.05.005},
  doi          = {10.1016/J.COMGEO.2004.05.005},
  bibsource    = {dblp computer science bibliography, https://dblp.org}
}

@article{DBLP:journals/corr/LubiwM17,
  author       = {Anna Lubiw and
                  Debajyoti Mondal},
  title        = {On compatible triangulations with a minimum number of Steiner points},
  journal      = {Theor. Comput. Sci.},
  volume       = {835},
  pages        = {97--107},
  year         = {2020},
  url          = {https://doi.org/10.1016/j.tcs.2020.06.014},
  doi          = {10.1016/J.TCS.2020.06.014},
  bibsource    = {dblp computer science bibliography, https://dblp.org}
}

@article{DBLP:journals/ior/FeketeMB05,
  author       = {S{\'{a}}ndor P. Fekete and
                  Joseph S. B. Mitchell and
                  Karin Beurer},
  title        = {On the Continuous Fermat-Weber Problem},
  journal      = {Oper. Res.},
  volume       = {53},
  number       = {1},
  pages        = {61--76},
  year         = {2005},
  url          = {https://doi.org/10.1287/opre.1040.0137},
  doi          = {10.1287/OPRE.1040.0137},
  bibsource    = {dblp computer science bibliography, https://dblp.org}
}

@inproceedings{DBLP:conf/icalp/Lingas82,
  author       = {Andrzej Lingas},
  title        = {The Power of Non-Rectilinear Holes},
  booktitle    = {9th International Colloquium on Automata, Languages, and Programming (ICALP 1982)},
  pages        = {369--383},
  year         = {1982},
  url          = {https://doi.org/10.1007/BFb0012784},
  doi          = {10.1007/BFB0012784},
  bibsource    = {dblp computer science bibliography, https://dblp.org}
}

@inproceedings{DBLP:conf/compgeom/DemaineEHILMOW04,
  author       = {Erik D. Demaine and
                  Jeff Erickson and
                  Ferran Hurtado and
                  John Iacono and
                  Stefan Langerman and
                  Henk Meijer and
                  Mark H. Overmars and
                  Sue Whitesides},
  title        = {Separating point sets in polygonal environments},
  booktitle    = {20th {ACM} Symposium on Computational Geometry (SoCG 2004)},
  pages        = {10--16},
  year         = {2004},
  url          = {https://doi.org/10.1145/997817.997822},
  doi          = {10.1145/997817.997822},
  bibsource    = {dblp computer science bibliography, https://dblp.org}
}

@inproceedings{DBLP:conf/cccg/KirkpatrickKP11,
  author       = {David G. Kirkpatrick and
                  Irina Kostitsyna and
                  Valentin Polishchuk},
  title        = {Hardness Results for Two-Dimensional Curvature-Constrained Motion
                  Planning},
  booktitle    = {23rd Annual Canadian Conference on Computational
                  Geometry (CCCG 2011)},
  year         = {2011},
  url          = {http://www.cccg.ca/proceedings/2011/papers/paper99.pdf},
  bibsource    = {dblp computer science bibliography, https://dblp.org}
}

@article{DBLP:journals/jacm/AbrahamsenAM22,
  author       = {Mikkel Abrahamsen and
                  Anna Adamaszek and
                  Tillmann Miltzow},
  title        = {The Art Gallery Problem is {\(\exists\)}{\(\mathbb{R}\)}-complete},
  journal      = {J. {ACM}},
  volume       = {69},
  number       = {1},
  pages        = {4:1--4:70},
  year         = {2022},
  url          = {https://doi.org/10.1145/3486220},
  doi          = {10.1145/3486220},
  bibsource    = {dblp computer science bibliography, https://dblp.org}
}

@inproceedings{DBLP:conf/focs/Abrahamsen21,
  author       = {Mikkel Abrahamsen},
  title        = {Covering Polygons is Even Harder},
  booktitle    = {62nd {IEEE} Annual Symposium on Foundations of Computer Science ({FOCS} 2021)},
  pages        = {375--386},
  year         = {2021},
  url          = {https://doi.org/10.1109/FOCS52979.2021.00045},
  doi          = {10.1109/FOCS52979.2021.00045},
  bibsource    = {dblp computer science bibliography, https://dblp.org}
}

@inproceedings{DBLP:journals/corr/abs-2311-10631,
  author       = {Mikkel Abrahamsen and
                  Joakim Blikstad and
                  Andr{\'{e}} Nusser and
                  Hanwen Zhang},
  title        = {Minimum Star Partitions of Simple Polygons in Polynomial Time},
  booktitle = {Symposium on Foundations of Computer Science (FOCS 2024)},
  year         = {2024},
  url          = {https://doi.org/10.48550/arXiv.2311.10631},
  doi          = {10.48550/ARXIV.2311.10631},
  bibsource    = {dblp computer science bibliography, https://dblp.org}
}

@incollection{chazelle1985optimal,
  title={Optimal convex decompositions},
  author={Chazelle, Bernard and Dobkin, David P.},
  series={Machine Intelligence and Pattern Recognition},
editor = {Godfried T. Toussaint},
  volume={2},
  pages={63--133},
volume = {2},
  year={1985},
booktitle = {Computational Geometry},
  addendum={Preliminary version at STOC 1979},
  doi={10.1016/B978-0-444-87806-9.50009-8}
}

@article{DBLP:journals/jacm/AsanoAI86,
  author    = {Takao Asano and
               Tetsuo Asano and
               Hiroshi Imai},
  title     = {Partitioning a polygonal region into trapezoids},
  journal   = {J. {ACM}},
  volume    = {33},
  number    = {2},
  pages     = {290--312},
  year      = {1986},
  url       = {https://doi.org/10.1145/5383.5387},
  doi       = {10.1145/5383.5387},
  bibsource = {dblp computer science bibliography, https://dblp.org},
  addendum = {Preliminary version af FOCS 1983.}
}

@article{DBLP:journals/siamcomp/ImaiA86,
  author       = {Hiroshi Imai and
                  Takao Asano},
  title        = {Efficient Algorithms for Geometric Graph Search Problems},
  journal      = {{SIAM} J. Comput.},
  volume       = {15},
  number       = {2},
  pages        = {478--494},
  year         = {1986},
  url          = {https://doi.org/10.1137/0215033},
  doi          = {10.1137/0215033},
  bibsource    = {dblp computer science bibliography, https://dblp.org}
}

@article{DBLP:journals/corr/abs-2211-01359,
  author       = {Mikkel Abrahamsen and
                  Nichlas Langhoff Rasmussen},
  title        = {Partitioning a Polygon Into Small Pieces},
  journal      = {CoRR},
  volume       = {abs/2211.01359},
  year         = {2022},
  url          = {https://doi.org/10.48550/arXiv.2211.01359},
  doi          = {10.48550/ARXIV.2211.01359},
  bibsource    = {dblp computer science bibliography, https://dblp.org}
}

@article{DBLP:journals/jpdc/LeungTWYC90,
  author       = {Joseph Y.{-}T. Leung and
                  Tommy W. Tam and
                  C. S. Wong and
                  Gilbert H. Young and
                  Francis Y. L. Chin},
  title        = {Packing Squares into a Square},
  journal      = {J. Parallel Distributed Comput.},
  volume       = {10},
  number       = {3},
  pages        = {271--275},
  year         = {1990},
  url          = {https://doi.org/10.1016/0743-7315(90)90019-L},
  doi          = {10.1016/0743-7315(90)90019-L},
  bibsource    = {dblp computer science bibliography, https://dblp.org}
}

@inproceedings{PackingSegments,
  title={Packing Segments in a Simple Polygon is {APX}-hard},
  author={Heuna Kim and Tillmann Miltzow},
  booktitle={European Conference on Computational Geometry (EuroCG 2015)},
  note = {\url{http://eurocg15.fri.uni-lj.si/pub/eurocg15-book-of-abstracts.pdf}},
  pages={24--27},
  year={2015}
}

@article{DBLP:journals/corr/abs-1008-1224,
  author       = {Erik D. Demaine and
                  S{\'{a}}ndor P. Fekete and
                  Robert J. Lang},
  title        = {Circle Packing for Origami Design Is Hard},
  journal      = {CoRR},
  volume       = {abs/1008.1224},
  year         = {2010},
  url          = {http://arxiv.org/abs/1008.1224},
  doi = {10.48550/arXiv.1008.1224},
  bibsource    = {dblp computer science bibliography, https://dblp.org}
}

@article{DBLP:journals/corr/abs-2004-07558,
  author       = {Mikkel Abrahamsen and
                  Tillmann Miltzow and
                  Nadja Seiferth},
  editor       = {Sandy Irani},
  title        = {Framework for ER-Completeness of Two-Dimensional Packing Problems},
  booktitle    = {61st {IEEE} Annual Symposium on Foundations of Computer Science, {FOCS}
                  2020, Durham, NC, USA, November 16-19, 2020},
  pages        = {1014--1021},
  publisher    = {{IEEE}},
  year         = {2020},
  url          = {https://doi.org/10.1109/FOCS46700.2020.00098},
  doi          = {10.1109/FOCS46700.2020.00098},
  bibsource    = {dblp computer science bibliography, https://dblp.org}
}

@article{PilzLayeredPlanar3SAT,
  author       = {Alexander Pilz},
  title        = {Planar 3-SAT with a Clause/Variable Cycle},
  journal      = {Discret. Math. Theor. Comput. Sci.},
  volume       = {21},
  number       = {3},
  year         = {2019},
  url          = {https://doi.org/10.23638/DMTCS-21-3-18},
  doi          = {10.23638/DMTCS-21-3-18},
  bibsource    = {dblp computer science bibliography, https://dblp.org}
}

@book{garey1979computers,
  title={Computers and intractability: {A} Guide to the Theory of {NP}-Completeness},
  author={Garey, Michael R. and Johnson, David S.},
  year={1979},
  publisher={W. H. Freeman \&\ Co.},
}

@inproceedings{DBLP:conf/approx/FeketeKKMS11,
  author       = {S{\'{a}}ndor P. Fekete and
                  Tom Kamphans and
                  Alexander Kr{\"{o}}ller and
                  Joseph S. B. Mitchell and
                  Christiane Schmidt},
  title        = {Exploring and Triangulating a Region by a Swarm of Robots},
  booktitle    = {Approximation, Randomization, and Combinatorial Optimization. Algorithms
                  and Techniques ({APPROX} 2011)},
  pages        = {206--217},
  year         = {2011},
  url          = {https://doi.org/10.1007/978-3-642-22935-0\_18},
  doi          = {10.1007/978-3-642-22935-0\_18},
  bibsource    = {dblp computer science bibliography, https://dblp.org}
}

@article{erdos1975packing,
  title={On packing squares with equal squares},
  author={Erd\H{o}s, Paul and Graham, Ron},
  journal={Journal of Combinatorial Theory, Series A},
  volume={19},
  number={1},
  pages={119--123},
  year={1975},
  publisher={Elsevier},
doi = {10.1016/0097-3165(75)90099-0}
}

@article{chung2019efficient,
  author    = {Fan Chung and
               Ron Graham},
  title     = {Efficient Packings of Unit Squares in a Large Square},
  journal   = {Discret. Comput. Geom.},
  volume    = {64},
  number    = {3},
  pages     = {690--699},
  year      = {2020},
  doi       = {10.1007/s00454-019-00088-9}
}

@Article{Gensane2005,
  author    = {Thierry Gensane and
               Philippe Ryckelynck},
  title     = {Improved Dense Packings of Congruent Squares in a Square},
  journal   = {Discret. Comput. Geom.},
  volume    = {34},
  number    = {1},
  pages     = {97--109},
  year      = {2005},
  doi       = {10.1007/s00454-004-1129-z},
  bibsource = {dblp computer science bibliography, https://dblp.org}
}

@article{DBLP:journals/ijrr/SoloveyH16,
  author       = {Kiril Solovey and
                  Dan Halperin},
  title        = {On the hardness of unlabeled multi-robot motion planning},
  journal      = {Int. J. Robotics Res.},
  volume       = {35},
  number       = {14},
  pages        = {1750--1759},
  year         = {2016},
  url          = {https://doi.org/10.1177/0278364916672311},
  doi          = {10.1177/0278364916672311},
  bibsource    = {dblp computer science bibliography, https://dblp.org}
}

@article{DBLP:journals/jacm/HochbaumM85,
  author       = {Dorit S. Hochbaum and
                  Wolfgang Maass},
  title        = {Approximation Schemes for Covering and Packing Problems in Image Processing
                  and {VLSI}},
  journal      = {J. {ACM}},
  volume       = {32},
  number       = {1},
  pages        = {130--136},
  year         = {1985},
  url          = {https://doi.org/10.1145/2455.214106},
  doi          = {10.1145/2455.214106},
  bibsource    = {dblp computer science bibliography, https://dblp.org}
}

@article{DBLP:journals/ipl/Chan04,
  author       = {Timothy M. Chan},
  title        = {A note on maximum independent sets in rectangle intersection graphs},
  journal      = {Inf. Process. Lett.},
  volume       = {89},
  number       = {1},
  pages        = {19--23},
  year         = {2004},
  url          = {https://doi.org/10.1016/j.ipl.2003.09.019},
  doi          = {10.1016/J.IPL.2003.09.019},
  bibsource    = {dblp computer science bibliography, https://dblp.org}
}

@article{DBLP:journals/comgeo/AgarwalKS98,
  author       = {Pankaj K. Agarwal and
                  Marc J. van Kreveld and
                  Subhash Suri},
  title        = {Label placement by maximum independent set in rectangles},
  journal      = {Comput. Geom.},
  volume       = {11},
  number       = {3-4},
  pages        = {209--218},
  year         = {1998},
  url          = {https://doi.org/10.1016/S0925-7721(98)00028-5},
  doi          = {10.1016/S0925-7721(98)00028-5},
  bibsource    = {dblp computer science bibliography, https://dblp.org}
}

\end{document}